\documentclass[12pt]{article}
\usepackage{amsmath}
\usepackage{amsthm}
\usepackage{amsfonts}
\usepackage{graphicx}
\usepackage{enumerate}
\usepackage{natbib}
\usepackage{url} 
\usepackage{color}
\usepackage{bm}
\RequirePackage[colorlinks,citecolor=blue,urlcolor=blue]{hyperref}
\usepackage{cleveref}
\usepackage{tikz, pgfplots}
\usetikzlibrary{positioning}
\usetikzlibrary{decorations.pathreplacing,calligraphy,patterns}
\usetikzlibrary{shapes,decorations,arrows,calc,arrows.meta,fit,positioning}
\tikzset{
	-Latex,auto,node distance =1 cm and 1 cm,semithick,
	state/.style ={ellipse, draw, minimum width = 0.7 cm},
	point/.style = {circle, draw, inner sep=0.04cm,fill,node contents={}},
	bidirected/.style={Latex-Latex,dashed},
	el/.style = {inner sep=2pt, align=left, sloped}
}

\newcommand{\blind}{1}

\begin{document}

\newtheorem{prop}{Proposition}
\newtheorem{assum}{Assumption}
\newtheorem{theorem}{Theorem}
\newtheorem{lemma}{Lemma}

\crefname{equation}{}{}
\crefname{theorem}{Theorem}{Theorems}
\crefname{lemma}{Lemma}{Lemmas}

\def\spacingset#1{\renewcommand{\baselinestretch}%
{#1}\small\normalsize} \spacingset{1}

\def\red{\color{red}}
\def\bmX{\bm{X}}
\def\bmZ{\bm{Z}}
\def\bmZrest{\bm{Z}_{-j}}
\def\bmz{\bm{z}}
\def\bmg{g}
\def\Ting{\color{red} Ting: }
\def\bbeta{{ \beta}}
\newcommand\E{\mathbb{E}}
\newcommand{\Ex}[1]{\mathbb{E}\!\left[#1\right]}
\def\cov{\mbox{Cov}}
\def\var{\mbox{Var}}
\def\red{\color{red}}
\def\zero{0}
\def\O{\bm O}
\def\bfalpha{\bm \alpha}
\def\bfbeta{ \beta}
\def\bfgamma{\bm \gamma}
\def\bfGamma{\bm \Gamma}
\def\bfeta{\bm \eta}
\def\bmx{\bm x}
\definecolor{dblue}{HTML}{0072B2}

\def\bmV{\bm{V}}
\def\bmQ{\bm Q}


\if1\blind
{
	 \begin{center} 
	\spacingset{1.5} 	{\LARGE\bf  GENIUS-MAWII: For Robust Mendelian Randomization with Many Weak Invalid Instruments} \\ \bigskip \bigskip
		\spacingset{1} 
		{\large  Ting Ye$ ^1 $, Zhonghua Liu$ ^2 $, Baoluo Sun$ ^3 $, and  Eric Tchetgen Tchetgen$ ^4 $} \\ \bigskip
	{  $ ^1 $Department of Biostatistics, University of Washington, Seattle, Washington, U.S.A. \\
	$ ^2 $Department of Biostatistics, Columbia University, New York City, New York, U.S.A. \\
	$ ^3 $Department of Statistics and Data Science, National University of Singapore, Singapore\\
	 $ ^4 $Department of Statistics and Data Science, The Wharton School, University of Pennsylvania, Philadelphia, Pennsylvania, U.S.A. \\
}
	\end{center}
} \fi

\if0\blind
{
  \bigskip
  \bigskip
  \bigskip
  \begin{center}\spacingset{1.5} 
    {\LARGE\bf GENIUS-MAWII: Identification and Estimation Using Heteroscedasticity in Mendelian Randomization with Many Weak Invalid Instruments}
\end{center}
  \medskip
} \fi

\bigskip
\begin{abstract}
Mendelian randomization (MR) addresses causal questions by using genetic variants as instrumental variables. We propose a new MR method, GENIUS-MAWII, which simultaneously addresses the two salient challenges in MR: many weak instruments and widespread horizontal pleiotropy. Similar to MR-GENIUS, we use heteroscedasticity of the exposure to identify the treatment effect. We derive influence functions of the treatment effect, and then we construct a continuous updating estimator and establish its asymptotic properties under a many weak invalid instruments asymptotic regime by developing novel semiparametric theory. We also provide a measure of weak identification, an overidentification test, and a graphical diagnostic tool. 
\end{abstract}

\noindent%
{\it Keywords:}  Causal inference; exclusion restriction; heteroscedastic errors; instrumental variables;
many weak moments; pleiotropy
\vfill

\newpage
\spacingset{1.5} 

\section{Introduction}
\label{sec: intro}

\vspace{-3mm}

\subsection{Challenges in Mendelian randomization}
Mendelian randomization (MR) is a method of using genetic variants -- typically single nucleotide polymorphisms (SNPs) -- as instrumental variables (IVs) to infer the causal effect of  a modifiable exposure on an outcome in the presence of unmeasured confounding \citep{Davey-Smith:2003aa, Smith:2004aa, Lawlor:2008aa, Davey-Smith:2014aa, Burgess:2015ab, Burgess:2015ad,  Zheng:2017aa}. As a powerful tool to disentangle causal relationship from complex environmental confounding, MR has become a popular method for establishing high-quality causal evidence based on observational data \citep{ Pingault:2018aa, Adam:2019aa}. 

For reliable causal inference using MR, genetic variants must be valid IVs that satisfy  three key assumptions \citep{Angrist2001, Baiocchi:2014aa, Hernan-Robins}: (i) (relevance) they are associated with the exposure; (ii) (independence) they are independent of any unmeasured confounder of the exposure-outcome relationship; (iii) (exclusion restriction) they affect the outcome exclusively through the exposure. The first assumption (relevance) is usually satisfied by selecting SNPs that are significantly associated with the exposure. A well-established challenge one is often faced with in MR, is the possibility that individual SNPs are only weakly associated with the exposure, resulting in weak IV bias \citep{Stock:2002aa, Burgess:2011aa, Burgess:2011ab} and extreme sensitivity to minor violations of the other two assumptions  \citep{Small:2008aa,  Wang:2018aa}. The second assumption (independence) is plausible within the framework of parent-offspring studies because of the random assortments of  genes from parents to offspring. The independence assumption also approximately holds in population data such as the UK Biobank as individuals share much common ancestry \citep{smith2020mendel}.   Among the three core IV assumptions, the exclusion restriction assumption is the most disputable, as emerging evidence has suggested that pleiotropy -- a phenomenon in which a genetic variant may affect multiple phenotypic traits  \citep{ Solovieff:2013aa, Verbanck:2018aa} -- is widespread. In fact, studies have identified hundreds of genetic variants from genome-wide association studies (GWASs) that are associated with multiple traits \citep{Sivakumaran:2011aa,  Parkes:2013aa, Gratten:2016aa, Pickrell:2016aa, Grassmann:2017aa, Webb:2017aa}. For example, a variant (rs2075650 in the APOE locus) is found to be significantly associated with several traits and diseases, including Body Mass Index (BMI), Alzheimer's disease, C-reactive protein,  high-density lipoprotein cholesterol, low-density lipoprotein cholesterol, plasma triglycerides, waist circumference, hip circumference and waist/hip ratio  \citep{Verbanck:2018aa}. Hence, using this variant to study the effect of BMI on systolic blood pressure (SBP) will likely violate the exclusion restriction assumption because the variant may affect SBP via other traits outside of the pathway of BMI. Failure to account for such \emph{horizontal pleiotropy} (i.e., SNPs having direct effects on the outcome) can lead to spurious findings.

\vspace{-3mm}
\subsection{Prior work}
In this article, we focus on the two salient challenges in MR: many weak IVs and widespread horizontal pleiotropy. These two challenges rarely act alone but rather interact with each other, because  weak IVs can amplify bias from pleiotropy  \citep{Small:2008aa}, and later in Sections \ref{sec: prelim}-\ref{sec: identification} we will see that the proposed method accounts for pleiotropy by exploiting heteroscedasticity and thus may be  more susceptible to weak IV bias than usual \citep{Lewbel:2012aa}. Therefore, it is important to address these two challenges simultaneously for desirable practical performance.

The issue of weak IV has been  extensively studied in econometrics \citep{Staiger:1997aa,  Chao:2005aa, Hansen:2008aa, Newey:2009aa, Stock:2002aa}. Typically, in linear models, an IV is considered weak if the first-stage F statistic is below 10 \citep{Stock:2002aa}.  
Recent papers by \cite{Zhao:2019aa, zhao2018statistical,  wang2019weakinstrument} and \cite{ye2021debiased} also develop methods that are robust to weak IVs in two-sample summary-data MR under an assumption of no systematic exclusion restriction violation. A common message from many of these works is that having many weak IVs can greatly circumvent  the difficulty from each individual IV being only weakly associated with the exposure and can improve estimation accuracy.

There has also been a rapidly growing development of statistical methods to address widespread horizontal pleiotropy, which mostly fall into the following two strands. The first strand of methods assumes that a certain proportion of candidate IVs are valid. For example,  \cite{Han:2008aa, Kang:2016aa, Bowden:2016aa} and \cite{Windmeijer:2019aa} propose methods that can recover the causal effect provided less than 50\% of IVs are invalid. \cite{Hartwig:2017aa, Guo:2018aa, guo2021post} and \cite{windmeijer2021confidence} develop methods based on the plurality rule,  assuming that the number of valid IVs is larger than any number of invalid IVs sharing the same ratio estimator limit. Other proposals in this first strand include  \cite{qi2019mendelian} and \cite{Verbanck:2018aa}. Clearly, none of these methods apply to the situation when pleiotropy is pervasive.  

The second strand of work allows for all the IVs to be pleiotropic but effectively restricts the effects of IVs on the exposure and outcome. Within the second strand,  it is also helpful to distinguish between two types of horizontal pleiotropy: uncorrelated pleiotropy, also known as the  instrument strength independent of direct effect (InSIDE) assumption \citep{Bowden:2015aa}, which says that the direct effects of the IVs on the outcome are uncorrelated with their effects on the exposure, and correlated pleiotropy, which says that the direct effects of the IVs on the outcome are correlated with their effects on the exposure. Of the two types, uncorrelated pleiotropy is easier to deal with, based on which  multiple methods have been developed, including \cite{Kolesar:2015aa} and \cite{Bowden:2015aa} for direct effects with  nonzero mean (directional horizontal pleiotropy), and  \cite{Zhao:2019aa, zhao2018statistical} and {\cite{ye2021debiased}} for direct effects with zero mean (balanced horizontal pleiotropy). Correlated  pleiotropy is more challenging.   \cite{morrison2020ng}  and {\cite{wang2021causal}}  allow a small proportion of genetic variants to exhibit correlated pleiotropy arising from one or several well-understood pleiotropic pathways.  \cite{Tchetgen2019_GENIUS} and \cite{sun2022selective} tackle this challenge from a different perspective; without assuming a certain structure underlying the correlated pleiotropy, the identification extends a novel strategy proposed in \cite{Lewbel:2012aa, Lewbel:2018aa} that exploits heteroscedasticity of the exposure variable. The details are reviewed in Section \ref{sec: prelim}. Other proposals in this second strand include  \cite{Burgess:2015ac}, \cite{Spiller:2019aa}  and {\cite{liu2020mendelian}}. 

\vspace{-3mm}
\subsection{Our contributions} 
\label{subsec: contribution}
In this work, we propose a new MR method, GENIUS-MAWII, that simultaneously addresses many weak IVs and widespread horizontal pleiotropy. We deal with widespread horizontal pleiotropy by leveraging heteroscedasticity of the exposure, and we account for   many weak IVs by establishing the consistency and asymptotic normality of  the continuous updating estimator (CUE) obtained from using the derived influence functions as moment conditions under many weak moment asymptotics. We also provide GENIUS-MAWII with a measure of weak identification, an overidentification test, and a graphical diagnostic tool.  We demonstrate in simulations and a real example using UK Biobank  the clear advantages of  GENIUS-MAWII in the presence of directional or correlated horizontal pleiotropy compared to other methods.

 Furthermore, our work makes important advances in the theory of generalized method of moments (GMM) involving unknown nuisance parameters under many weak moment conditions, which to our knowledge has not been studied in the literature.   This is a challenging task due to two main reasons. First, with the number of moment conditions growing to infinity, the number of nuisance parameters also grows to infinity. Second, the many weak moment asymptotics, which is well suited for MR studies with a large number of SNPs, is fundamentally different from the classical asymptotics (with a fixed number of ``strong'' moment conditions).  Importantly,  under the classical asymptotics,  it is well known that utilizing the influence function which belongs to the ortho-complement of the nuisance tangent space and estimating the nuisance parameters at a fast enough rate ensure us that the impact of  estimating the nuisance parameters is negligible; this is the key insight that drives many other successful applications of using the influence function to handle nuisance parameters    \citep{Newey1994, Ackerberg2014aa, ning2017, robins2017hoif, Chernozhukov2018ddml, bravo2020}.	To our surprise, we find that this appealing property does not hold under many weak moment asymptotics  in general, but still holds for GENIUS-MAWII because its moment conditions are  linear in the parameter of interest  (see Theorem \ref{theo: GMM}). In addition, our proof handles infinite-dimensional nuisance parameters.

The rest of the article proceeds as follows. In Section \ref{sec: prelim}, we introduce the invalid IV model and review the GENIUS identification strategy.  In Section \ref{sec: identification}, we derive  the class of influence functions which are shown to be multiply robust, and the efficient influence function. In Section \ref{sec: semi}, we consider estimation and inference of the treatment effect.  In Section \ref{sec: measure}, we provide a measure of weak identification, an overidentification test, and a graphical diagnostic tool. The article is concluded with simulations in Section \ref{sec: simu}, a real data application in Section \ref{sec: real}, and more discussion in Section \ref{sec: disc}. All technical proofs are in the supplementary materials. The \textsf{R} code for the proposed methods is publicly available at \texttt{https://github.com/tye27/mr.genius}.

\vspace{-5mm}
\section{Review of the GENIUS identification strategy} 
\label{sec: prelim}
Suppose that we observe an independent and identically distributed sample $ (\O_1, \dots, \O_n) $ with $ \O= ( \bmZ, \bmX, A, Y) $, where $\bmZ=(Z_1, \dots, Z_m)^T$ is a column vector including $m$ SNPs, each taking on values from the set $\{0,1,2\}$  which represents the number of minor alleles,  $ \bmX $ is a vector of observed covariates which can be empty when there are no observed covariates, $A$ and $ Y $ are continuous exposure and outcome variables. We emphasize that the $ m $ SNPs $ Z_1, \dots, Z_m $ are not required to be independent, i.e., we allow the $ m $ SNPs to be in linkage disequilibrium. We are interested in the causal effect of $ A $ on $ Y $, denoted by $ \beta_0 $, in the presence of unmeasured confounders $ U $.

When there are no observed covariates, we consider the following structural equations:
\begin{align}
	& E(Y\mid A,  U, \bmZ)=\beta_0A+ \alpha(\bmZ) +\xi_y(U), \label{eq: out model}\\
	& E(A\mid  U, \bmZ)= \gamma(\bmZ)+ \xi_a(U), \label{eq: exp model}
\end{align}
where $\alpha, \gamma,  \xi_y, \xi_a$ are unspecified functions, and $ \bmZ\perp U  $. In particular, $ \alpha(\bmZ) $ encodes the direct effect of $ \bmZ $ on $ Y $, and $\alpha(\bmZ)\neq 0 $ indicates that the exclusion restriction assumption is violated.  \cite{Lewbel:2012aa} also considers models (\ref{eq: out model})-(\ref{eq: exp model}).  \cite{Kolesar:2015aa} and \cite{Bowden:2015aa} consider the special case with $ \alpha(\bmZ)=\sum_{j=1}^{m} \alpha_j Z_j $ and $ \gamma(\bmZ)= \sum_{j=1}^{m} \gamma_j Z_j $, and assume that $ \alpha_j, \gamma_j, j=1, \dots, m $ are  random effects satisfying $ \alpha_j\perp \gamma_j $ (commonly referred to as the InSIDE assumption or uncorrelated pleiotropy), which is likely violated when there are SNPs affecting the exposure and outcome through common pathways \citep{morrison2020ng}. In contrast, we make no such restrictions.  Furthermore, as reviewed in Section \ref{sec: intro}, many existing MR methods, including \cite{Kang:2016aa, Bowden:2016aa, Hartwig:2017aa, Guo:2018aa, Windmeijer:2019aa} and \cite{guo2021post},  rely on the assumption that pleiotropy only sparsely involves  a small proportion of SNPs, whereas we allow every SNP to be pleiotropic.

Assume (\ref{eq: out model})-(\ref{eq: exp model}) and $ \bmZ\perp U $, it is shown in \cite{Tchetgen2019_GENIUS} that $ \beta_0 $ is the unique solution to 
\begin{align}
	E\{ (\bmZ- E(\bmZ))R_A (Y- \beta A)\}=0, \label{eq: continuous y}
\end{align}
provided that $E\{(\bmZ- E(\bmZ)) R_A A\}\neq \zero $, where $ R_A=A- E(A\mid \bmZ) $ is the conditionally centered exposure. Equation (\ref{eq: continuous y})  provides an identification formula for $ \beta_0 $ in the presence of unmeasured confounding by leveraging possibly invalid IVs. This identification strategy is named ``G-Estimation under No
	Interaction with Unmeasured Selection'' (GENIUS) in \cite{Tchetgen2019_GENIUS}.

We elaborate the key of identification in \eqref{eq: continuous y}. With $ \bmZ $ being potentially invalid IVs that have a direct effect on the outcome, the usual IV-based identification formula no longer holds because $ E\{ (\bmZ-E(\bmZ)) (Y-\beta_0 A)\}=E\{ (\bmZ-E(\bmZ))  \alpha(\bmZ)\} \neq \zero $.   In fact, when having a direct effect on the outcome,  the invalid IVs $ \bmZ $ are nothing more than observed confounders that are independent of $ U $. If the effect of $ \bmZ $ on the outcome is not modified by $ U $, as is the case under  \eqref{eq: out model}, then $  (\bmZ- E(\bmZ))  c(U) $ for any function $ c(\cdot) $ satisfying $ E(c(U)) = 0$ can be conceptualized as ``valid IVs'' satisfying $E\{  (\bmZ- E(\bmZ))  c(U)  (Y- \beta_0A) \} = 0$ because 
$  (\bmZ- E(\bmZ))  c(U)    $ are uncorrelated with any function of $ U $ and any function of $  \bm Z$, and do not have a direct effect on the outcome.  The conceptualized valid IVs $  (\bmZ- E(\bmZ))  c(U)  $  are infeasible as $ U $ is unobserved, but under \eqref{eq: exp model} a noisy version of which can be constructed as an additive interaction between conditionally centered exposure and centered IVs $ (\bmZ- E(\bmZ)) R_A$  and is used as the feasible ``valid IVs''.

There are three comments about the above intuition. First, the idea of using gene-environment interactions as valid IVs also appears in \cite{Spiller:2019aa}, but unlike \cite{Spiller:2019aa}, the  gene-environment interactions used in GENIUS  can be unobserved. Second, there are interesting tradeoffs between GENIUS and the two-stage least squares (2SLS), which is widely-used when $ \bmZ $ are valid IVs. On the one hand, when $ \bmZ $ are valid IVs, 2SLS imposes no assumption on the exposure model whereas GENIUS does. On the other hand, when $ \bmZ $ has a direct effect on the outcome, 2SLS fails while GENIUS can still identify the treatment effect of interest $ \beta_0 $. Moreover, with $ \bmZ $ having a direct effect on the outcome,  even if there is an interaction between $ \bmZ $ and $ U $ in the exposure model \eqref{eq: exp model}, its magnitude is usually small compared to the main effects of $ \bmZ $ and $ U $, then the bias of GENIUS is also relatively small. Finally, the key condition encoded by   \eqref{eq: out model}-\eqref{eq: exp model}, i.e., the effects of $ \bmZ $ on the exposure and outcome not being modified by $ U $ is stronger than needed and can be relaxed to some extent (see Section 3.1 of the supplementary materials).  We can also circumvent this restriction by collecting information about the part of $ U $ that interacts with SNPs and adjust for them as part of the observed covariates. This will be discussed further in Section \ref{sec: identification}.

When $ m=1 $ (i.e., one SNP),   $ \beta_0 $ identified via (\ref{eq: continuous y})  can be rewritten as a Wald ratio
\begin{align}
	\beta_0= \frac{E\{ (\bmZ-E(\bmZ)) R_A Y\}}{E\{ (\bmZ- E(\bmZ))R_A A\}}, \nonumber
\end{align}
where the numerator is the effect of $ (\bmZ-E(\bmZ)) R_A$ on $ Y $, the denominator is the effect of $ (\bmZ-E(\bmZ)) R_A$ on $ A $, and $ \beta_0 $ is simply the ratio. When $ m>1 $ (i.e., multiple SNPs), $ \beta_0 $ is over-identified. Moreover, identification using (\ref{eq: continuous y}) requires that  $ E\{ (\bmZ- E(\bmZ)) R_AA \} \neq \zero$, which is analogous to the relevance assumption in the IV literature, except here we conceptualize $ (\bmZ- E(\bmZ)) R_A  $ as the valid IVs. Specifically, since $ R_A= A- E(A\mid\bmZ) $, simple calculations reveal that $ E\{ (\bmZ- E(\bmZ)) R_AA \}=E\{ (\bmZ- E(\bmZ)) R_A^2\}= \cov (\bmZ, R_A^2)= \cov (\bmZ, \var(A\mid\bmZ))$, which means that identification using  (\ref{eq: continuous y}) requires $ A $ being heteroscedastic, i.e., $ \var(A\mid\bmZ) $ depends on  at least some $ \bmZ $.  We remark that the condition  $ E\{ (\bmZ- E(\bmZ)) R_AA \} \neq \zero$ can be empirically checked since $ \cov (\bmZ, R_A^2) $ can be estimated by the sample covariance between $ \bmZ $ and the squared residuals from fitting a linear regression of $ A $ on $ \bmZ $.  One can also apply tests for heteroscedasticity such as the tests in  \cite{KOENKER1981} and \cite{White:1980heterosce}. {Heteroscedasticity can be due to gene-environment interactions \citep{pare2010use}; see \cite{wang2019genotype} and \cite{sulc2020quantification} for some recent discoveries.
}


\vspace{-5mm}

\section{Semiparametric theory}
\label{sec: identification}
The identification result in Section \ref{sec: prelim} can be easily extended when there is an observed covariate vector $ \bmX $. Consider the following structural equations:
\begin{align}
	& E(Y\mid A,  U, \bmZ, \bmX)=\beta_0A+ \alpha(\bmZ, \bmX)   +\xi_y(U,\bmX),\label{eq: out model X}\\
	& E(A\mid U, \bmZ, \bmX)= \gamma(\bmZ, \bmX)+ \xi_a(U, \bmX),\label{eq: exp model X}
\end{align}
where $ \alpha, \gamma,  \xi_y, \xi_a$ are unspecified functions and $  \bmZ\perp U \mid  \bmX\!$. Then, as shown in Section 3.2 of the  supplementary materials,  $ \beta_0  $ is the unique solution to
\begin{align}
	E\{ (\bmZ- E(\bmZ\mid\bmX)) R_A (Y- \beta A)\}=0, \label{eq: continuous y, X}
\end{align}
provided that $ E\{ (\bmZ- E(\bmZ\mid\bmX)) R_AA \} =E[\cov \{  \bmZ, \var (A\mid\bmZ, \bmX)\mid\bmX \}] \neq \zero$, 	where  $ R_A=A- E(A\mid\bmZ,\bmX) $ is the conditionally centered exposure. Hence,  identification by \eqref{eq: continuous y, X} requires that $ \var(A\mid\bmZ, \bmX) $ depends on some $ \bmZ $. Note that the GENIUS identification strategy can be extended to binary exposure and/or binary outcome that follow semiparametric log-linear models; see Section 6 of the supplementary materials for details.

{Comparing structural equations (\ref{eq: out model X})-(\ref{eq: exp model X}) and $ \bmZ\perp U \mid\bmX $ with their unconditional counterparts, we see that to satisfy these assumptions, $\bm X$ should include covariates that  (i) are correlated with $\bm Z $; (ii)  modify the effect of $\bm Z$ on the outcome (which should be rare as $\bm Z$ should primarily influence the exposure); (iii) are confounders of the exposure-outcome relationship and modify the effect of $\bm Z$ on the exposure. Another interesting type of covariates is those that do not affect the outcome but modify the effect of  $\bm Z$ on the exposure. Adjusting for these covariates can weaken heteroscedasticity and thus weaken identification, and may even make the exposure effect unidentifiable if conditioning on all such covariates. However, identification of exposure effect can still be achieved if there is residual latent heterogeneity in the effect of $\bm Z$ on $A$ within all levels of $\bm X$. A diagram of how to choose $\bmX$ is in Section 1.1 of the supplementary materials.
}

We derive the class of influence functions and the efficient influence function \citep{Bickel:1993}  under the sole observed data restriction $  E(R_A (Y- \beta_0 A )\mid \bmZ,\bmX)=E(R_A (Y- \beta_0 A )\mid \bmX) $ implied by structural equations (\ref{eq: out model X})-(\ref{eq: exp model X}) and $ \bmZ\perp U\mid \bmX $. 
\begin{theorem} \label{theoX}
	(a) Under the conditional moment restriction $  E(R_A (Y- \beta_0 A )\mid \bmZ,\bmX)=E(R_A (Y- \beta_0 A )\mid \bmX) $,	 let $ h(\bm Z, \bm X) $ be any scalar-valued function, the class of influence functions of $ \beta_0 $ is 
	\begin{align}
	& \big\{h(\bmZ, \bmX)- E(h(\bmZ, \bmX)\mid\bmX)\big\} \{\Delta- E(\Delta\mid \bmX )\}, \label{eq: general IF}
	\end{align}
	where $ \Delta = R_A R_Y- \beta R_A^2 $,  $ R_A= A- E(A\mid \bmZ, \bmX), $ and  $ R_Y= Y- E(Y\mid \bmZ, \bmX) $. \\
	(b) The efficient influence function of $ \beta_0 $ is  obtained with $ 	h(\bmZ, \bmX) = C (\bmX)^T\bar{\bmZ},  $ where 
	\begin{align*}
		C (\bmX)&=\left\{  E\left[ (\bar\bmZ-E(\bar\bmZ\mid \bmX))(\bar\bmZ-E(\bar\bmZ\mid \bmX))^T(\Delta- E(\Delta\mid \bmX))^2  \mid \bmX\right]\right\}^{-1} \\
		&\qquad \qquad \qquad \qquad \qquad \qquad E\left\{ (\bar\bmZ-E(\bar\bmZ\mid \bmX)) (R_A^2- E(R_A^2\mid \bmX))\mid \bmX\right\},
	\end{align*}
	and 
  $ \bar{\bmZ}  $ is a column vector of all the dummy variables for the joint levels defined by $ \bmZ $.
\end{theorem}
The proof is given in the supplementary materials. Theorem \ref{theoX} includes the results without observed covariates as a special case by setting $\bmX$ to be empty.
 
Identification using the influence function (\ref{eq: general IF}) is in fact multiply robust. As shown in the supplementary materials,  the influence function (\ref{eq: general IF}) evaluated at $ \beta= \beta_0 $ has expectation zero  when  either one of the following three sets of the models is correctly specified: $ \{E(A\mid\bmZ, \bmX),  E(h(\bmZ, \bmX)\mid\bmX)\}$, $ \{E(Y\mid\bmZ, \bmX),  E(h(\bmZ, \bmX)\mid\bmX)\}$,  or $ \{E(A\mid\bmZ, \bmX),  E( R_A ( R_Y- \beta R_A) \mid \bmX)\}$. Therefore, in classical settings, multiply robust estimation and  inference about $ \beta_0 $ is straightforward via the classical GMM results \citep{Hansen:1982aa}.

{In principle, we can also leverage the scalar-valued influence function \eqref{eq: general IF} in Theorem \ref{theoX} as a moment condition under classical asymptotics. The optimal combination of SNPs $C(\boldsymbol{X})^T\bar{\boldsymbol{Z}} $ can generally be estimated in a first stage to improve efficiency \citep{chamberlain1987asymptotic,newey1990efficient}. In practice, however, this can pose computational difficulties due to the need to estimate high-dimensional conditional covariance matrices, a challenge that is also raised in \cite{stephens2014locally}. In fact, a poorly estimated optimal index $C(\boldsymbol{X})^T\bar{\boldsymbol{Z}} $ may lead to the unintended consequence of efficiency loss relative to a fix but arbitrary choice of $h(\bm Z, \bm X)$  in finite samples.   On the other hand, using influence function \eqref{eq: general IF} with a predefined \( h(\bm Z, \bm X) \), such as a sum-score of the components of \( \bm Z \), may also be inefficient. 

For these reasons, we follow the GMM approach of \cite{Newey:2009aa} under many weak moment asymptotics  and consider a $m$-dimensional vector of moment conditions
\begin{align}
  g^{IF}(\bm O; \beta, \bm \eta_0) = \{ \bm Z - E( \bm Z \mid \bm X) \}  \{\Delta- E(\Delta\mid \bm X )\},   \label{eq: IF}
\end{align}
where  $ \Delta$ is defined in Theorem \ref{theoX}, $ \bm \eta  $ denotes the vector of nuisance parameters, and $ \bm \eta_0 $ its true value. With $ \beta_0  $ being over-identified by  \eqref{eq: IF}, it can be estimated using GMM methods, which is known to be semiparametric efficient in the absence of weak IV in models with a finite number of moment restrictions.  Moreover, as will be discussed in Section \ref{subsec: overidentification}, compared to using a scalar \( h(\bm Z, \bm X) \),  using a vector of moment conditions in \eqref{eq: IF} also offers the additional benefit of providing an overidentification test. 
}

\vspace{-5mm}

\section{Estimation and inference  with many invalid IVs}
\label{sec: semi}

We introduce some additional notations. Let 
\begin{align}
	& \bmg_i(\bbeta, \bfeta)= \bmg^{IF}(\O_i; \bbeta, \bfeta) ,  \qquad \hat{\bmg}(\bbeta, \bfeta)=\frac{1}{n} \sum_{i=1}^{n} \bmg_i (\bbeta, \bfeta), \qquad \bmg_i=\bmg_i(\bbeta_0, \bfeta_0), \nonumber\\
	& \hat{\Omega}(\bbeta, \bfeta)=  \frac{1}{n} \sum_{i=1}^{n} \bmg_i(\bbeta,\bfeta) \bmg_i(\bbeta, \bfeta)^T, \qquad \Omega(\bbeta, \bfeta) = E[ \bmg_i(\bbeta, \bfeta) \bmg_i(\bbeta, \bfeta)^T] , \qquad \Omega= \Omega(\bbeta_0, \bfeta_0),\nonumber\\
	& G_i( \bfeta)= \frac{\partial \bmg_i(\bbeta, \bfeta)}{\partial \bbeta},\quad  \hat G(\bfeta)= \frac1n \sum_{i=1}^{n}G_i(\bfeta), \quad G_i= G_i( \bfeta_0), \quad G( \bfeta)= E\left[ G_i( \bfeta)\right], \quad G= G(\bfeta_0). \nonumber 
\end{align}
Note that $ G_i(\bfeta) $ does not depend on $ \beta $ because $ g_i(\beta, \bfeta) $ is linear in $ \beta $. This largely simplifies the problem. 

As always, asymptotic theory is useful if it provides a good approximation to finite-sample performance in applications. In MR with a large number of SNPs while each individual SNP is only weakly related to the exposure,  many weak moment asymptotics is well-suited and provides an improved approximation to finite sample behavior of invalid IV robust  inference than the classical asymptotics with a fixed number of ``strong'' moment conditions as $ n $ goes to infinity (see simulations in Section \ref{sec: simu}). Now we are ready to give the formal characterization of the many weak moment asymptotics.

\begin{assum}[many weak moment asymptotics] \label{assump: many weak moments}
	There are scalars $ \! \mu_n^2, c, \! c'>0$ such that 
	$$ 
	\mu_n^{2} c\leq  n G^T \Omega^{-1}  G \leq 	\mu_n^{2} c'.
	$$
	Then, $\mu_n^2\rightarrow \infty $ as $n\to \infty$ and $ m/\mu_n^2 $ is bounded for all $n$. 
\end{assum}
Assumption \ref{assump: many weak moments} provides an improved approximation when the many moment conditions are weak.  When $ \mu_n=\! \sqrt{n} $ and $m$ is finite, it agrees with the classical asymptotics (a finite number of ``strong'' moment conditions). More discussion on $ \mu_n $ is in Section 1.3 of the supplementary materials.

The many weak moment asymptotics is fundamentally different from the classical asymptotics.   Analogous to the weak IV bias arising from linear models, it has also been recognized that many weak moment conditions can make the usual GMM inference inaccurate \citep{Stock:2002aa}. For example,  \cite{Newey:2009aa} find that the two-step GMM is biased and has non-normal asymptotic distribution, while estimators in the generalized empirical likelihood (GEL) family \citep{Smith:1997aa, Parente:2014aa} are consistent and asymptotically normal but have larger asymptotic variance than usual.  Furthermore,  as outlined in Section \ref{subsec: contribution}, the many weak moment asymptotics poses several technical difficulties on dealing with nuisance parameters, which to our knowledge has not been addressed in the literature. 
For the rest of this section, we develop novel semiparametric theory to handle unknown nuisance functions under many weak moment conditions. These theoretical developments enable fast and stable estimation and inference about $\beta_0$.

For estimation purposes, we will assume linear SNP 
\citep{Sun2021}, exposure, and outcome models in Assumption \ref{assump: linear}. 
\begin{assum} (nuisance parameters)\label{assump: linear}
	Suppose that $ \bmZ, \bmX  $ are bounded, $ \bmX\in \mathbb{R}^{d_x} $ with $ d_x<\infty $,  
	\begin{align*}
		&E(Z_j \mid \bmX= \bmx) = \bmx^T \bm \pi_{j0}, ~ j=1,\dots, m, \\
		&E(A\mid \bmZ= \bmz, \bmX= \bmx) =(\bmx^T, \bmz^T) \bm \mu_0, \\
		&E(Y\mid \bmZ= \bmz, \bmX= \bmx) =(\bmx^T, \bmz^T) \bm \lambda_0, \\
		& E(R_AR_Y \mid \bmX= \bmx)=\omega_0(\bm x; \bm\mu_0, \bm\lambda_0)   , \\
		& E(R_A^2\mid \bmX= \bmx)= \theta_0(\bm x; \bm \mu_0 ) , 
	\end{align*}
	where $ \omega_0 $ and $  \theta_0 $ are unspecified functions, {$\bm \pi_{j0}, j=1,\dots, m, \bm \mu_0, \bm \lambda_0$ are unknown parameters},
 and the first component of  $ \bmX $ is  1 representing the intercept term. 
\end{assum}
Under Assumption \ref{assump: linear},  $ \bfeta= (\bm \pi_1^T,\dots, \bm \pi_m^T, \bm\mu^T,  \bm \lambda^T, \omega(\bm x; \bm\mu, \bm\lambda)  , \theta(\bm x; \bm \mu ) )^T$ collects all the nuisance parameters, $ \bfeta_0 $ is the true value of $ \bfeta $.  Write the estimator of $ \bm\bfeta_0 $ as $ \hat{\bm\bfeta} $, which includes the least squares estimators  $\hat{\bm \pi}_{1}, \dots, \hat{\bm \pi}_{m}, \hat{\bm \mu}, \hat{\bm \lambda}$ {from fitting the linear models in Assumption \ref{assump: linear},}  and the kernel estimators with plug-in estimated parameters  $ \hat \omega(\bm x; \hat{\bm\mu}, \hat{\bm\lambda}), \hat \theta(\bm x; \hat{\bm \mu})  $; {see Section 1.4 of the Supplement for the details of kernel estimators.} We choose the more flexible kernel estimators for  $ \omega_0(\bm x; \bm\mu, \bm\lambda) $ and $\theta_0(\bm x; \bm \mu )$ to avoid modeling the second moment terms. Alternatively, one can assume that   $E(R_A R_Y\mid \bm X)$ and $E(R_A^2 \mid  \bm X)$ follow parametric models, for example,	linear models that include a full set of quadratic terms of $\bm X$ or saturated models when $\bmX$ consists of only categorical variables. Then all the nuisance parameter estimators $\hat \bfeta$ can be obtained from the least squares estimation. Either way,  {the estimated nuisance parameters $\hat \bfeta$ converge to their true values under the assumed conditions}.  A special case is  when not adjusting for covariates, i.e., $ \bmX $ only includes the intercept term, $\omega(\bm x; \bm\mu, \bm\lambda)   $ and $ \theta(\bm x; \bm \mu ) $  become two one-dimensional parameters, and their estimators degenerate to simple averages.

 We focus on the continuous updating estimator (CUE) -- a member of the  GEL family -- in this article, because its objective function has an explicit form and its empirical performance is similar to the other estimators in the GEL family \citep{Newey:2004aa}. We propose the following GENIUS estimator that leverages MAny Weak Invalid IVs (GENIUS-MAWII), which is obtained using the influence function  $ g^{IF} (\O; \bbeta, \bfeta) $ defined in \eqref{eq: IF} with a  plug-in nuisance parameter  estimator $ \hat{\bfeta} $ (defined above),
\begin{align}
	\hat{\bbeta}=\arg\min_{\beta\in B} \hat{Q} (\bbeta, \hat{\bfeta}), \qquad \hat{Q}(\bbeta, \hat{\bfeta})=\hat{\bmg}(\bbeta, \hat{\bfeta})^T \hat{\Omega} (\bbeta, \hat{\bfeta})^{-1} \hat{\bmg}(\bbeta, \hat{\bfeta})/2, \label{eq: cue}
\end{align}
where $B$ is a compact set of parameter values, {chosen to confidently encompass the true value $\beta_0$}. Note that the form of the CUE is similar to the familiar two-step GMM estimator, except that the objective function is simultaneously minimized over $ \bbeta $ in the optimal weighting matrix $ \hat{\Omega} (\bbeta, \hat \bfeta) $. This is key in eliminating the many weak moment bias of two-step GMM estimator \citep{Newey:2009aa}.

The following theorem establishes the asymptotic properties of $ \hat{\beta} $ defined in (\ref{eq: cue}).

\begin{theorem} \label{theo: GMM}
	Under structural equations (\ref{eq: out model X})-(\ref{eq: exp model X}) and $ \bmZ\perp U\mid \bmX $, Assumptions \ref{assump: many weak moments}-\ref{assump: linear} and regularity conditions {stated in Assumptions 3-6} in Section 4.1 of the supplementary materials, $ m^2/n\rightarrow 0 $ as $n\to \infty$, $ \hat{\beta} $ in (\ref{eq: cue}) is consistent, i.e., $ \hat{\beta}\xrightarrow{p} \beta_0 $ as $n\to \infty$. If additionally $ m^3/n\rightarrow0 $ holds as $n\to \infty$, then $ \hat\beta $ is asymptotically normal, i.e., as $ n\rightarrow\infty $, 
	\begin{align}
		\frac{\mu_n(\hat{\beta}-\beta_{0})}{ \sqrt{ \frac{1}{ n \mu_n^{-2}  G^T\Omega^{-1} G}+  \frac{ \mu_n^{-2}E[U_i^T\Omega^{-1} U_i]}{[ n \mu_n^{-2} G^T\Omega^{-1} G]^2}  }} \xrightarrow{d} N\left(0, 1\right), \label{eq: CUE normality}
	\end{align}
	where $U_i=G_i-G-\{\Omega^{-1} E(\bmg_i G_i^T)\} ^T \bmg_i $ is the population residual from least squares regression of $G_i-G$ on $\bmg_i$. 
\end{theorem}

The proof is given in the supplementary materials. Here, we outline the key steps. Taylor expansion of the first-order condition $ \partial  \hat{Q}(\bbeta, \hat \bfeta) /\partial \beta |_{\beta= \hat  \beta} = 0$ gives 
\[
0=n \mu_n^{-1} \frac{\partial \hat Q(\beta, \hat\bfeta)}{\partial \beta}\bigg|_{\beta=\beta_0} + n \mu_n^{-2}  \frac{\partial^2 \hat Q (\beta, \hat\bfeta)}{\partial \beta^2} \bigg|_{\beta=\bar\beta} \mu_n (\hat\beta- \beta_0),
\]  
where $\bar\beta$ is some value between $\beta_0$ and $\hat\beta$. We  prove that  $n\mu_n^{-1} \partial  \hat{Q}(\bbeta, \hat \bfeta) /\partial \beta |_{\beta= \beta_0}$ is asymptotically equivalent to the sum of  a usual GMM term $ n\mu_n^{-1} G^T \Omega^{-1} \hat{g} (\beta_0, \bfeta_0) $ and a U-statistic term $ (n\mu_n)^{-1} \sum_{i\neq j} {U}_i^T \Omega^{-1} g_j $ that is no longer negligible due to the many weak moment asymptotics, and both terms are mean zero. The asymptotic normality of  $n\mu_n^{-1} \partial  \hat{Q}(\bbeta, \hat \bfeta) /\partial \beta |_{\beta= \beta_0}$ then follows from 
 the U-statistic term being uncorrelated with the usual GMM term, the asymptotic variance of the usual GMM term being $ n \mu_n^{-2} G^T \Omega^{-1} G $, the asymptotic variance of the U-statistic term being $ \mu_n^{-2}  E[  U_i^T \Omega^{-1} U_i] $, and the central limit theorem. The result in \eqref{eq: CUE normality} follows from  showing
$ n \mu_n^{-2}   \partial^2 \hat Q (\beta, \hat\bfeta)/ \partial \beta^2 |_{\beta=\bar\beta}   =  n \mu_n^{-2}  G^T \Omega^{-1} G  + o_p(1) . $
 
Interestingly, the  outline above implies that $ \hat\beta $ defined in \eqref{eq: cue} is asymptotically equivalent to $ \arg\min_{\beta\in B} \hat{Q}(\bbeta, \bfeta_0) $, which means that estimation of $ \bfeta_0 $ does not affect the asymptotic distribution of the CUE, even with multiple complications arising from the many weak moment asymptotics, the number of nuisance parameters growing to infinity, and nonparametric kernel estimators that themselves involve estimated parameters. {This result relies on the estimated nuisance parameters converging to the truth.} But still, this is an unusual property that does not hold in general, but holds in the current setting  due to three factors: (i) the number of moment conditions $ m $ grows to infinity at a rate slower than $ n^{1/3} $; (ii) the use of influence functions as the moment conditions and the fast convergence rate of the estimated nuisance parameters which imply  that $ \sqrt{n} \|\hat g (\beta_0, \hat\bfeta) - \hat g (\beta_0, \bfeta_0)  \| = o_p(1)$; and (iii) the moment conditions $ g_i(\beta, \bfeta) $ being linear in $ \beta $ which implies that  $ \sqrt{n} \|\hat G ( \hat\bfeta) - \hat G ( \bfeta_0)  \| = o_p(1)$, where $ \| \cdot\| $ is the $ \ell_2 $ vector norm. Crucially, we need (iii) to make sure that the impact of estimating $ \bfeta_0 $ is negligible for the U-statistic term; in contrast, (iii) is not needed under the classical asymptotics because the U-statistic term is a higher order term that is negligible. 



In Theorem \ref{theo: GMM}, the number of SNPs $ m $ is required to grow slower than the sample size $ n $, which is more restrictive than the limited information maximum likelihood (LIML) estimator \citep{Chao:2005aa} where $ m $ can grow at the same rate as $n$ or even faster. The reason behind this difference is that LIML assumes homoscedasticity, while CUE makes no such assumption. Consequently, for consistency of the CUE, $ m^2/n\rightarrow 0$ as $n\to\infty$ seems necessary given the need to consistently estimate  the heteroscedastic weight matrix $ \Omega $ which has $ m^2 $ elements.

According to Theorem \ref{theo: GMM}, under the many weak moment asymptotics, as long as $ m^3/n$  is small, the CUE $ \hat\beta $ is consistent and asymptotically normal, and  the convergence rate is $ \mu_n $.  The asymptotic variance of $ \hat{\bbeta}$ consists of the limit of two terms:
\begin{align}
	\frac{1}{ n G^T\Omega^{-1} G} ~\mbox{ and }  ~ 
	\frac{ E(U_i^T\Omega^{-1} U_i)}{(n G^T\Omega^{-1} G)^2} \label{eq: a.var}. 
\end{align} 
The first term is the classical GMM variance, while the second term is the variance contribution due to the  variability of the moment derivative $ G_i $ which does not vanish under many weak moment asymptotics. Specifically, when $ m/\mu_n^2\rightarrow 0$ as $n\to\infty$ or $G_i $ is a constant,  the second term  is negligible compared to the first term, so that  (\ref{eq: CUE normality})  agrees with classical GMM theory; otherwise, the additional variance is not negligible and results in larger variance of the CUE. Interestingly, in this many weak moment asymptotic regime where identification is weak, the impact of  estimation of the weight matrix $ \Omega $ is small compared to estimation of $ G $ and does not appear in the variance formula.

In practice, we can estimate the asymptotic variance of $ \hat \beta $ using $ \hat{V}/n$, where
\begin{equation}
\begin{aligned}
 &\hat{V} =\hat{H}^{-1} \hat{D}^T \hat{\Omega}^{-1} \hat{D} \hat{H}^{-1}, \qquad \hat{H}= \partial^2 \hat{Q}({\bbeta},  \hat{\bfeta})/  \partial \bbeta^2|_{\bbeta= \hat{\bbeta}}, \qquad \hat{\Omega}=\hat \Omega (\hat{\bbeta},  \hat{\bfeta}),\\
 &\hat{D}=  \hat G( \hat{\bfeta})-  \left\{  \frac1n   \sum_{i=1}^n  G_i( \hat{\bfeta}) g_i(\hat \bbeta, \hat{\bfeta})^T \right\} \hat{\Omega}^{-1} \hat{g}(\hat \bbeta,  \hat{\bfeta}	).  \label{eq: var est}
\end{aligned}
\end{equation}
   Here, $ \hat{H} $ is an estimator of $ G^T \Omega^{-1} G $, the middle term $ \hat{D}^T \hat{\Omega}^{-1} \hat{D} $ is an estimator  of  the asymptotic variance of $\sqrt{n} \partial \hat{Q} (\bbeta, \hat \bfeta)/\partial \bbeta|_{\bbeta=\bbeta_0} $.  Notice that $ \hat{G}^T\hat{\Omega}^{-1}\hat{G} $ cannot be used in place of $ \hat{H} $ because $ \hat{G}^T\hat{\Omega}^{-1}\hat{G} $ is biased under many weak moment asymptotics \citep{Newey:2009aa}.  Based on the variance estimator, we can test the hypothesis $H_0: \beta=\beta^* $ using the Wald statistic $ T= \sqrt{n}(\hat{\beta}- \beta^*)/ \hat{V}^{1/2}$ or we can construct a $ 1-\alpha $ confidence interval based on normal approximation. Other identification robust statistics, e.g., Lagrange multiplier statistic, conditional likelihood ratio statistics can also be applied here, and they are in fact asymptotically equivalent to the Wald statistic under many weak moment asymptotics; see \cite{Newey:2009aa} for details.  It is worth noting that  under strong identification, Theorem \ref{theo: GMM} and the variance formula in \eqref{eq: var est}  are asymptotically equivalent to classical GMM counterparts,  and thus remain applicable. Hence, the results in this section are not only well-suited for MR analysis but also applicable to a wider range of regimes compared to classical GMM results.

			The exposure and outcome models in Assumption \ref{assump: linear} can be extended to include SNP-SNP and SNP-covariate interactions, by simply fitting linear models with those interaction terms included. Then, consistency in Theorem \ref{theo: GMM} continues to apply when the numbers of regressors in the exposure and outcome models go to infinity slower than $n^{1/2}$, and the asymptotic normality result in Theorem \ref{theo: GMM} continues to apply when the numbers of regressors in the exposure and outcome models go to infinity slower than $n^{1/3}$.

Finally, as discussed after Theorem \ref{theoX}, estimators based on influence functions enjoy a multiple robustness property in classical settings with a finite number of ``strong'' moment conditions, and    $\hat\beta$ achieves the semiparametric efficiency bound \citep{Ackerberg2014aa}. However, multiple robustness and semiparametric efficiency under many weak moment asymptotics are more delicate and will be interesting to pursue in future work.


\vspace{-5mm}

\section{Measure of weak identification and diagnosis} 
\label{sec: measure}
We have developed a new MR method, GENIUS-MAWII, which simultaneously addresses the two salient challenges in MR:  many weak IVs and widespread horizontal pleiotropy. We account for many weak IVs by establishing consistency and asymptotic normality of the CUE obtained from using the derived influence functions as moment conditions under many weak moment asymptotics, and we deal with widespread horizontal pleiotropy by leveraging heteroscedasticity of the exposure based on the GENIUS identification strategy. However, GENIUS-MAWII will break down if identification is too weak for many weak moment asymptotics to provide good approximation, resulting in  \emph{estimation bias}; or if the untestable assumptions \eqref{eq: out model X}-\eqref{eq: exp model X} and $ \bmZ\perp U\mid \bmX $ do not hold, resulting in \emph{identification bias}. Therefore, to enhance the reliability of GENIUS-MAWII in MR analysis, besides using domain knowledge, it is useful to have tools to gauge whether identification is strong enough for the promised  asymptotic results to kick in and whether  there is any evidence to falsify the assumptions.  For these two purposes, we present a measure of weak identification, an overidentification test,  and a graphical diagnostic tool in this section. 

\vspace{-3mm}

\subsection{Measure of weak identification}
\label{subsec: measure of weak identification}
Detection of weak identification is an important part in IV analysis, because weak identification can result in unreliable estimation and inference \citep{Stock:2002aa}. Until now, several formal procedures are readily available for weak IV detection based on linear IV models. {For example, \cite{Stock2001:ivtest} propose to use the first-stage F-statistic to assert whether IVs are weak under the homoscedastic error assumption, which is later extended by \cite{olea2013robust} to handle heteroscedastic errors.}  \cite{Hahn:2002aa} develop a specification test for strong IVs.   In this section, we provide a measure of weak identification for GENIUS-MAWII, which measures the extent of heteroscedasticity and can serve as a helpful diagnosis for reliable inference.

{Write the moment equations in \eqref{eq: IF} as 
$$
  g^{IF}(\O; \beta, \bfeta_0) = \{ \bmZ - E( \bm Z \mid \bmX) \}  \left[ R_A R_Y - E(R_A R_Y \mid \bmX) - \beta \{R_A^2 -E(R_A^2 \mid \bmX)   \}  \right],
$$ 
which can be viewed as using $ \bmZ - E( \bm Z \mid \bmX)  $  as the standard IV, $R_A R_Y - E(R_A R_Y \mid \bmX)$ as the derived outcome and  $R_A^2 - E(R_A^2 \mid \bmX)$ as the exposure. {As suggested by an anonymous
reviewer,} we use the heteroscedasticity-robust F-statistic in the regression of $R_A^2 - E(R_A^2 \mid \bmX)$ on $ \bmZ - E( \bm Z \mid \bmX)  $ as a measure of weak identification. This F-statistic can also be thought of as a Koenker test for heteroscedasticity \citep{KOENKER1981}, specifically for whether $\var (R_A \mid \bmZ , \bmX )$ depends on $\bmZ$. {In practice, we replace the unknown quantities by their estimators and denote the heteroscedasticity-robust F-statistic in the regression of $ R_A (\hat{\bm \mu})^2  - \hat \theta (\bm X; \hat{\bm\mu})$ on $ \bmZ - (\bmX^T \hat{\bm\pi}_{1}, \dots, \bmX^T \hat{\bm\pi}_{m})^T  $ as $F_{\rm GENIUS}$, where $R_A (\hat{\bm \mu})= A - (\bmX^T, \bmZ^T)\hat{\bm \mu}, $ and $\hat{\bm\pi}_j's$ and $\hat \theta (\bm X; \hat{\bm\mu})$ are defined after Assumption \ref{assump: linear}.  {From simulation studies in Section \ref{subsec: simu2}, we recommend check to make sure $F_{\rm GENIUS}$ is larger than 2. However, a rigorous theoretical evaluation of the weak identification test for CUE will be conducted in future work. 
}}

\vspace{-3mm}

\subsection{Overidentification test and graphical diagnosis}
\label{subsec: overidentification}

{In the GMM literature, it is common to perform overidentification tests to test whether $E[g_i(\beta_0, \bfeta_0)]=0$ holds \citep{Hansen:1982aa}. A popular statistic is simply a scaled minimized CUE objective function	$ 2n  \hat{Q} (\hat \bbeta, \hat{\bfeta}) $, which is often called the $J$-statistic.

	\begin{theorem} \label{theo: overidentification}
		Under the same conditions in Theorem \ref{theo: GMM} and $m^3/n\to 0$  as $n\to\infty$, when the null hypothesis $H_0: E[g_i(\beta_0, \bfeta_0)]=0$ holds, 	 
		$$ 
		P\left( 2n  \hat{Q} (\hat \bbeta, \hat{\bfeta}) \geq \chi_{1-\alpha}^2 (m-1) \right)  \rightarrow \alpha
		$$ 
		 as $n\to\infty$, where $\chi_{1-\alpha}^2 (m-1)$ is the $(1-\alpha)$-quantile of the $\chi^2 (m-1) $ distribution.  
	\end{theorem}
	Theorem \ref{theo: overidentification} shows that we can reject $H_0$ if $ 2n  \hat{Q} (\hat \bbeta, \hat{\bfeta}) \geq  \chi_{1-\alpha}^2 (m-1)$, which is the same as the overidentification test under the classical setting with a fixed number of ``strong'' moment equations.}  In Section 1.2 of the supplementary materials, we analytically show that the overidentification test has power to detect assumption violations in typical MR applications.

In addition to the overidentification test, another diagnosis approach is as follows. Note that  our method relies on untestable assumptions \eqref{eq: out model X}-\eqref{eq: exp model X} and $\bmZ \perp U\mid \bmX$. These assumptions imply the conditional moment restriction 
\begin{align*}
	E\big\{R_A (Y- \beta_0 A) - E(R_A (Y- \beta_0 A) \mid \bmX) \mid \bmZ, \bmX\big\} = 0  .
\end{align*}
We have used part of its implications to identify $ \beta_0 $  in \eqref{eq: IF}. But this conditional moment restriction has many other implications that we can use for falsification. This motivates  a graphical diagnostic tool which plots  the ``residual'' 
\[
\hat t_i = \hat R_{Ai}(Y_i - \hat\beta A_i)- \hat E\{R_{Ai} (Y_i - \hat\beta A_i)\mid \bmX_i\} 
\]
against $f(\bmZ_i, \bmX_i)$, where $\hat R_{Ai}= A_i- \hat E(A_i \mid \bmZ_i, \bmX_i )$, $R_{Ai}= A_i- E(A_i \mid \bmZ_i, \bmX_i )$,  and $f(\bmz, \bmx)$ is a pre-specified function that is non-linear in $ \bmz $ to avoid using duplicated information as that used for identification in  \eqref{eq: IF}. If modeling assumptions hold and all estimators have negligible bias, then $\hat t_i$ should be centered around zero across different values of  $f(\bmZ_i, \bmX_i)$; evidence that $\hat t_i$ is not centered around zero indicates violation of the assumptions. Note that unlike the typical residual plots for diagnostic in regression models, variance difference of $ \hat t_i $ across different values of $f(\bmZ_i, \bmX_i)$ does not violate our assumptions. As a final remark, the graphical diagnostic tool is also applicable to the situation where there is no observed covariates by setting $\bmX$ to be empty. In Section \ref{sec: real}, we set $f(\bmZ_i)=\{ \hat E(A_i \mid \bmZ_i)\}^2$.


\vspace{-5mm}

\section{Simulations}
\label{sec: simu}

\vspace{-3mm}
\subsection{A simulation when assumptions for GENIUS-MAWII hold} \label{subsec: simu1}

We conduct a simulation study to evaluate the finite-sample performance of GENIUS-MAWII when its assumptions hold, i.e., under models \eqref{eq: out model X}-\eqref{eq: exp model X} and $Z\perp U\mid \bmX$. Its performance is compared to the GENIUS estimators obtained from using ordinary CUE and two-step GMM. All three GENIUS estimators are based on the influence functions in  \eqref{eq: IF}, {with nuisance parameters estimated using the models in Assumption \ref{assump: linear}.}  GENIUS-GMM is the two-step GMM from the \textsf{gmm} package in \textsf{R}. {GENIUS-MAWII and GENIUS-CUE have the same point estimators defined by  \eqref{eq: cue} and are computed in the same way using the  \textsf{optimize} and  \textsf{uniroot} functions in \textsf{R}. Specifically, we first minimize the objective function  in  (9) using \textsf{optimize} function with the specified boundary (-10, 10), and if the returned value is very close to the specified boundary, we recompute by applying \textsf{uniroot} to the derivative of the objective function in  (9). The main purpose of this extra step is to stabilize the numerical optimization. The difference between GENIUS-MAWII and GENIUS-CUE is in the variance estimators: 
GENIUS-CUE uses the classical textbook variance estimator, while GENIUS-MAWII uses \eqref{eq: var est}.}

GENIUS-MAWII is also compared to four IV estimators: the two-stage least squares (2SLS), limited information maximum likelihood (LIML), {2SLS with the confidence interval selection method (CIIV-2SLS) \citep{windmeijer2021confidence}}, and five MR estimators: the inverse variance-weighted (IVW) estimator \citep{Burgess:2013aa}, robust adjusted profile score estimator (MR-raps) \citep{zhao2018statistical}, MR-Egger regression \citep{Bowden:2015aa}, {weighted median estimator (MR-median)} \citep{Bowden:2016aa}, and MR-mode \citep{Hartwig:2017aa}. {Additionally, we implemented GMM with CIIV and  2SLS with the two-stage hard thresholding (TSHT) method \citep{Guo:2018aa}. The results from these approaches are similar to those obtained with CIIV-2SLS, and hence are not reported.} We also note that the five MR methods are developed as two-sample MR methods, but we apply them to our one-sample setting regardless to see their performance; {see \cite{minelli2021use} for a comprehensive investigation of this practice.}  2SLS is implemented using the \textsf{AER} package, LIML using the \textsf{ivmodel} package, CIIV using the   \textsf{CIIV} package (using the heteroskedasticity-robust variance option and with first-stage thresholding for weak IVs), IVW  using the \textsf{mr.divw} package, MR-raps using the \textsf{mr.raps} package (with huber loss), and MR-Egger, MR-median, and MR-mode using the \textsf{MendelianRandomization} package. 

{We generate $m=100$ independent SNPs $Z_1, \dots, Z_m$ with $P(Z_j=0)=0.25, P(Z_j=1)=0.5$, $P(Z_j=2)=0.25$, $j=1,\dots, m$. We generate the exposure and outcome from 
\begin{align*}
A&= \sum_{j=1}^m  \gamma_j  Z_j +\eta_A  U+ \left( 1 + \sum_{j=1}^m  \delta_j  Z_j \right) \epsilon_A, \\ 
Y&= \beta_0 A+ \sum_{j=1}^m  \alpha_j Z_j + \eta_Y U+ \epsilon_Y,
\end{align*}
where $\eta_A=\eta_Y=1$, $\beta_0=0.4$, $U\sim N(0, 0.6(1- h^2))$, and  $\epsilon_A, \epsilon_Y \sim N(0, 0.4(1- h^2))$. 	Note that $\gamma_j = \varphi_{\gamma j} \sqrt{ h^2 /(1.5m)}$ and $\delta_j  =  \varphi_{\delta j} \kappa \sqrt{ h^2 /(1.5m)} $, where $m$ is the total number of SNPs, $\varphi_{\gamma j}, \varphi_{\delta j}  $'s are constants that are generated once from a standard normal distribution. Here,  $h^2$ can be interpreted as the proportion of variance in $A$ that is attributed to $E(A\mid \bmZ)$, and $\kappa$ controls the level of heteroscedasticity. We set $h^2=0.2$ and $\kappa=1$.

As illustrated in Figure \ref{fig:snps},  we consider three types of SNPs: $\mathcal{S}_1, \mathcal{S}_2, \mathcal{S}_3$ with proportion $p_1, p_2, p_3$, where $p_1+p_2+p_3=1$. Specifically, $\mathcal{S}_1$ consists of valid IVs with $\alpha_j=0$ (i.e., no direct effect on the outcome); $\mathcal{S}_2$ consists of invalid IVs with uncorrelated pleiotropic effects and $\alpha_j \sim_{i.i.d.}  N(\sqrt{\tau_0^2}, \tau_0^2) $ (i.e., InSIDE is satisfied), where $\tau_0^2 =   h^2 /(1.5m) $; $\mathcal{S}_3$ consists of invalid IVs that affect $A$ and $Y$ through a common factor $C$, which leads to correlated pleiotropic effects and $ \alpha_j = \gamma_j/2$. We consider four settings:
\begin{enumerate}
	\item (No invalid IVs) $p_1 = 1, p_2=p_3=0$;
	\item (40\% invalid IVs) $p_1 =0.6, p_2= 0.2, p_3=0.2$;
	\item  (90\% invalid IVs with InSIDE) $p_1=0.1, p_2=0.9, p_3=0$;
	\item (90\% invalid IVs without InSIDE) $p_1=0.1, p_2=0, p_3=0.9$.
\end{enumerate}
 We consider two values of the sample size: $n=10,000$ and $n=100,000$.} {\color{black}
The results with 1,000 Monte Carlo repetitions are in Table \ref{tb: sim1}, which  summarizes (i) the Monte Carlo  mean and Monte Carlo standard deviation (SD) of each estimator, (ii) average of standard errors (SEs), and (iii) coverage probability (CP) of 95\% confidence intervals from normal approximation.
}

\begin{figure}
	\centering
	\resizebox{!}{!}{
		\begin{tikzpicture}
			\node[color=dblue] (1) {\footnotesize $\{Z_j\}_{j \in \mathcal{S}_2}$};
			\node[color=dblue, below=of 1,yshift=1.5mm] (6) {\footnotesize $\{Z_j\}_{j \in \mathcal{S}_3}$};
			\node[color=dblue, above=of 1,yshift=-2mm] (7) {\footnotesize $\{Z_j\}_{j \in \mathcal{S}_1}$};
			\node[] (4) [right= of 1 ] {\footnotesize $A$};
			\node[below= of 4] (2) {\footnotesize $C$};
			\node[] (3) [right =of 4,xshift=1cm] {\footnotesize $Y$};
			\path (2) edge node[above,xshift=-3mm, yshift=-2mm] {\footnotesize 0.2} (4);
			\node[gray, align=left] (5) [right =of 2, xshift=.6cm] {\footnotesize $U$}; 
			\path[gray] (5) edge node[el,above] {} (2);
			\path[gray] (5) edge node[el,above] {} (3);
			\path[gray] (5) edge node[el,above] {} (4);
			\path[color=dblue] (1) edge[bend left=30]  node[el,above] {} (3);
			\path[color=dblue] (6) edge node[el,above] {} (2);
			\path[color=dblue] (1) edge node[el,above] {} (4);
			\path[color=dblue] (7) edge node[el,above] {} (4);
			\path (2) edge node[above] {\footnotesize 0.1 } (3);
			\path[color=red] (4) edge node[el,above] {{\color{red}$ \beta_0 $}} (3);
	\end{tikzpicture}}	
	\caption{Illustration of three types of SNPs: $\mathcal{S}_1$ consists of valid IVs, $\mathcal{S}_2$ consists of invalid IVs with uncorrelated pleiotropic effects (i.e., InSIDE is satisfied), and $\mathcal{S}_3$ consists of invalid IVs with correlated pleiotropic effects.}	 \label{fig:snps}
\end{figure}
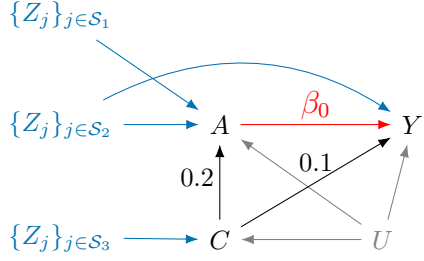

From Table \ref{tb: sim1}, the performance of GENIUS-MAWII is similar across Setting 1-4.  When $n=10,000$, the average F-statistic 
$F_{\rm GENIUS} =1.83$,   {GENIUS-MAWII shows nominal coverage probability, but has some attenuation bias and its SD is slightly more than $\sqrt{10}=3.16$  times larger than the SD with $n=100,000$. 
	This is because when the sample size is small, there are some outliers in some simulation runs due to instability of numerical optimization. Numerical optimization becomes more stable as the sample size becomes larger, essentially resolving this issue.}  When $n=100,000$, the average F-statistic $F_{\rm GENIUS} =9.48$ and  GENIUS-MAWII  shows negligible bias and nominal coverage, which  agrees with our theoretical assessment that GENIUS-MAWII  performs well when the identification is not too weak. Across all scenarios, the SEs calculated using \eqref{eq: var est} are close to Monte Carlo  SDs of GENIUS-MAWII.  Notice that with $m=100$ and $ n=10,000$, $m^2/n=1$ and $m^3/n=100$ are not small, but the GENIUS-MAWII estimator still performs quite well, indicating that our method is able to work well in typical MR studies with around 100 SNPs and 10,000-500,000 sample size. 

 Under all scenarios, the SEs underestimate the  Monte Carlo  SDs of the GENIUS-CUE estimator, which is also reflected by the fact that the CPs are below the nominal level 95\%. This is expected because according to our Theorem \ref{theo: GMM}, a higher order variance term is no longer negligible under many weak moment asymptotics. Comparing the GENIUS-CUE estimator and the GENIUS-MAWII estimator, we see that many weak moment asymptotics indeed provides a better finite sample approximation. 

Across all simulation scenarios, the GENIUS-GMM estimator has a larger bias than GENIUS-MAWII and GENIUS-CUE estimators, especially when $n=10,000$. Moreover,  it is not difficult to derive that the ordinary least squares (OLS) estimates obtained from regressing $Y$ on  intercept, $A$, and $Z_1,\dots, Z_m$ is approximately $\beta_0 + 1/(1+h^2) = 1.23$. Hence, we see that the GENIUS-GMM estimator is in fact biased towards the OLS. This is a GMM version of the well-known phenomenon that 2SLS is biased towards the OLS when IVs are weak \citep{Stock:2002aa}.

The 2SLS, LIML, IVW and MR-raps estimators are valid (i.e., their identification assumptions hold) under Setting 1. In Setting 1, when $n=10,000$, all estimators except LIML has some  weak IV bias; when $n=100,000$, all estimators are unbiased. Note that for IVW and MR-raps, their SEs over-estimate their Monte Carlo SDs because the SEs are developed for two-sample MR. Under Setting 2-4, the 2SLS, LIML, IVW and MR-raps estimators have large biases. 

{The CIIV-2SLS applies 2SLS after the confidence interval selection method and is developed under the plurality rule, which holds under Setting 1-3. In Setting 1, where all IVs are valid, CIIV-2SLS behaves similarly to 2SLS, showing some weak IV bias at a sample size of $10,000$, even when employing default first-stage thresholding. This bias is reduced at a larger sample size of $100,000$. In Setting 2, CIIV-2SLS demonstrates less bias than 2SLS, particularly at $n=100,000$, although it still exhibits a small weak IV bias, resulting in slight undercoverage. However, in Setting 3 where 90\% IVs are invalid, CIIV-2SLS shows considerable bias and large variance. We also observe that in Settings 2-3, SDs are much larger than SEs due to outliers in some simulation runs. In Setting 4, where the plurality rule does not hold, 
CIIV-2SLS shows a bias level comparable to that of 2SLS. 
}

The MR-median and MR-mode estimators are valid under Setting 2. In Setting 2, MR-median and MR-mode have some weak identification bias when $n=10,000$, and the bias becomes smaller when $n=100,000$. However, the SD and SE of MR-mode are very large compared to the other methods. Under Setting 3-4, the MR-median and MR-mode estimators have large biases. 

{The INSIDE assumption required by the MR-Egger estimator holds under Setting 3; however, MR-Egger exhibits bias in this setting, and its SD and SE are larger compared to GENIUS-MAWII. Notably, under Setting 3, the correlation between the error components in exposure and outcome is around 0.73, and the variability in instrument strength (measured by $I^2$ in \cite{minelli2021use}) is about 0.94  when $n=10,000$ and is about 0.99 when $n=100,000$. Hence, our observation that MR-Egger is biased when $n=10,000$ but the bias is reduced when $n=100,000$ aligns with the conclusions in \cite{minelli2021use}.  Under Settings 2 and 4, the MR-Egger estimator is also biased.}

In Section 2 of the supplementary materials, we conduct similar simulation studies where there is an observed covariate or when the SNPs are dependent (i.e., SNPs are in linkage disequilibrium).  Similar to the results presented in Table \ref{tb: sim1}, we generally find that the GENIUS-MAWII estimator has desirable performance with negligible bias and nominal coverage probability. 

\vspace{-3mm}
\subsection{A simulation under assumption violation and weak identification} \label{subsec: simu2}

We conduct more simulations for GENIUS-MAWII under  assumption violation and weak identification, and demonstrate the use of F-statistic and diagnostics tools to help identify situations where GENIUS-MAWII  can be reliably applied.  The setting is identical to that in Section \ref{sec: simu} with $m=100$  except that we generate the outcome from $	Y= \beta_0 A+ \sum_{j=1}^m  \alpha_j Z_j + (1+ \sum_{j=1}^{20}  \eta_{Yj}  Z_j ) U+ \epsilon_Y$ when \eqref{eq: out model X}  is violated, and we generate the exposure from  $ 	A= \sum_{j=1}^m  \gamma_j  Z_j + \left( 1 + \sum_{j=1}^{20}  \eta_{Aj}  Z_j \right)  U+ \left( 1 + \sum_{j=1}^m  \delta_j  Z_j \right) \epsilon_A$ when \eqref{eq: exp model X}  is violated, where
 $\eta_{Yj}  =  \varphi_{Y \eta j} \kappa \sqrt{ h^2 /(1.5m)}, \eta_{Aj}= \varphi_{A \eta j} \kappa \sqrt{ h^2 /(1.5m)} $, and $ \varphi_{Y \eta j}, \varphi_{A \eta j} $'s are constants that are generated once from a standard normal distribution. In other words, 20\% SNPs can have interactions with the unmeasured confounder, and the magnitude of which is similar to their interactions with $\epsilon_A$. 
 
 We consider four situations:  no model assumption is violated, only the exposure or the outcome model assumption is violated, and both model assumptions are violated. In each situation, we consider $\kappa=0, 0.1, 0.5, 1$ for increasing level of heteroscedasticity, and $n=10,000$,  $50,000$ and $ 100,000$. Note $\kappa=0$ means there is no model violation. The results are in Table \ref{tb: sim2}. 
 
From Table \ref{tb: sim2}, when $F_{\rm GENIUS}>2$, GENIUS-MAWII has negligible bias and nominal coverage probability when there is no model violation. When there is model misspecification and $F_{\rm GENIUS}>2$,  GENIUS-MAWII has nontrivial power (more than 50\% power) to detect model misspecification. Therefore, although assumption violations and/or weak heteroscedasticity can severely bias the GENIUS-MAWII estimator, the combined use of $F_{\rm GENIUS}$ and overidentification test can effectively identify those situations and provide guidance about when   GENIUS-MAWII can be reliably applied.

\begin{table}[ht] 	\spacingset{1.2}
\caption{Simulation results based on 1,000 Monte Carlo repetitions with $\beta_0=0.4$ and $m=100$;  $F_{\rm GENIUS}$ represents the average F-statistics for GENIUS-MAWII discussed in Section \ref{subsec: measure of weak identification}, SD is the Monte Carlo standard deviation,  SE is the average standard error, CP is the coverage probability of 95\% asymptotic confidence interval. For each setting, the methods of which the identification assumptions hold are in boldface. 
	\label{tb: sim1} } \vspace{2mm}
\centering 	\spacingset{1.2}
\resizebox{0.86\textwidth}{!}{\begin{tabular}{llcrrrrlrrrr} \hline 
           &              & \multicolumn{4}{c}{$ n=10,000, F_{\rm GENIUS}= 1.83$}         &  &             & \multicolumn{4}{c}{$n=100,000, F_{\rm GENIUS}= 9.48$} \\ \cline{3-6} \cline{9-12}
Setting    & Method       & Mean      & SD       & SE       & CP     &  &      & Mean   & SD     & SE     & CP    \\\hline
Setting 1      & \textbf{GENIUS-MAWII} & 0.367  & 0.124 & 0.126  & 97.5 &  &       & 0.398  & 0.029 & 0.031 & 96.4 \\
{\small(No invalid IVs)} & GENIUS-CUE   & 0.367   & 0.124 & 0.067  & 72.7 &  &   & 0.398  & 0.029 & 0.028 & 94.3 \\
           & GENIUS-GMM   & 0.779  & 0.054 & 0.045  & 0.0  &  &             & 0.470  & 0.025 & 0.026 & 24.6 \\
           & \textbf{2SLS}         & 0.439  & 0.023 & 0.023  & 60.3 &  &             & 0.404  & 0.008 & 0.008 & 91.6 \\
           & \textbf{LIML}         & 0.401  & 0.024 & 0.024  & 95.0 &  &             & 0.400  & 0.008 & 0.008 & 95.7 \\
           & \textbf{CIIV-2SLS}     & 0.442	& 0.024	&	0.024 & 	56.0 &&& 0.404	& 0.008 &	0.008 &	91.3  \\
           &\textbf{IVW}          & 0.439  & 0.023 & 0.032  & 84.9 &  &             & 0.404  & 0.008 & 0.010 & 99.3 \\
           & \textbf{MR-raps}      & 0.452  & 0.023 & 0.034  & 74.2 &  &             & 0.405  & 0.008 & 0.011 & 98.5 \\
           & \textbf{MR-Egger}     & 0.451  & 0.036 & 0.051  & 90.4 &  &             & 0.406  & 0.013 & 0.017 & 98.4 \\
           & \textbf{MR-median}    & 0.439  & 0.031 & 0.047  & 95.6 &  &             & 0.404  & 0.010 & 0.015 & 99.3 \\
           &\textbf{MR-mode}      & 0.438  & 0.302 & 5.004  & 99.6 &  &             & 0.402  & 0.060 & 0.443 & 99.8 \\ \hline 
Setting 2      & \textbf{GENIUS-MAWII} & 0.365  & 0.125 & 0.127  & 96.3  &  &       & 0.397  & 0.031 & 0.031 & 95.1 \\
{\small(40\% invalid IVs)} & GENIUS-CUE   & 0.365  & 0.125 & 0.067  & 73.1  &  &  & 0.397  & 0.031 & 0.028 & 93.7 \\
           & GENIUS-GMM   & 0.780  & 0.054 & 0.045  & 0.0   &  &             & 0.469  & 0.026 & 0.026 & 27.6 \\
           & 2SLS         & 0.568  & 0.055 & 0.021  & 2.0   &  &             & 0.540  & 0.054 & 0.007 & 0.6  \\
           & LIML         & 0.330  & 0.075 & 0.030  & 40.3  &  &             & 0.331  & 0.068 & 0.009 & 11.7 \\
           & \textbf{CIIV-2SLS}     & 0.498	&0.086	&	0.029	&31.4		&&&	0.408&	0.010	&0.010	&85.4\\
           & IVW          & 0.568  & 0.055 & 0.034  & 4.2   &  &             & 0.540  & 0.054 & 0.011 & 0.7  \\
           & MR-raps      & 0.614  & 0.061 & 0.054  & 4.6   &  &             & 0.549  & 0.056 & 0.046 & 12.0 \\
           & MR-Egger     & 0.582  & 0.095 & 0.090  & 46.8  &  &             & 0.540  & 0.089 & 0.087 & 62.9 \\
           & \textbf{MR-median}    & 0.535  & 0.081 & 0.060  & 46.2  &  &             & 0.445  & 0.061 & 0.023 & 69.8 \\
           & \textbf{MR-mode}      & 0.455  & 0.892 & 8.165  & 97.9  &  &             & 0.411  & 0.212 & 1.292 & 99.7 \\ \hline 
Setting 3      & \textbf{GENIUS-MAWII} & 0.369  & 0.122 & 0.126  & 96.9  &  &       & 0.397  & 0.031 & 0.031 & 94.4 \\
{\small(90\% invalid IVs with INSIDE)} & GENIUS-CUE   & 0.369  & 0.122  & 0.067  & 72.2  &  &  & 0.397  & 0.031 & 0.028 & 91.9 \\
           & {GENIUS-GMM}   & 0.780  & 0.052 & 0.045  & 0.0   &  &             & 0.469  & 0.027 & 0.026 & 25.9 \\
           & 2SLS         & 0.589  & 0.043 & 0.023  & 0.1   &  &             & 0.560  & 0.034 & 0.007 & 0.0  \\
           & LIML         & -0.555 & 0.110 & 0.076  & 0.0   &  &             & -0.553 & 0.057 & 0.024 & 0.0  \\
           & \textbf{CIIV-2SLS}     &0.031	&0.643	&	0.059&	0.0 &&&-1.559	&0.344	&0.090	&0.7\\
           & IVW          & 0.590  & 0.047 & 0.035  & 0.5   &  &             & 0.561  & 0.034 & 0.011 & 0.0  \\
           & MR-raps      & 0.628  & 0.053 & 0.119  & 55.9  &  &             & 0.564  & 0.036 & 0.114 & 94.0 \\
           & \textbf{MR-Egger}     & 0.604  & 0.105 & 0.169  & 87.5  &  &             & 0.539  & 0.077 & 0.172 & 99.3 \\
           & MR-median    & 0.789  & 0.097 & 0.079  & 4.0   &  &             & 0.789  & 0.104 & 0.037 & 5.4  \\
           & MR-mode      & 0.503  & 1.174 & 28.187 & 98.1  &  &             & 0.508  & 0.804 & 7.962 & 89.4 \\ \hline
Setting 4      & \textbf{GENIUS-MAWII} & 0.369  & 0.122 & 0.126  & 96.9  &  &      & 0.397  & 0.031 & 0.031 & 94.4 \\
{\small (90\% invalid IVs without INSIDE)} & GENIUS-CUE   & 0.369  &  0.122 & 0.067  & 72.2  &  & & 0.397  & 0.031 & 0.028 & 91.9 \\
           & {GENIUS-GMM}   & 0.780   & 0.052 & 0.045  & 0.0     &  &             & 0.469  & 0.027 & 0.026 & 25.9 \\
           & 2SLS         & 0.866  & 0.030  & 0.017  & 0.0     &  &             & 0.851  & 0.026 & 0.006 & 0.0    \\
           & LIML         & 0.839  & 0.038 & 0.018  & 0.0     &  &             & 0.837  & 0.032 & 0.006 & 0.0    \\
            & {CIIV-2SLS}     &0.900&	0.023	&	0.018	&0.0	&&&		0.899	&0.006	&0.006	&0.0 \\
           & IVW          & 0.866  & 0.030  & 0.038  & 0.0     &  &             & 0.851  & 0.026 & 0.012 & 0.0    \\
           & MR-raps      & 0.896  & 0.023 & 0.040   & 0.0     &  &             & 0.890   & 0.009 & 0.013 & 0.0    \\
           & MR-Egger     & 0.870   & 0.050  & 0.053  & 0.0     &  &             & 0.851  & 0.045 & 0.026 & 0.0    \\
           & MR-median    & 0.898  & 0.026 & 0.055  & 0.0     &  &             & 0.896  & 0.008 & 0.018 & 0.0    \\
           & MR-mode      & 0.965  & 1.888 & 4.592  & 43.2  &  &             & 0.898  & 0.040  & 0.370  & 11.1\\\hline 
\end{tabular}}
\end{table}

\begin{table}[ht]	\spacingset{1.2}
	\caption{Simulation results based on 1,000 Monte Carlo repetitions for GENIUS-MAWII with $\beta_0=0.4$ and $m=100$;  $F_{\rm GENIUS}$ represents the average F-statistics for GENIUS-MAWII discussed in Section \ref{subsec: measure of weak identification}, Power is the empirical power of the overidentification test in Section \ref{subsec: overidentification},		SD is the Monte Carlo standard deviation,  SE is the average standard error, CP is the coverage probability  of 95\% asymptotic confidence interval.
		\label{tb: sim2} } \vspace{2mm}
	\centering 	\spacingset{1.2}
	\resizebox{0.74\textwidth}{!}{
	\begin{tabular}{lrrrrrrrrr}\hline 
		Assumption   violation           & $\kappa$ & $n$     & $ F_{\rm GENIUS}$  & Power &  & Mean  & SD    & SE    & CP   \\\hline 
		No    model violation                            & 0     & 10,000 & 1.028  & 0.006    &  & 0.981 & 1.483 & 2.910 & 58.0 \\
		&       &  50,000  & 1.003  & 0.016    &  & 0.940 & 1.503 & 3.238 & 60.3 \\
		&       &  100,000  & 0.997  & 0.011    &  & 0.968 & 1.323 & 2.427 & 60.0 \\
		& 0.1   & 10,000 & 1.040  & 0.008    &  & 0.922 & 1.529 & 2.808 & 61.2 \\
		&       & 50,000 & 1.064  & 0.050    &  & 0.524 & 1.384 & 2.602 & 77.9 \\
		&       & 100,000 & 1.125  & 0.036    &  & 0.317 & 1.107 & 1.505 & 85.9 \\
		& 0.5   & 10,000 & 1.289  & 0.032    &  & 0.292 & 0.481 & 0.472 & 93.7 \\
		&       & 50,000 & 2.377  & 0.049    &  & 0.383 & 0.087 & 0.087 & 95.4 \\
		&       & 100,000 & 3.786  & 0.041    &  & 0.396 & 0.058 & 0.057 & 94.8 \\
		& 1     & 10,000 & 1.828  & 0.040    &  & 0.368 & 0.125 & 0.127 & 97.1 \\
		&       & 50,000 & 5.206  & 0.050    &  & 0.390 & 0.044 & 0.045 & 94.9 \\
		&       & 100,000 & 9.483  & 0.039    &  & 0.397 & 0.031 & 0.031 & 94.3 \\ \hline 
		Outcome model   violation & 0.1   & 10,000 & 1.040  & 0.011    &  & 0.925 & 1.497 & 2.663 & 61.5 \\
		&       & 50,000 & 1.064  & 0.056    &  & 0.497 & 1.514 & 2.799 & 79.7 \\
		&       & 100,000 & 1.125  & 0.069    &  & 0.121 & 1.272 & 2.086 & 88.3 \\
		& 0.5   & 10,000 & 1.289  & 0.075    &  & 0.045 & 0.758 & 0.707 & 97.9 \\
		&       & 50,000 & 2.377  & 0.558    &  & 0.188 & 0.126 & 0.125 & 69.0 \\
		&       & 100,000 & 3.786  & 0.940    &  & 0.203 & 0.086 & 0.082 & 28.5 \\
		& 1     & 10,000 & 1.828  & 0.284    &  & 0.130 & 0.203 & 0.202 & 91.6 \\
		&       & 50,000 & 5.206  & 0.999    &  & 0.154 & 0.071 & 0.073 & 3.2  \\
		&       & 100,000 & 9.483  & 1.000    &  & 0.162 & 0.051 & 0.050 & 0.0  \\ \hline 
		Exposure model violation  & 0.1   & 10,000 & 1.051  & 0.008    &  & 0.905 & 1.395 & 2.527 & 60.7 \\
		&       & 50,000 & 1.121  & 0.038    &  & 0.540 & 1.241 & 2.123 & 75.5 \\
		&       & 100,000 & 1.241  & 0.042    &  & 0.469 & 0.705 & 0.818 & 79.1 \\
		& 0.5   & 10,000 & 1.632  & 0.094    &  & 0.532 & 0.177 & 0.182 & 76.9 \\
		&       & 50,000 & 4.192  & 0.674    &  & 0.563 & 0.053 & 0.053 & 17.5 \\
		&       & 100,000 & 7.451  & 0.984    &  & 0.571 & 0.036 & 0.035 & 1.1  \\
		& 1     & 10,000 & 3.387  & 0.576    &  & 0.578 & 0.065 & 0.070 & 28.5 \\
		&       & 50,000 & 13.104 & 1.000    &  & 0.595 & 0.026 & 0.027 & 0.0  \\
		&       & 100,000 & 25.259 & 1.000    &  & 0.599 & 0.018 & 0.018 & 0.0  \\ \hline 
		Both model violation     & 0.1   & 10,000 & 1.051  & 0.012    &  & 0.845 & 1.635 & 3.337 & 61.8 \\
		&       & 50,000 & 1.121  & 0.055    &  & 0.454 & 1.280 & 2.195 & 79.5 \\
		&       & 100,000 & 1.241  & 0.065    &  & 0.274 & 0.927 & 1.134 & 87.3 \\
		& 0.5   & 10,000 & 1.632  & 0.160    &  & 0.436 & 0.215 & 0.217 & 90.6 \\
		&       & 50,000 & 4.192  & 0.955    &  & 0.472 & 0.065 & 0.065 & 74.4 \\
		&       & 100,000 & 7.451  & 1.000    &  & 0.481 & 0.045 & 0.043 & 50.5 \\
		& 1     & 10,000 & 3.387  & 0.838    &  & 0.507 & 0.081 & 0.087 & 70.8 \\
		&       & 50,000 & 13.104 & 1.000    &  & 0.523 & 0.033 & 0.034 & 4.9  \\
		&       & 100,000 & 25.259 & 1.000    &  & 0.529 & 0.023 & 0.023 & 0.2  \\\hline 
	\end{tabular}}
\end{table}

\vspace{-5mm}
\section{Application to the UK Biobank data}
\label{sec: real}

UK Biobank is a large-scale ongoing prospective cohort study with around 500,000 participants aged 40-69 at recruitment from 2006 to 2010. Participants provided biological samples, completed questionnaires, underwent assessments, and had nurse led interviews. Follow up is chiefly through cohort-wide linkages to National Health Service data, including electronic, coded death certificate, hospital, and primary care data \citep{sudlow2015uk}. Prevalent disease was coded using ICD-9 and ICD-10, and cause of death was coded using ICD-10. Genotyping was performed using two arrays, the Affymetrix UK BiLEVE (UK Biobank Lung Exome Variant Evaluation) Axiom array (about 50,000 participants) and Affymetrix UK Biobank Axiom array (about 450,000 participants). The SNPs included for analysis were directly genotyped or imputed using the Haplotype Reference Consortium panel. To reduce confounding bias due to population stratification, we restrict our analysis to people of genetically verified white British descent, as in previous studies \citep{Tyrrell2016}. For quality control, we exclude participants with (1) excess relatedness (more than 10 putative third-degree relatives), or (2) mismatched information on sex between genotyping and self-report, or (3) sex-chromosomes not XX or XY, or (4) poor-quality genotyping based on heterozygosity and missing rates $> 2\%$.

We are interested in estimating the causal effect of body mass index (BMI) on systolic blood pressure (SBP). We also exclude participants who are taking blood pressure medication based on self report. In total, the sample size for the final analysis is 292,757. We use $m=93$ SNPs that are associated with BMI at genome-wide significance level  \citep{Locke:2015aa}. 

We apply our method to the UK Biobank data and the results are summarized in Table \ref{tb:data}. {The implementation details are the same as in Section \ref{sec: simu} (including that $B=(-10,10)$ and without adjusting for any covariates).} {For comparison, we also include the unadjusted results from the ordinary least squares (OLS) analysis of SBP on BMI, which produces the largest point estimate among all methods, likely due to confounding bias.} From Table \ref{tb:data}, the F-statistic for the standard IV is $F_{\rm IV} = 59.7$.  The 2SLS, LIML, IVW, and MR-raps all have point estimates larger than that from GENIUS-MAWII and MR-Egger, which is likely due to failing to account for horizontal pleiotropic effects with nonzero mean. Moreover, the Sargan test \citep{sargan1958estimation} rejects the null hypothesis that all SNPs are valid IVs with a p-value of $< 2.22\times 10^{-16}$. 
Compared to MR-Egger, GENIUS-MAWII produces a similar point estimate but with a much higher precision. In particular, using GENIUS-MAWII, we find a significant positive effect of BMI on SBP ($\hat{\beta}= 0.140$, 95\% CI: [0.005, 0.275]).
This means a  one  $ kg/m^2  $ unit increase in BMI increases SBP by 0.140 $ mmHg $. Our analysis results can also be compared to other MR studies of BMI on SBP based on the UK Biobank data. For example,  \cite{Lyall:2017aa} finds a significant positive effect of BMI on SBP ($\hat{\beta}= 0.342$, 95\% CI: [0.161, 0.522]) using 2SLS and finds no significant effect using MR-Egger, which is consistent with the results of our analysis.

Finally, we run the tools developed in Section \ref{sec: measure} to assess the strength of identification and plausibility of the assumptions. 
The heteroscedasticity robust F-statistic for GENIUS is  {$F_{\rm GENIUS}= 11.1$}, large enough for application of GENIUS-MAWII. The overidentification test statistic is $2n  \hat{Q} (\hat \beta, \hat{\bm \eta})  = 108.3$, smaller than the critical value $\chi^2_{0.95}(92)=115.4 $. In addition, the diagnostic plot in Figure \ref{fig:bmi-sbp} shows that the blue line, which is the estimated conditional mean using the smoothing splines, {is close to a straight horizontal line through zero.} Therefore, both diagnosis
 approaches  find no evidence of assumption violation in this application.


\begin{figure}
    \centering
    \includegraphics[scale=0.6]{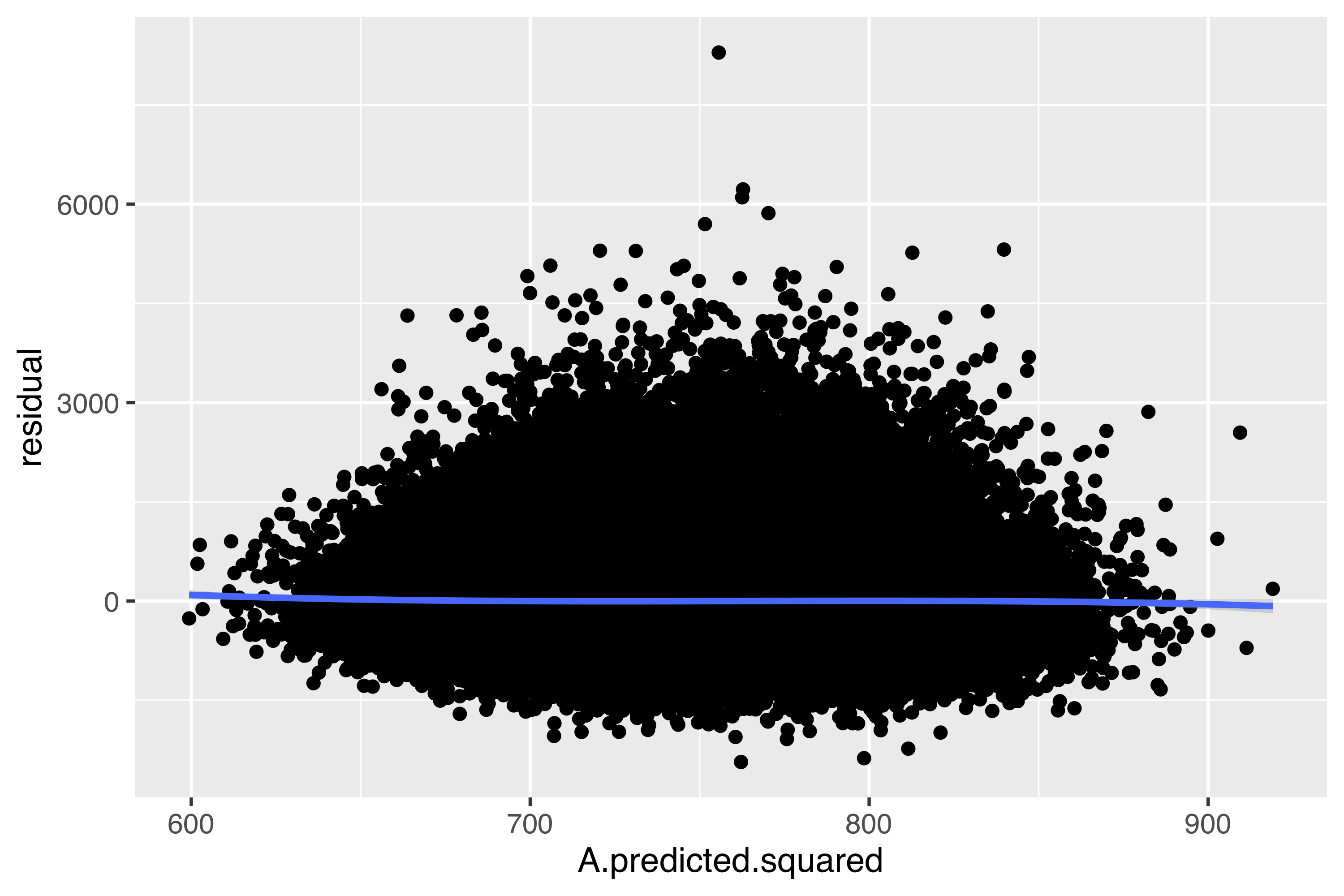}
    \caption{Residual plot for GENIUS-MAWII. The blue line is the estimated conditional mean using the smoothing splines, with gray point-wise confidence band (almost invisible in this plot). {We see that the blue line is close to a straight horizontal line through zero, indicating that the errors are centered at zero, so there is no evidence of assumption violation.}}
    \label{fig:bmi-sbp}
\end{figure}

\begin{table}[t]
\caption{Point estimates of exposure effect (Est) and their SEs from different MR methods and OLS in the BMI-SBP application using individual participant  data from UK Biobank (number of SNPs: 93, sample size: 292,757, $F_{\rm GENIUS}=11.1$, and $F_{\rm IV}=59.7$).} \label{tb:data} \vspace{2mm}
\centering
\resizebox{\textwidth}{!}{
\begin{tabular}{lccccccccccc} \hline
Method          & GENIUS-MAWII & GENIUS-CUE & GENIUS-GMM & 2SLS  & LIML  & IVW    & MR-raps & MR-Egger & OLS \\ \hline 
Est & 0.140     & 0.140& 0.175 & 0.321 & 0.277 & 0.338  & 0.482   & 0.175 &0.811   \\
SE  & 0.069      & 0.062& 0.062& 0.056 & 0.059 & 0.057  & 0.061   & 0.247 & 0.008  \\\hline
\end{tabular}}
\end{table}

\vspace{-5mm}

\section{Discussion}
\label{sec: disc}
In this paper, we have developed GENIUS-MAWII, a new method for Mendelian randomization (MR) which simultaneously addresses the two salient phenomena that adversely affect MR analyses: many weak IVs and widespread horizontal pleiotropy.  

We show via theory and simulations that GENIUS-MAWII can incorporate a large number of SNPs, allows for every SNP to be pleiotropic,  and is able to account for directional or correlated  horizontal pleiotropy.  These features make GENIUS-MAWII stand out with clear advantages over existing methods in the presence of directional or correlated horizontal pleiotropy.   In an application to the UK biobank data to study the effect of BMI on SBP,  GENIUS-MAWII produces a plausible effect size estimate ($ \hat\beta= 0.140 $, 95\% CI: [0.005, 0.275]), whereas 2SLS, LIML, IVW, and MR-raps produce larger effect size estimates ($ \hat\beta $ ranges from 0.277 to 0.482) that are likely due to failing to account for horizontal pleiotropic effects with nonzero mean. In addition, MR-Egger produces an effect size estimate ($ \hat\beta= 0.175 $) that is of similar magnitude to that from GENIUS-MAWII, but is much less precise and fails to yield statistical significance.

 {
GENIUS-MAWII leverages  heteroscedasticity to identify the causal effect, which can occur due to gene-environment interactions and is plausible for many situations \citep{pare2010use,wang2019genotype,sulc2020quantification}. However, if  the degree of heteroscedasticity is not very strong or the sample size is not large, estimation and inference may become challenging due to weak identification, and certain deviations away from the assumptions can generate large biases. Therefore, we recommend to perform the overidentification test and the graphical diagnosis to check for any evidence of assumption violation, and check to make sure the GENIUS F-statistic is larger than 2. These tests are very useful in determining whether GENIUS-MAWII can be applied reliably.
}

  Finally, we have developed novel semiparametric theory for handling unknown nuisance parameters under many weak moment conditions,  which to our knowledge has not been studied in the literature. Our theory addresses three main technical challenges: (i) the number of weak moment conditions grows to infinity with the sample size; (ii) the number of nuisance parameters grows to infinity with the sample size; and (iii) there exist infinite-dimensional nuisance parameters.  Our theoretical developments enable fast and stable estimation and inference about the causal effect of interest.

%



\vspace{-5mm}
\section*{Supplementary Materials}
The supplementary materials contain all technical proofs, identification results for binary exposure and/or binary outcome with the exponential link, and additional analytical and simulation results. 

\vspace{-5mm}
\section*{Acknowledgments}
The authors would like to thank Professor Dylan S. Small for constructive discussion and helpful feedback. We would also like to thank the anonymous referees, an Associate Editor and the Editor for their constructive comments that led to a much improved paper.

\spacingset{1.0}
\bibliographystyle{apalike}
\bibliography{reference_IV}

\clearpage

 \begin{center} 
	\spacingset{1.5} 	{\LARGE\bf  Supplement to ``GENIUS-MAWII: For Robust Mendelian Randomization with Many Weak Invalid Instruments''} \\ \bigskip \bigskip
	\spacingset{1} 
	{\large  Ting Ye$ ^1 $, Zhonghua Liu$ ^2 $, Baoluo Sun$ ^3 $, and  Eric Tchetgen Tchetgen$ ^4 $} \\ \bigskip
	{  $ ^1 $Department of Biostatistics, University of Washington, Seattle, Washington, U.S.A. \\
	$ ^2 $Department of Biostatistics, Columbia University, New York City, New York, U.S.A. \\
	$ ^3 $Department of Statistics and Data Science, National University of Singapore, Singapore\\
	 $ ^4 $Department of Statistics and Data Science, The Wharton School, University of Pennsylvania, Philadelphia, Pennsylvania, U.S.A. \\
}
\end{center}

\spacingset{1.5} 

\pagenumbering{arabic}
\setcounter{equation}{0}
\setcounter{table}{0}
\setcounter{section}{0}
\setcounter{lemma}{0}
\setcounter{assum}{2}
\setcounter{theorem}{2}
\renewcommand{\theequation}{S\arabic{equation}}
\renewcommand{\thetable}{S\arabic{table}}
\renewcommand{\thelemma}{S\arabic{lemma}}

\section{Additional analytical results}
\subsection{A diagram of how to choose $X$ for GENIUS-MAWII}
\label{subsec: discussion}

\begin{figure}[h]
	\centering
	\includegraphics[scale=.5]{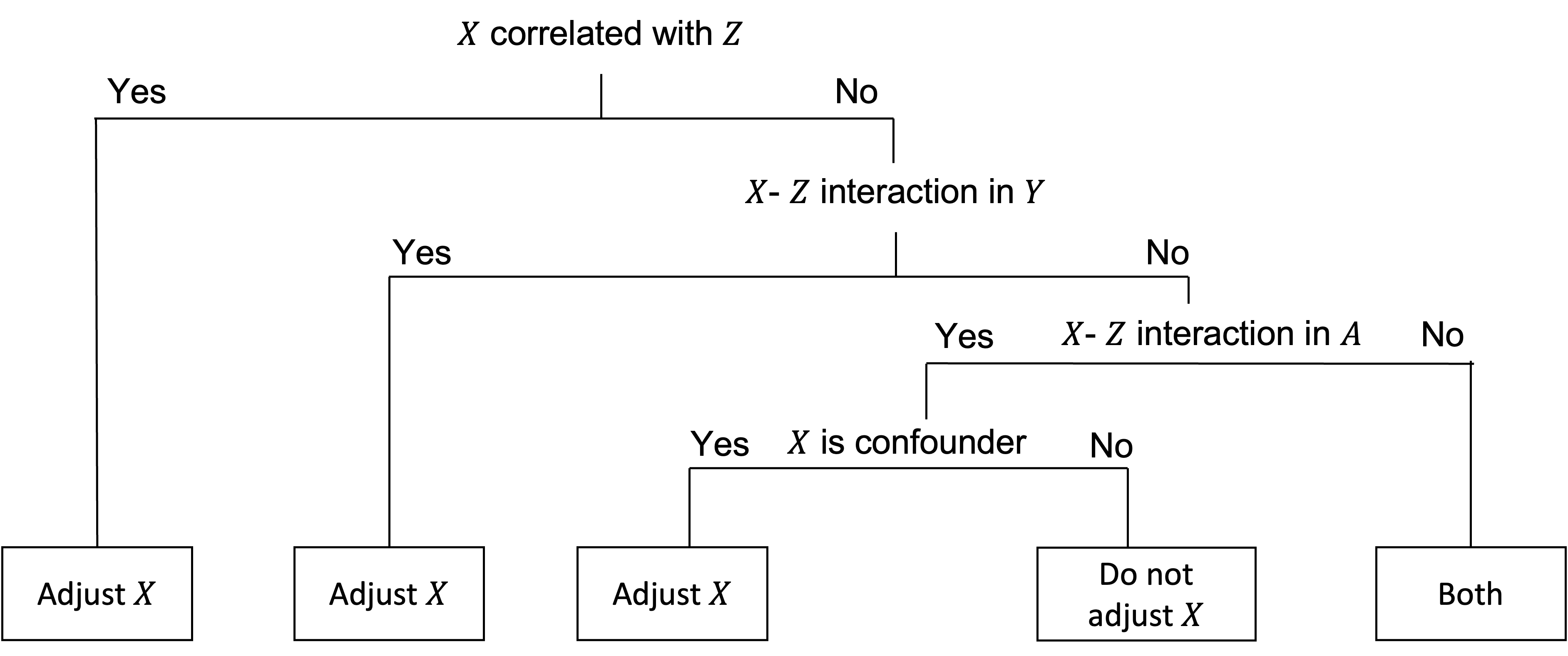}
	\caption{A diagram of how to choose $\bm X$ in order to satisfy our identifying assumptions.} \label{fig: adjust}
\end{figure}

\subsection{Identification under assumption violation}
 Consider the following exposure and outcome models:
 \begin{align*}
 	Y&= \beta_0 A + \alpha (\bm Z) + \xi_y (\bm Z) U + \epsilon_Y, \\ 
 	A&= \gamma (\bm Z) + \xi_a(\bm Z) U + \sigma(\bm Z) \epsilon_A.
 \end{align*}
 For any SNP $Z_j$, the identified parameter is 
 \begin{align}
 	\begin{split}
	\beta_j &= \frac{E[ (Z_j - E(Z_j) ) (A - E(A\mid \bm Z)) Y ]}{E[ (Z_j - E(Z_j) ) (A - E(A\mid \bm Z)) A ]} \\
& = \beta_0 +  \frac{E[ (Z_j - E(Z_j) ) (A - E(A\mid \bm Z)) \{ \alpha (\bm Z) + \xi_y (\bm Z) U \} ]}{E[ (Z_j - E(Z_j) ) (A - E(A\mid \bm Z)) A ]}   \\
& =  \beta_0 +   \frac{E[ (Z_j - E(Z_j) )  \xi_a(\bm Z) (U - E(U))  \{ \alpha (\bm Z) + \xi_y (\bm Z) U \} ]}{E[ (Z_j - E(Z_j) ) \{\xi_a(\bm Z) (U - E(U)) + \sigma(\bm Z) \epsilon_A \}^2 ]}   \\
& = \beta_0 +   \frac{E[ (Z_j - E(Z_j) )  \xi_a(\bm Z)  \xi_y (\bm Z)   ]  }{E[ (Z_j - E(Z_j) ) \xi_a^2(\bm Z) ]   +  E[ (Z_j - E(Z_j) ) \sigma^2(\bm Z)] {\rm Var} ( \epsilon_A)/{\rm Var} ( U)  }.
 	\end{split} \label{eq: assumption violation}
 \end{align}
 When there exist $\bm Z$-$U$ interactions in the specification of $Y$ and $A$, i.e., when $\xi_a(\bm Z)$ and/or $\xi_a(\bm Z)$ are not constants, $\beta_j$ would typically 
 depend on $j$ in MR applications, so there exists no $\beta$ that satisfies all the moment conditions, under which the overidentification test has power to detect assumption violations. 
 
 One situation when the overidentification test has no power to detect assumption violations is when the $\beta_j$ is equally biased for all $j$. This can happen under a peculiar situation when $\xi_a(\bm Z) = \xi_{a0} +  \xi_{az}   \sum_{j=1}^m Z_j $,  $\xi_y(\bm Z) = \xi_{y0} +  \xi_{yz}   \sum_{j=1}^m Z_j $, and  $\sigma(\bm Z) =\sigma_0 + \sigma_z  \sum_{j=1}^m Z_j $, and  $Z_j$'s are mutually independent and identically distributed, i.e., when the role of each $Z_j$ is homogeneous.

 \subsection{Interpretation of $\mu_n$}

 The key quantity that determines the asymptotic variance and convergence rate of $ \hat{\bbeta}- \bbeta_0 $ in \eqref{eq: a.var} is $  n G^T \Omega^{-1} G $. Moreover, from Assumption \ref{assump: many weak moments}, the many weak moment asymptotics requires $ \mu_n\rightarrow \infty $, which corresponds to $ n G^T \Omega^{-1} G $ going to infinity. 
 
 We give a simple example to provide an interpretation of $ n G^T \Omega^{-1} G $. Suppose that 
 \begin{eqnarray}
 	Y&=&\beta_0A+\sum_{j=1}^{m} \alpha_j Z_j +\xi_y(U)+\epsilon_Y, \nonumber\\
 	A&=&\sum_{j=1}^{m} \gamma_j Z_j +\xi_a(U)+\sigma(\bmZ) \epsilon_A, \nonumber
 \end{eqnarray}
 $\bmZ\perp ( U, \epsilon_Y, \epsilon_A)$,  $ \epsilon_A\perp (\epsilon_Y, U) $,   $E(\epsilon_Y)= E(\epsilon_A)= 0$,  and  $Z_1, \dots, Z_m$ are mutually independent. We again emphasize that the mutual independence of $ Z_1,\dots, Z_m $ is not needed  for our method but is assumed here for ease of illustration. The influence function in \eqref{eq: IF} suppressing $\bmX$ evaluated at $\beta=\beta_0$ is 
 \begin{align*}
 	g^{IF} (\O; \beta_0, \bfeta_0) =  \{ \bmZ- E(\bmZ)\} \{\Delta_0- E(\Delta_0)\}, \quad \text{with } \Delta_0= R_AR_{Y}-\beta_0 R_A^2. 
 \end{align*}
 By definition, expressions for $ G, \Omega$ and $ nG^T \Omega^{-1} G $ can be written as 
 \begin{align}
 	&G=-E\left\{ (\bmZ- E(\bmZ))\var (A\mid\bmZ)\right\}= -\big(\cov\{Z_1, \var(A\mid\bmZ)\}, \dots, \cov\{Z_m,\var(A\mid\bmZ)\}\big)^T,\nonumber\\
 	&\Omega= E\left\{ (\bmZ-E(\bmZ)) (\bmZ-E(\bmZ))^T \Delta_0^2\right\}\approx \var (\bmZ) E( \Delta_0^2),    \label{eq: approximate Omega}\\
 	&nG^{T} \Omega^{-1} G \approx \frac{n}{E( \Delta_0^2) } \sum_{j=1}^{m} \frac{\{\cov (Z_j, \var(A\mid\bmZ))\}^2}{\var(Z_j)}, \nonumber
 \end{align}
 where the approximation for $\Omega$ in \eqref{eq: approximate Omega} holds when $\sigma(\bmZ) - E\{\sigma(\bmZ)\} $ is small in magnitude compared to $E\{\sigma(\bmZ)\}$ (shown below). 
 Hence, for identification to be strong in the sense that $ nG^{T} \Omega^{-1} G $ is of order $ n $,  we need $ \sum_{j=1}^{m} \{\cov (Z_j, \var(A\mid\bmZ))\}^2/\var(Z_j)$ to be of a constant order. A closer look at the formula indicates that $\sum_{j=1}^{m}\{ \cov (Z_j, \var(A\mid\bmZ))\}^2/\var(Z_j)= \sum_{j=1}^{m} a_j^2 \var(Z_j)$, which  measures the total variance of $ \var(A\mid\bmZ) $ that can be explained by the set of $ m $ SNPs, where $ a_j = \cov(Z_j, \var(A\mid\bmZ))/\var(Z_j)$ is the population coefficient from the least squares regression of $ \var(A\mid\bmZ) $ on $ Z_j $. Evidently, the use of many SNPs that are predictive of $ \var(A\mid\bmZ) $ can strengthen identification and improve estimation accuracy.

Now we show the approximation  for $\Omega$ in \eqref{eq: approximate Omega}  when  $\sigma(\bmZ) - E\{\sigma(\bmZ)\} $ is small in magnitude compared to $E\{\sigma(\bmZ)\}$.  Let $\bar\xi_a(U)= \xi_a(U) - E\{ \xi_a(U)\}$ and $\bar\xi_y(U)= \xi_y(U) - E\{ \xi_y(U)\}$, and we have 
\begin{align*}
	&R_A= \bar\xi_a(U)+ \sigma(\bmZ) \epsilon_A,\\
	&R_{Y} - \beta_0 R_A= \bar\xi_y(U)+  \epsilon_Y.
\end{align*}
It is not difficult to show that 
\begin{align*}
	\Omega&= E\{(\bmZ- E(\bmZ))(\bmZ- E(\bmZ))^T\Delta_0^2\}\\
	&=  E\{(\bmZ- E(\bmZ))(\bmZ- E(\bmZ))^TR_A^2 (R_Y- \beta_0 R_A)^2\}   \\
	&=  E\{(\bmZ- E(\bmZ))(\bmZ- E(\bmZ))^T(\bar\xi_a(U)+ \sigma(\bmZ) \epsilon_A)^2 (\bar\xi_y(U)+  \epsilon_Y )^2\}   \\
	&= E\{(\bmZ- E(\bmZ))(\bmZ- E(\bmZ))^T(\bar\xi_a(U)^2+ 2\sigma(\bmZ) \epsilon_A + \sigma^2(\bmZ) \epsilon_A^2) (\bar\xi_y(U)+  \epsilon_Y )^2\}   \\
	&= E\{(\bmZ- E(\bmZ))(\bmZ- E(\bmZ))^T \bar\xi_a(U)^2 (\bar\xi_y(U)+  \epsilon_Y )^2\}  \\
	&\qquad +  E\{(\bmZ- E(\bmZ))(\bmZ- E(\bmZ))^T \sigma^2(\bmZ) \epsilon_A^2 (\bar\xi_y(U)+  \epsilon_Y )^2\}  \\
	&= \var(\bmZ) E\{\bar\xi_a(U)^2 (\bar\xi_y(U)+  \epsilon_Y )^2\}   + \var(\bmZ) E\{\sigma^2(\bmZ)\}  E\{\epsilon_A^2 (\bar\xi_y(U)+  \epsilon_Y )^2\} +o(1) \\
	&= \var(\bmZ) E\{\bar\xi_a(U)^2 (\bar\xi_y(U)+  \epsilon_Y )^2\}   + \var(\bmZ) E\{\sigma^2(\bmZ) \epsilon_A^2 (\bar\xi_y(U)+  \epsilon_Y )^2\} +o(1) \\
	&=  \var(\bmZ) E(\Delta_0^2) + o(1),
\end{align*}
where the fifth equality is from $ \epsilon_A \perp (\bmZ, U, \epsilon_Y) $, the sixth equality uses the condition that $ \sigma(\bmZ)- E\{ \sigma(\bmZ) \}  $ is small compared to $ E\{ \sigma(\bmZ) \}  $.

\subsection{Details of kernel estimators $ \hat \omega(\cdot), \hat \theta(\cdot)$}

The kernel estimator of $  \omega_0(\bm x; \bm\mu_0, {\bm\lambda_0}) = E(R_A R_Y\mid \bmX = \bmx)$, denoted as $ \hat \omega(\bm x; \hat{\bm\mu}, \hat{\bm\lambda})$, is defined as 
\begin{align*}
    \hat \omega(\bm x; \hat{\bm\mu}, \hat{\bm\lambda})= \frac{\sum_{i=1}^n R_{Ai} (\hat{\bm \mu}) R_{Yi}(\hat{\bm \lambda})  K_\sigma (\bm x - \bmX_i) }{\sum_{i=1}^n K_\sigma (\bm x - \bmX_i) }
\end{align*}
where $R_{Ai} (\hat{\bm \mu}) = A_i - (\bmX_i^T, \bmZ_i^T)\hat{\bm \mu}, 
R_{Yi}(\hat{\bm \lambda}) = Y_i - (\bmX_i^T, \bmZ_i^T)\hat{\bm \lambda},$   
$K_\sigma (\bmx - \bmX_i ) = \sigma^{-d_x} K (\frac{\bmx - \bmX_i }{\sigma})$, $d_x$ is the dimension of $\bmX$, $K(u)$ is a function such that $\int K(u) du = 1$, and $\sigma$ is a bandwidth term. The kernel estimator of $  \theta_0(\bm x; \bm\mu_0) = E(R_A^2\mid \bmX = \bmx)$, denoted as $ \hat \theta(\bm x; \hat{\bm \mu})$, is defined as 
\begin{align*}
    \hat \theta(\bm x; \hat{\bm \mu}) = \frac{\sum_{i=1}^n R_{Ai} (\hat{\bm \mu})^2  K_\sigma (\bm x - \bmX_i) }{\sum_{i=1}^n K_\sigma (\bm x - \bmX_i) } .
\end{align*}
In our simulations, since both the exposure and outcome are continuous variables, we use the Gaussian kernel, and the bandwidth is selected using least-squares cross-validation using the \textsf{npregbw} function in \textsf{np} package in \textsf{R}.

\section{Additional simulation results}

\subsection{When there are  observed covariates}
We evaluate the finite sample performance of the proposed GENIUS-MAWII estimator when there exists an observed covariate and we estimate $ E(R_AR_Y\mid \bmX=\bmx) $ and $ E(R_A^2\mid \bmX=\bmx) $ using the  nonparametric kernel regression or least squares estimation (LSE) with quadratic terms.  The setting is identical to that in Section \ref{subsec: simu1}  except that we generate the outcome from $ Y= \beta_0 A+  \sum_{j=1}^m \alpha_j Z_j + \eta_Y U (1+X) +\epsilon_Y$, the exposure from $ A=   \sum_{j=1}^m  \gamma_j Z_j +  \eta_A U (1+X) + \left(1+ \sum_{j=1}^m \delta_j  Z_j \right) \epsilon_A$, where $ X\sim N(0,0.6(1-h^2))$. Since both the exposure and outcome are continuous variables, we use the Gaussian kernel, and the bandwidth is selected using least-squares cross-validation using the \textsf{npregbw} function in \textsf{np} package in \textsf{R}.
Table \ref{supptb: sim1} presents the results of GENIUS-MAWII (kernel), GENIUS-MAWII (LSE), 2SLS, and LIML under Setting 4; the other two GENIUS estimators are not included because we have shown in Table \ref{tb: sim1} that GENIUS-MAWII has better performance, the  other four MR methods are not included because they cannot adjust for observed covariates. The conclusions from Table \ref{supptb: sim1} are similar to those in Table \ref{tb: sim1}. Specifically, {when $n=100,000$,
	 the two GENIUS-MAWII estimators are similar, showing negligible bias and nominal coverage, and the SEs calculated using \eqref{eq: var est} are close to the simulation SDs; when $n=10,000$, $F_{\rm GENIUS}$ is too small and the two GENIUS-MAWII estimators have some attenuation bias.} In contrast, the 2SLS and LIML are biased because of failing to address the horizontal pleiotropy. 

\begin{table}[ht]
	\caption{Simulation results based on 1,000 Monte Carlo samples under Setting 4 with $\beta_0=0.4$ and $m=100$ when there are observed covariates; $F_{\rm GENIUS}$ represents the average F-statistics for GENIUS-MAWII discussed in Section \ref{subsec: measure of weak identification}, SD is the Monte Carlo standard deviation,  SE is the average standard error, CP is the coverage probability  of 95\% asymptotic confidence interval. The empirical power for $n=10,000$ and $n=100,000$ is respectively 0.043 and 0.040. \label{supptb: sim1}} \vspace{2mm}
	\centering
	\resizebox{0.98\textwidth}{!}{\begin{tabular}{llccccclcccc} \hline 
			&              & \multicolumn{4}{c}{$n=10,000, F_{\rm GENIUS} =1.51 $}         &  &             & \multicolumn{4}{c}{$n=100,000, F_{\rm GENIUS} = 5.59$} \\ \cline{3-6} \cline{9-12}
			Setting     & Method       & Mean   & SD     & SE     & CP &&& Mean   & SD     & SE     & CP   \\ \hline
			Setting 4 & GENIUS-MAWII (Kernel)& 0.364&	0.167&	0.174 &	96.2 &&& 0.395&	0.044&	0.044&	95.7   \\
			 & GENIUS-MAWII (LSE)& 0.363&	0.168&	0.175 &	96.5 &&&0.395	&0.044	&0.044	&96.1  \\
			& 2SLS         &0.870	&0.031	&0.018&	0.0 &&   & 0.852	&0.026	&0.006	&0.0  \\
			& LIML         & 0.835	&0.039	&0.019	&0.0  && &  0.835	&0.033	&0.006	&0.0 \\\hline
	\end{tabular}}
\end{table}

\subsection{When SNPs are in Linkage Disequilibrium (LD)}
We evaluate the finite sample performance of the proposed GENIUS-MAWII estimator when SNPs are in LD. For each simulation, our SNPs data are sampled with replacement from 100 SNPs in a region of chromosome 1 on 219,762 individuals of European ancestry from the UK Biobank dataset that passed the following standard quality control process: removing variants and individuals with missing data, removing variants with low p-value from the Hardy-Weinberg Equilibrium Fisher's exact test and deleting variants with minor allele frequency less than 0.05. Other aspects of the setting are identical to Setting 4 in Section \ref{sec: simu}. The results are in Table \ref{supptb: sim2}, based on which the conclusions are similar to that in Table \ref{tb: sim1}.

\begin{table}[ht]
	\caption{Simulation results based on 1,000 Monte Carlo repetitions with $\beta_0=0.4$ and $m=100$ when SNPs are in LD;  $F_{\rm GENIUS}$ represents the average F-statistics for GENIUS-MAWII discussed in Section \ref{subsec: measure of weak identification}, SD is the Monte Carlo standard deviation,  SE is the average standard error, CP is the coverage probability  of 95\% asymptotic confidence interval. The empirical power for $n=10,000$ and $n=100,000$ is respectively 0.041 and 0.058. \label{supptb: sim2} } \vspace{2mm}
	\centering
	\resizebox{!}{!}{\begin{tabular}{llccccclcccc}  \hline 
			&              & \multicolumn{4}{c}{$n=10,000, F_{\rm GENIUS}= 1.69$}         &  &             & \multicolumn{4}{c}{$n=100,000, F_{\rm GENIUS}= 8.12$} \\ \cline{3-6} \cline{9-12} 
			     & Method       & Mean   & SD     & SE     & CP    &  &       & Mean   & SD     & SE     & CP    \\ \hline
			 & GENIUS-MAWII & 0.352	&0.146	&0.147	&97.2  &&& 0.395 &	0.036&	0.034&	93.8 \\
			 & 2SLS      &  0.872	&0.039	&0.021	&0.0  &&& 0.852	&0.033	&0.007	&0.0   \\
			& LIML    & 0.838	&0.047&	0.022	&0.0 &&& 0.839	&0.039&	0.007&	0.0 	\\\hline
			\end{tabular}}
\end{table}

\section{Proof of results in Sections \ref{sec: prelim}-\ref{sec: identification}}
\subsection{Relaxed assumptions for identification using \eqref{eq: continuous y}}
We present a weaker assumption under which $ \beta_0 $ is the unique solution to \eqref{eq: continuous y}.

\noindent{\bf Assumption. }
(a) $E(Y\mid A,  U, \bmZ)=\beta_0A+ \alpha ( U, \bmZ) +\xi_y(U)$.\\
(b) $E(A\mid U, \bmZ)=\gamma(U, \bmZ)+\xi_a(U)$. \\
(c) $\cov(\xi_{y} (U), \xi_{a}(U)\mid\bmZ)=c$ with probability 1, where $ c $ is a generic constant.\\
(d) The orthogonality conditions $\cov(\alpha(U, \bmZ), \gamma(U, \bmZ)\mid\bmZ)=0$, $\cov(\alpha (U, \bmZ), \xi_a(U)\mid\bmZ)=0$, and
$\cov(\xi_y(U), \gamma(U, \bmZ)\mid\bmZ)=0$ hold with probability 1. \\
Here,  $ \alpha, \gamma, \xi_y,  \xi_a $ are unspecified functions satisfying  $ \alpha(U, \zero)= \gamma(U, \zero)=0 $.

\subsection{Proof of \eqref{eq: continuous y, X}}
Under \eqref{eq: out model X}-\eqref{eq: exp model X} and $ \bmZ\perp U\mid\bmX $, then
\begin{align}
	& E(R_A( Y- \beta_0 A)\mid\bmZ, \bmX)\nonumber\\
	&=E\{R_AE( Y- \beta_0 A\mid A, U, \bmZ, \bmX)\mid\bmZ, \bmX\}\nonumber\\
	&=E\{R_A(\alpha(\bmZ, \bmX)+ \xi_y(U, \bmX))\mid\bmZ, \bmX\}\nonumber\\
	&=E\{E(R_A\mid U, \bmZ, \bmX)(\alpha(\bmZ, \bmX)+ \xi_y(U, \bmX))\mid\bmZ, \bmX\} \nonumber\\
	&=E\left[\left\{\gamma (\bmZ, \bmX)+ \xi_a(U, \bmX)- E(A\mid\bmZ, \bmX)\right\}\left\{\alpha(\bmZ, \bmX)+ \xi_y(U, \bmX)\right\}\mid\bmZ, \bmX\right] \nonumber\\
	&=\cov \left\{\gamma (\bmZ, \bmX)+ \xi_a(U, \bmX),\alpha(\bmZ, \bmX)+ \xi_y(U, \bmX)\mid\bmZ, \bmX\right\}\nonumber \\
	&=\cov \left\{ \xi_a(U, \bmX), \xi_y(U, \bmX)\mid\bmZ, \bmX\right\} \nonumber\\
	&=\cov \left\{ \xi_a(U, \bmX), \xi_y(U, \bmX)\mid \bmX\right\} \label{seq: conditional mean 1},
\end{align}
where the last equality is from $ \bmZ\perp U\mid\bmX $. Therefore, 
\begin{align*}
&	E[(\bmZ- E(\bmZ\mid\bmX)) R_A( Y- \beta_0 A)] = E[(\bmZ- E(\bmZ\mid\bmX)) E(R_A( Y- \beta_0 A) \mid\bmZ, \bmX)]  \\
	&= E[(\bmZ- E(\bmZ\mid\bmX))  \cov \left\{ \xi_a(U, \bmX), \xi_y(U, \bmX)\mid \bmX\right\}  ]= 0.
\end{align*}

To show $ \beta_0 $ is the unique solution to  \eqref{eq: continuous y, X}, note that 
\begin{align*}
	&E[(\bmZ- E(\bmZ\mid\bmX)) R_A( Y- \beta A)]\\ 
	&=E[(\bmZ- E(\bmZ\mid\bmX)) R_A( Y- \beta A) ] - E[(\bmZ- E(\bmZ\mid\bmX)) R_A ( Y- \beta_0 A) ] \\
	&= (\beta_0-\beta)  E\left\{ (\bmZ- E(\bmZ\mid\bmX)) R_A A    \right\} \nonumber,
\end{align*}
which is zero if and only if $ \beta_0=\beta $, provided that  $ E\left\{ (\bmZ- E(\bmZ\mid\bmX)) R_A A    \right\}\neq 0 $.

\subsection{Proof of Theorem \ref{theoX}} 
Let $P_{\theta}$ denote a parametric submodel with $P_0=P$, where $P$ is the true distribution of the observed data $\bm O=(Y,A,\bmZ,\bmX)$. The score corresponding to $P_{\theta}$ with density $dP_{\theta}$ is $S_\theta(O)=\partial \ln (dP_{\theta})/\partial \theta$, which can be decomposed as $S_{\theta}(Y\mid A,\bmZ,\bmX)+S_{\theta}(A\mid\bmZ,\bmX)+S_{\theta}(\bmZ\mid\bmX)+S_{\theta}(\bmX)$. Let $E_{\theta}[\cdot]$ denote the expectation at the distribution $P_{\theta}$. Let $R_Y(\theta)=Y-E_\theta [Y\mid \bmZ, \bmX]$, $R_A(\theta)=A-E_{\theta}[A\mid\bmZ,\bmX]$ and $R^h_Z(\theta)=h(\bmZ,\bmX)-E_{\theta}[h(\bmZ,\bmX)\mid\bmX]$. The conditional independence restriction $E_{\theta}[(Y- \beta(\theta)A) R_A(\theta)\mid\bmZ,\bmX]=E_{\theta}[(Y- \beta(\theta)A)  R_A(\theta)\mid\bmX]$ is equivalent to the class of unconditional restrictions 
\begin{flalign}
	E_{\theta}[(Y- \beta(\theta)A) R_{A}(\theta)R^h_Z(\theta)]=0, \label{eq:restriction}
\end{flalign}
for any scalar-valued function $ h(\bmZ,\bmX) $. 
Following \cite{Newey1994},  we obtain the influence functions for estimation of $\beta_0$ by deriving the pathwise derivatives $\partial \beta({\theta})/\partial \theta \rvert_{\theta=0}$ based on (\ref{eq:restriction}). For each $h$, differentiating under the integral yields 
\begin{flalign*}
	\frac{\partial \beta(\theta)}{\partial \theta}&=-\mathcal{L}_h(\theta)^{-1} (E_{\theta}[ (Y- \beta(\theta)A) R_{A}(\theta)R^h_Z(\theta)S_{\theta}(\bm O)]\\
	&\quad \quad \phantom{=} -E_{\theta}[(Y- \beta(\theta)A) E_{\theta}\{A S_{\theta}(A\mid\bmZ,\bmX)\rvert\bmZ,\bmX\}R^h_Z(\theta)]\\
	&\quad\quad\phantom{=}-E_{\theta}[ (Y- \beta(\theta)A) R_{A}(\theta)E_{\theta}\{hS_{\theta}(\bmZ\mid\bmX)\rvert \bmX\} ])\\
	&\equiv -\mathcal{L}_h(\theta)^{-1} (A_1 - A_2 - A_3),
\end{flalign*}
where $\mathcal{L}_h(\theta)=- E_{\theta}[AR_{A}(\theta)R^h_Z(\theta)]$. Then, we use the following identities repeatedly:
\begin{flalign}
	\label{identity_1}	&  E_{\theta}[b(A,\bmZ,\bmX)S_{\theta}(Y\mid A,\bmZ,\bmX)] = 0, \quad \text{for all } b;\\
	\label{identity_2}	& E_{\theta}[c(\bmZ,\bmX)S_{\theta}(A\mid\bmZ,\bmX)]  =E_{\theta}[c(\bmZ,\bmX) R_{A}(\theta)]= 0, \quad \text{for all } c; \\
	\label{identity_3}	& E_{\theta}[d(\bmX)S_{\theta}(\bmZ\mid\bmX)] = E_{\theta}[d(\bmX)R^h_Z(\theta)]= 0, \quad \text{for all } d,h.
\end{flalign}
Consider the terms $A_1, A_2, A_3$ separately:
\begin{flalign*}
	A_1&=E_{\theta}[ (Y- \beta(\theta)A) R_A(\theta)R^h_Z(\theta)S_{\theta}(\bm O)],\\
	A_2&=E_{\theta}[ (Y- \beta(\theta)A) E_{\theta}\{A S_{\theta}(A\mid\bmZ,\bmX)\rvert\bmZ,\bmX\}R^h_Z(\theta)]\\
	&=E_{\theta}[ E_{\theta}\{ (Y- \beta(\theta)A) \mid\bmZ,\bmX\}E_{\theta}\{A S_{\theta}(A\mid\bmZ,\bmX)\rvert\bmZ,\bmX\}R^h_Z(\theta)]\\
	&=E_{\theta}[ E_{\theta}\{ (Y- \beta(\theta)A) \mid\bmZ,\bmX\}A S_{\theta}(A\mid\bmZ,\bmX)R^h_Z(\theta)]\\
	&=E_{\theta}[E_{\theta}\{(Y- \beta(\theta)A) \mid\bmZ,\bmX\}\{A-E_{\theta}(A\mid\bmZ,\bmX)\} S_{\theta}(A\mid\bmZ,\bmX)R^h_Z(\theta)]\quad \text{(\ref{identity_2})}\\
	&=E_{\theta}[  E_{\theta}\{(Y- \beta(\theta)A) \mid\bmZ,\bmX\}R_A(\theta) R^h_Z(\theta)S_{\theta}(A\mid\bmZ,\bmX)]\\
	&\phantom{=}+ E_{\theta}[  E_{\theta}\{ (Y- \beta(\theta)A)\mid\bmZ,\bmX\}R_A(\theta) R^h_Z(\theta)S_{\theta}(Y\mid A,\bmZ,\bmX) ] \quad \text{(\ref{identity_1})}\\
	&\phantom{=}+E_{\theta}[  E_{\theta}\{ (Y- \beta(\theta)A)\mid\bmZ,\bmX\}R_A(\theta) R^h_Z(\theta)\{S_{\theta}(\bmZ\mid\bmX)+S_{\theta}(\bmX)\}] \quad \text{(\ref{identity_2})}\\
	&=E_{\theta}[  E_{\theta}\{ (Y- \beta(\theta)A)\mid\bmZ,\bmX\}R_A(\theta) R^h_Z(\theta) S_{\theta}(\bm O)],\\
	A_3&=E_{\theta}[  (Y- \beta(\theta)A) R_A(\theta)E_{\theta}\{hS_{\theta}(\bmZ\mid\bmX)\rvert \bmX\} ]\\
	&=E_{\theta}[  E_{\theta}\{ (Y- \beta(\theta)A) R_A(\theta) - E( (Y- \beta(\theta)A) \mid \bmZ, \bmX  ) R_A(\theta) \rvert \bmX\}E_{\theta}\{hS_{\theta}(\bmZ\mid\bmX)\rvert \bmX\} ]\\
	&=E_{\theta}[  E_{\theta}\{ (Y- \beta(\theta)A) R_A(\theta) - E( (Y- \beta(\theta)A) \mid \bmZ, \bmX  ) R_A(\theta) \rvert \bmX\}hS_{\theta}(\bmZ\mid\bmX) ]\\
	&=E_{\theta}[  E_{\theta}\{ (Y- \beta(\theta)A) R_A(\theta) - E((Y- \beta(\theta)A) \mid \bmZ, \bmX  ) R_A(\theta) \rvert \bmX\}R^h_Z(\theta)S_{\theta}(\bmZ\mid\bmX) ]\quad \text{(\ref{identity_3})}\\
	&=E_{\theta}[  E_{\theta}\{(Y- \beta(\theta)A)R_A(\theta) - E((Y- \beta(\theta)A)\mid \bmZ, \bmX  ) R_A(\theta) \rvert \bmX\}R^h_Z(\theta)S_{\theta}(\bmZ\mid\bmX) ]\\
	&\phantom{=}+ E_{\theta}[  E_{\theta}\{ (Y- \beta(\theta)A)R_A(\theta)- E( (Y- \beta(\theta)A)\mid \bmZ, \bmX  ) R_A(\theta)\rvert \bmX\}R^h_Z(\theta)S_{\theta}(\bmX) ]   \quad \text{(\ref{identity_3})}\\
	&\phantom{=}+E_{\theta}[  E_{\theta}\{ (Y- \beta(\theta)A)R_A(\theta)- E((Y- \beta(\theta)A)\mid \bmZ, \bmX  ) R_A(\theta) \rvert \bmX\}R^h_Z(\theta)S_{\theta}(A\mid\bmZ,\bmX) ]  \quad \text{(\ref{identity_2})}\\
	&\phantom{=}+E_{\theta}[  E_{\theta}\{(Y- \beta(\theta)A)R_A(\theta)- E( (Y- \beta(\theta)A) \mid \bmZ, \bmX  ) R_A(\theta) \rvert \bmX\}R^h_Z(\theta)S_{\theta}(Y\mid A,\bmZ,\bmX) ]   \quad \text{(\ref{identity_1})}\\
	&=E_{\theta}[E_{\theta}\{ (Y- \beta(\theta)A) R_A(\theta)- E( (Y- \beta(\theta)A)\mid \bmZ, \bmX  ) R_A(\theta) \rvert \bmX \} R^h_Z(\theta) S_{\theta}(\bm O)].
\end{flalign*}

Therefore $\partial \beta({\theta})/\partial \theta \rvert_{\theta=0}=E\{g^{IF}(\bm O;h)S(\bm O)\}$, where
$$g^{IF}(\bm O;h)=\mathcal{L}_h^{-1}R^h_Z[R_A R_{Y} -\beta R_A^2 -E(  R_A R_{Y} -\beta R_A^2  \mid  \bmX )],$$
and $\mathcal{L}_h=E[AR_{A}R^h_Z]$. Since this derivation holds for each $h$, the set of influence functions of $\beta_0$  is given by $\{g^{IF}(\bm O;h)\}$. In addition, we note that $\mathcal{L}_h=E[\text{var}(A\mid\bmZ,\bmX) R^h_Z]$; a necessary condition for non-singularity of $\mathcal{L}_h$ is the dependence of $\text{var}(A\mid\bmZ,\bmX)$ on $\bmZ$, an assumption made in the main manuscript.

(b) First, notice that because every $ Z_j $ is discrete and takes values 0,1,2, any function $ h(\bmZ, \bm X) $ can be expressed as $ C(\bm X)^T \bar{\bmZ} $, where $\bar  \bmZ $ contains all the dummy variables of the strata defined by $ \bmZ $.  Define $ d(\bmZ, \bmX) =\bar\bmZ- E(\bar{\bmZ}\mid\bmX) $,  all the influence functions of $ \beta_0 $ can be characterized  by 
\begin{align}
	U(C) = C(\bmX)^T d(\bmZ, \bmX)\{ \Delta- E(\Delta \mid \bmX ) \}\nonumber.
\end{align}
From the property of the efficient influence function \citep{tsiatis2007semiparametric}, for any $ C(\bmX) $, 
\begin{align}
	E[U(C) U(C^{opt})] = - E\left[ \frac{\partial U(C)}{\partial \beta}\right] \nonumber,
\end{align}
where $ U(C^{opt}) $ corresponds to the efficient influence function. Thus,
\begin{align}
	&E\left[  C (\bmX)^T d(\bmZ, \bmX) d(\bmZ, \bmX)^T C^{opt}(\bmX)\{\Delta - E(\Delta \mid \bmX )\} ^2 \right] \nonumber\\
	&\qquad + E\left[  C(\bmX)^T d(\bmZ, \bmX)\frac{\partial \{ \Delta- E(\Delta\mid \bmX) \} }{\partial \beta}  \right]  \nonumber \\
	&=E\left[  C(\bmX)^T d(\bmZ, \bmX) \left\{d(\bmZ, \bmX)^T C^{opt}(\bmX)\{\Delta - E(\Delta \mid \bmX )\} ^2+\frac{\partial \{ \Delta- E(\Delta\mid \bmX) \} }{\partial \beta} \right\} \right] \nonumber \\
	&= E\bigg[  C(\bmX)^T\bigg[ E\left\{d(\bmZ, \bmX) d(\bmZ, \bmX)^T\{\Delta - E(\Delta \mid \bmX )\} ^2 \mid\bmX\right\}C^{opt}(\bmX) \nonumber  \\
	&\qquad 	+E\bigg\{ d(\bmZ, \bmX)\frac{\partial\{ \Delta- E(\Delta\mid \bmX)\}}{\partial \beta} \mid\bmX\bigg\}\bigg] \bigg] =0 \nonumber.
\end{align}
Since the above equation holds for any $ C(\bmX) $, it holds when 
\begin{align}
C(\bmX)&=  E\left\{d(\bmZ, \bmX) d(\bmZ, \bmX)^T(\Delta- E(\Delta\mid\bmX) \} ^2 \mid\bmX\right\}C^{opt}(\bmX) \nonumber\\
&\qquad +E\left\{ d(\bmZ, \bmX)\frac{\partial\{ \Delta- E(\Delta\mid \bmX)\}}{\partial \beta} \mid\bmX\right\}, \label{seq: c opt}
\end{align}
which implies that \eqref{seq: c opt} equals zero almost surely. Therefore,  
\begin{align}
	&C^{opt} (\bmX) \\
	&= -\left[ E\left\{d(\bmZ, \bmX) d(\bmZ, \bmX)^T(\Delta- E(\Delta\mid \bmX)) ^2 \mid\bmX\right\}\right]^{-1} E\left\{ d(\bmZ, \bmX)\frac{\partial(\Delta-E(\Delta\mid\bmX))}{\partial \beta} \mid\bmX\right\}\nonumber,
\end{align}
where $ \frac{\partial (\Delta- E(\Delta\mid \bmX))}{\partial \beta} = -R_A^2+ E(R_A^2\mid\bmX)$.

\subsection{Proof of multiple robustness}
We will show that the influence function in \eqref{eq: general IF} evaluated at $ \beta=\beta_0 $ has expectation zero  when either one of the following three sets of the models is correctly specified: $ \{E(A\mid\bmZ, \bmX),  E(h(\bmZ, \bmX)\mid\bmX)\}$, $ \{E(Y\mid\bmZ, \bmX),  E(h(\bmZ, \bmX)\mid\bmX)\}$,  or $ \{E(A\mid\bmZ, \bmX),  E( R_A ( R_Y- \beta R_A) \mid \bmX)\}$. In the following, we use a bar notation to denote a general specification of the model that is not necessarily correct.

 In the first scenario,   $ \bar E(A\mid \bmZ, \bmX) =  E(A\mid \bmZ, \bmX) $ and $ \bar E(h(\bmZ, \bmX)\mid \bmX)  =   E(h(\bmZ, \bmX)\mid \bmX)   $. Then 
 \begin{align*}
 &E\left[	\big\{h(\bmZ, \bmX)- E(h(\bmZ, \bmX)\mid\bmX)\big\} \{R_A \bar R_Y - \beta_0 R_A^2- \bar E(R_A \bar R_Y  - \beta_0 R_A ^2 \mid \bmX )\} \right] \\
 &= E\left[	\big\{h(\bmZ, \bmX)- E(h(\bmZ, \bmX)\mid\bmX)\big\} \{R_A \bar R_Y - \beta_0 R_A^2 \}  \right] \\ 
 &= E\left[	\big\{h(\bmZ, \bmX)- E(h(\bmZ, \bmX)\mid\bmX)\big\} \{R_A (Y-  \bar E(Y\mid \bm Z, \bmX ) -\beta_0 A  + \beta_0 E(A\mid \bmZ, \bmX) ) \}  \right] \\ 
 &= E\left[	\big\{h(\bmZ, \bmX)- E(h(\bmZ, \bmX)\mid\bmX)\big\} \{R_A (Y-\beta_0 A  ) \}  \right] \\
 &= 0 .
 \end{align*}

In the second scenario,  $ \bar E(Y| \bmZ, \bmX) =  E(Y\mid \bmZ, \bmX) $ and $ \bar E(h(\bmZ, \bmX)\mid \bmX)  =   E(h(\bmZ, \bmX)\mid \bmX)   $. Then 
 \begin{align*}
	&E\left[	\big\{h(\bmZ, \bmX)- E(h(\bmZ, \bmX)\mid\bmX)\big\} \{\bar R_A R_Y - \beta_0 \bar R_A^2- \bar E(\bar R_A  R_Y  - \beta_0\bar  R_A ^2 \mid \bmX )\} \right] \\
	&= E\left[	\big\{h(\bmZ, \bmX)- E(h(\bmZ, \bmX)\mid\bmX)\big\} \{\bar R_A R_Y - \beta_0 \bar R_A^2  \}  \right] \\ 
	&= E\left[	\big\{h(\bmZ, \bmX)- E(h(\bmZ, \bmX)\mid\bmX)\big\} \{\bar R_A (Y-   E(Y\mid \bm Z, \bmX ) -\beta_0 A  + \beta_0 \bar E(A\mid \bmZ, \bmX) ) \}  \right] \\ 
	&= E\left[	\big\{h(\bmZ, \bmX)- E(h(\bmZ, \bmX)\mid\bmX)\big\} \{(\bar R_A  - R_A)(Y-   E(Y\mid \bm Z, \bmX ) -\beta_0 A  + \beta_0 \bar E(A\mid \bmZ, \bmX) ) \}  \right] \\ 
	&\qquad + E\left[	\big\{h(\bmZ, \bmX)- E(h(\bmZ, \bmX)\mid\bmX)\big\} \{ R_A(Y-   E(Y\mid \bm Z, \bmX ) -\beta_0 A  + \beta_0 \bar E(A\mid \bmZ, \bmX) ) \}  \right] \\ 
	&= E\left[	\big\{h(\bmZ, \bmX)- E(h(\bmZ, \bmX)\mid\bmX)\big\} \{(E(A\mid \bmZ, \bmX) - \bar E(A\mid \bmZ, \bmX) )( -\beta_0 A ) \}  \right] \\
	&= E\left[	\big\{h(\bmZ, \bmX)- E(h(\bmZ, \bmX)\mid\bmX)\big\} \{(E(A\mid \bmZ, \bmX) - \bar E(A\mid \bmZ, \bmX) )( -\beta_0 R_A ) \}  \right] \\ 
	&=0 .
\end{align*}

In the third scenario, $ \bar E(A \mid \bmZ, \bmX) =  E(A\mid \bmZ, \bmX) $ and $ \bar E( \bar R_A ( \bar R_Y- \beta \bar R_A) \mid \bmX)  =    E( R_A ( \bar R_Y- \beta R_A) \mid \bmX)   $. Then 
 \begin{align*}
	&E\left[	\big\{h(\bmZ, \bmX)- \bar E(h(\bmZ, \bmX)\mid\bmX)\big\} \{R_A \bar R_Y - \beta_0 R_A^2-  E(R_A  \bar R_Y  - \beta_0 R_A ^2 \mid \bmX )\} \right] \\
	&= E\left[	\big\{h(\bmZ, \bmX)-  E(h(\bmZ, \bmX)\mid\bmX)\big\} \{R_A \bar R_Y - \beta_0 R_A^2-  E(R_A  \bar R_Y  - \beta_0 R_A ^2 \mid \bmX )\} \right] \\
	&\qquad + E\left[	\big\{E(h(\bmZ, \bmX)\mid\bmX) -  \bar E(h(\bmZ, \bmX)\mid\bmX)\big\} \{R_A \bar R_Y - \beta_0 R_A^2-  E(R_A \bar R_Y  - \beta_0 R_A ^2 \mid \bmX )\} \right] \\
		&= E\left[	\big\{h(\bmZ, \bmX)-  E(h(\bmZ, \bmX)\mid\bmX)\big\} \{R_A \bar R_Y - \beta_0 R_A^2-  E(R_A  \bar R_Y  - \beta_0 R_A ^2 \mid \bmX )\} \right] \\
		&= E\left[	\big\{h(\bmZ, \bmX)-  E(h(\bmZ, \bmX)\mid\bmX)\big\} \{R_A \bar R_Y - \beta_0 R_A^2\} \right] \\
		&= 0,
\end{align*}
where the last line is from the derivations in the first scenario.

\section{Proof of results in Section \ref{sec: semi}}
For any vector $ x $, we denote its $ \ell_1 $, $ \ell_2 $ and $ \ell_\infty $ norms  by $ \| x\|_1 $,  $ \|x\|$, and  $ \|x\|_\infty$. For any symmetric matrix $ A $, let  $tr (A)$ denote the trace, $\xi_{\min} (A)$ and $\xi_{\max} (A)$ respectively denote the smallest and largest eigenvalues.  For any matrix $ A $, let  $ \sigma_{\min} (A)$ and $ \sigma_{\max} (A)$ respectively be the smallest and largest singular values of $ A $,    $\|A\|= \sigma_{\max} (A) $ be the spectral norm, and $\|A\|_F= \{tr(A^T A)\}^{1/2}$ be the Frobenius norm.  Let $C$ be a generic positive constant that may be different in different uses.

 Let $\bfeta = (\bfeta_{p}^T , \bfeta_{np}^T)^T \in \mathbb{R}^q$ collect all the nuisance parameters, with $ \bfeta_{p}= (\bm \pi_1^T,\dots, \bm \pi_m^T, \bm\mu^T,  \bm \lambda^T)^T $ collecting the finite-dimensional parameters  and $ \bfeta_{np}= (\omega(\bmx), \theta(\bmx))^T $ collecting the infinite-dimensional functional parameters. Denote the true value of $ \bfeta $ as $ \bfeta_0=  (\bfeta_{p0}^T , \bfeta_{np0}^T)^T  $.  Let $\nabla_{\bfeta} f_i(\beta, \bfeta_0) = \partial f_i(\beta, \bfeta)/\partial \bfeta|_{\bfeta=\bfeta_0} \in \mathbb{R}^{m\times q} $ for $f\in \{g, G\}$.  Let $ \bmV= (\bmX^T, \bmZ^T)^T $. {As mentioned in the main article, if one would like to include interactions terms in the exposure and outcome models, we can simply include those interaction terms in $\bmV$ and then the following proof still goes through as long as ${\rm dim}(\bmV)=O(m)$. 
 }

For two sequences of  real numbers $a_n$ and $b_n$, we write
$a_n = O(b_n)$ if $|a_n|\leq C b_n$ for all $n$  and some $C>0$, $a_n = o(b_n)$ if $a_n/b_n\rightarrow 0$ as $n\rightarrow\infty$, $a_n= \Theta(b_n) $ if $C b_n \leq |a_n|\leq C' b_n$ for all $n$ and some $C, C'>0$. We use $\xrightarrow{p}$ to denote convergence in probability, $\xrightarrow{d}$ to denote convergence in distribution. For random variables $X$ and $Y$, we denote $X= o_p(Y)$ if $X/Y\xrightarrow{p}0$, $X= O_p(Y)$ if $X/Y$ is bounded in probability. We will use  w.p.a.1. as abbreviation for with probability approaching 1. 

Moreover, we define
\begin{align*}
   &\bar g(\beta, \bfeta) = E\{g_i(\beta, \bfeta)\},\\
   & \hat G(\bfeta) = \partial \hat g(\beta, \bfeta) /\partial \beta,\\
   & Q(\beta, \bfeta) = \bar g(\beta, \bfeta)^T \Omega (\beta, \bfeta)^{-1}\bar g(\beta, \bfeta) /2+ m/(2n), \\
   &\tilde{Q}(\beta, \hat \bfeta) = \hat g (\beta, \hat\bfeta)^T \Omega (\beta, \bfeta_0)^{-1} \hat g (\beta, \hat\bfeta)/2.   
\end{align*}


Importantly, 
with $g_i(\beta, \bfeta)$ being the  influence function, we have  $E\{\nabla_{\bfeta} g_i(\beta_0, \bfeta_0) \} =0$.   Interestingly, this property also holds for $ G_i(\bfeta) $  because $ g_i(\beta, \bfeta) $ is linear in $ \beta $, i.e., 
	\begin{align*}
		E\left\{   \frac{\partial G_i ( \bfeta) }{\partial \bfeta} \big|_{\bfeta= \bfeta_{0}} \right\}=0.
	\end{align*} 
We will show later that this property is crucial for the estimation of $ \bfeta $ to have negligible impact on the distribution of  $ \hat\beta $.


\subsection{Regularity conditions}

\begin{assum}[Kernel] \label{assump: kernel}
		(i) $ K(u) $ is bounded, $ K(u) $ is zero outside a bounded set, $ \int K(u) du=1 $ and $ \int uK(u)du=0 $;  \\
		(ii) 	 $  E(Y^8\mid \bmX) f_0(\bmX)$  and  $  E(A^8\mid \bmX)f_0(\bmX) $ are bounded, and 
		the density $ f_0(\bmx) $ is bounded away from zero in the support of $ \bmX $; \\
		(iii) The bandwidth of kernel estimator $ \sigma $ satisfies $ \sigma^4 \sqrt{nm}\rightarrow0 $ and $ \sigma^{d_x} \sqrt{n}/(\sqrt{m} \log n)\rightarrow\infty $ as $ n\rightarrow\infty $. 
\end{assum}
Assumptions \ref{assump: kernel}(i)-(ii) correspond to Assumptions 8.1 and 8.3 in \cite{Newey1994_handbook}. Assumption \ref{assump: kernel}(iii) corresponds  to the bandwidth condition imposed in Lemma 8.10 of \cite{Newey1994_handbook}.

	\begin{assum}\label{assump: moment} 
	$ \xi_{\min} (n^{-1} \sum_{i=1}^{n} \bmX_i \bmX_i^T ) \geq C$, and $ E(Y^8)< \infty, E(A^8) <\infty$. 
\end{assum}

\begin{assum}\label{assump:omega}
    There is $C>0$ such that for all $ \beta\in B $, 
    $1/C\leq \xi_{\min} (\Omega(\beta, \bfeta_0))$, $ \xi_{\max} (\Omega(\beta, \bfeta_0)) \leq C$ ,  $ \xi_{\max} (E(G_iG_i^T))\leq C$, and 
    \begin{align*}
    &\xi_{\max} \left(E\left\{\frac{\partial g_i (\beta, \bfeta_0 ) }{\partial (\bmV_i^T\bm \lambda_0)}  \frac{\partial g_i (\beta, \bfeta_0 )^T }{\partial (\bmV_i^T\bm \lambda_0)}  \right \}\right)\leq C, \quad     \xi_{\max} \left(E\left\{ \frac{\partial g_i (\beta, \bfeta_0 ) }{\partial (\bmV_i^T\bm \mu_0)}    \frac{\partial g_i (\beta, \bfeta_0 )^T }{\partial (\bmV_i^T\bm \mu_0)}  \right \}\right)\leq C.
    \end{align*}

\end{assum}
Assumption \ref{assump:omega} corresponds to Assumption 3 of \cite{Newey:2009aa}.  From Assumption \ref{assump:omega} and  \cite{tripathi1999matrix}, we immediately have $ 	E(G_i g_i(\beta, \bfeta_0)^T  ) \Omega(\beta, \bfeta_0)^{-1} 	E( g_i(\beta, \bfeta_0) G_i^T )  \leq E(G_i G_i^T ) $ and thus, $ \| E(G_i g_i(\beta, \bfeta_0)^T  )\|\leq C $. Similarly, we have that 
    \[
 \left\| E\left\{  g_i (\beta, \bfeta_0 ) \frac{\partial g_i (\beta, \bfeta_0 )^T }{\partial (\bmV_i^T\bm \lambda_0)}  \right \}\right\| \leq C, \quad     \left\| E\left\{ g_i (\beta, \bfeta_0 )   \frac{\partial g_i (\beta, \bfeta_0 )^T }{\partial (\bmV_i^T\bm \mu_0)}  \right \} \right\| \leq C. 
\]

\begin{assum} \label{assump: 4th moment}
 $ \{ E(\|g_i\|^4)+E(\|G_i\|^4)  \}m/n \rightarrow0$. 
\end{assum}
 Assumption \ref{assump: 4th moment} is from Assumption 6 of \cite{Newey:2009aa}. This imposes a stronger restriction on the growth rate of the number of moment conditions than that was imposed for consistency. If each component in $g_i$ were uniformly bounded, a sufficient condition would be $m^3/n\rightarrow 0$.  
 


\subsection{Lemmas}
We will first prove some  lemmas, which will be used in the proof of Theorem \ref{theo: GMM} in Section \ref{sec: proof of theo gmm}. The organization of the proof is illustrated in Figure \ref{fig:organization}. 

\begin{figure}[t]
	\centering
	\begin{tikzpicture} \small 
		\node[name=lemS2] {Lemma \ref{lemma: cond}};
		\node[name=lemS3,below= of lemS2]  {Lemma \ref{lemma: g bdd by beta}};
		\node[name=lemS4,below= of lemS3]  {Lemma \ref{lemma: Omega multiply}};
		\node[name=lemS5,below= of lemS4]  {Lemma \ref{lemma: sup g square}};	
		\node[name=lemS6,below= of lemS5]  {Lemma \ref{lemma: g beta, nuisance}};	
		\node[name=lemS7,below= of lemS6]  {Lemma \ref{lemma: Omega eta hat}};	
		\node[name=lemS8,right= of lemS3,xshift=1cm]  {Lemma \ref{lemma:unif Q}};	
		\node[name=lemS9,right= of lemS4,xshift=1cm, yshift=-1cm]  {Lemma \ref{lemma: Q tilde}};	
		\node[name=lemS10,right= of lemS6,xshift=1cm, yshift=-1cm]  {Lemma \ref{lemma: Q second deriv}};	
		\node[name=consistency,right= of lemS8, xshift=1cm]  {Section \ref{subsec: consistency} (Consistency)};	
		\node[name=normal,right= of lemS10, xshift=1cm]  {Section \ref{subsec: normal} (Asymptotic Normality)};	
		\draw[->] (lemS2) to (consistency);
		\draw[->] (lemS3) to (lemS8);
		\draw[->] (lemS4) to (lemS8);
		\draw[->] (lemS5) to (lemS8);
		\draw[->] (lemS6) to (lemS8);
		\draw[->] (lemS7) to (lemS8);
		\draw[->] (lemS6) to (lemS9);
		\draw[->] (lemS7) to (lemS9);
		\draw[->] (lemS6) to (lemS10);
		\draw[->] (lemS7) to (lemS10);
		\draw[->] (lemS8) to (consistency);
		\draw[->] (lemS9) to (normal);
		\draw[->] (lemS10) to (normal);
		\draw[->] (lemS6) to (normal);
		\draw[->] (consistency) to (normal);
	\end{tikzpicture}
	\caption{Organization of the proof.}
	\label{fig:organization}
\end{figure}
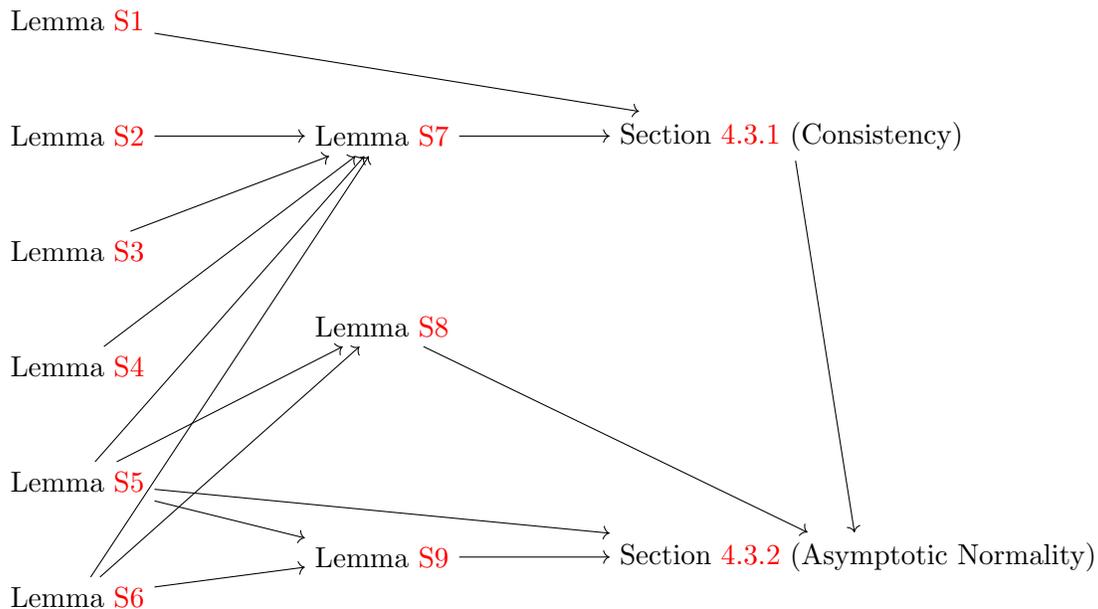

The first lemma is important for global identification of $\beta_0$. 
\begin{lemma}\label{lemma: cond} Under Assumptions \ref{assump: many weak moments} and \ref{assump:omega}, \\
    (i) There is $C>0$ with $|\beta- \beta_0|\leq C\sqrt{n} \|\bar g(\beta, \bfeta_0)\| /\mu_n$ for all $\beta\in B$;\\
    (ii) There is $C>0$ and $\hat M=O_p(1)$ such that $|\beta-\beta_0|\leq C\sqrt{n}\|\hat g(\beta, \bfeta_0) \| /\mu_n+ \hat M$ for all $\beta\in B$.
\end{lemma}
\begin{proof}
(i) Note that with $g_i(\beta, \bfeta_0)$ defined as in \eqref{eq: IF}, it is true that 
\[
G=E\left[(\bmZ - E(\bmZ\mid \bmX)  ) \{-R_A^2 +  E(R_A^2\mid \bmX)  \}\right].
\]
Moreover, as $\bar g(\beta_0, \bfeta_0)= 0 $, 
\begin{align}
\bar g(\beta, \bfeta_0)=\bar g(\beta, \bfeta_0)- \bar g(\beta_0, \bfeta_0)= (\beta- \beta_0) G.\label{eq: g bar}    
\end{align}
Also, under Assumptions \ref{assump: many weak moments} and  \ref{assump:omega}, we have that $G^T G = \Theta(\mu_n^2/n)$, and thus
\begin{align*}
    \sqrt{n} \|\bar g(\beta, \bfeta_0) \| /\mu_n= \sqrt{n} |\beta- \beta_0| (G^T G)^{1/2} /\mu_n = \Theta(|\beta-\beta_0|).
\end{align*}
This concludes the proof. \\
(ii) Note that 
\begin{align}
  &\mu_n^{-1} \sqrt{n}  \hat g(\beta, \bfeta_0) \label{eq: g hat}\\
  &=    \mu_n^{-1} \sqrt{n} \hat g(\beta_0, \bfeta_0) + \mu_n^{-1} \sqrt{n}(\beta- \beta_0) \frac{1}{n} \sum_{i=1}^n G_i\nonumber\\
  & =    \mu_n^{-1} \sqrt{n} \hat g(\beta_0, \bfeta_0) + \mu_n^{-1} \sqrt{n}(\beta- \beta_0) \frac{1}{n} \sum_{i=1}^n (G_i-G) + \mu_n^{-1}\sqrt{n} (\beta-\beta_0) G \nonumber.
\end{align}
Next, we have that $\sqrt{n/m} \|\hat g (\beta_0, \bfeta_0)\|=O_p(1)$ from the Markov inequality and 
\begin{align*}
&m^{-1} n E\{\|\hat g (\beta_0, \bfeta_0) \|^2\} = m^{-1} n E\big[ tr\{ \hat g (\beta_0, \bfeta_0)^T  g (\beta_0, \bfeta_0)\}\big] = m^{-1} n E\left[ \frac{1}{n^2}\sum_{i, j=1}^n g_i^T g_j  \right] \\\
&= m^{-1} n E\left[ \frac{1}{n^2}\sum_{i=1}^n g_i^T g_i \right] =m^{-1}  E\left[  g_i^T g_i \right]    = m^{-1}  E\left[ tr( g_i g_i^T) \right]   = m^{-1} tr( \Omega) \leq C,
\end{align*}
where the last inequality is from Assumption \ref{assump:omega}. We thus have $ \|\hat g (\beta_0, \bfeta_0)\|= O_p(\sqrt{m/n})$ and also $\|\hat g (\beta_0, \bfeta_0)\|= O_p(\mu_n/\sqrt{n})$ as $m/\mu_n^2\leq C$ from Assumption \ref{assump: many weak moments}. 

We can similarly show that $ \| n^{-1} \sum_{i=1}^n (G_i-G)\|=O_p(\mu_n/\sqrt{n})$ because $ \mu_n^{-2} E\{(G_i-G)^T (G_i-G)\} = \mu_n^{-2} E(G_i^TG_i) - \mu_n^{-2}  G^T G \leq C$ from Assumption \ref{assump: many weak moments} and Assumption \ref{assump:omega}. 

Moreover, from Assumptions \ref{assump: many weak moments}, $\|\mu_n^{-1} \sqrt{n} (\beta-\beta_0) G\|=\mu_n^{-1} \sqrt{n} |\beta-\beta_0| \|G\| \geq C |\beta-\beta_0| $. 

As $B$ is compact, we can define 
\begin{align}
    \hat M=  \mu_n^{-1} \sqrt{n} \sup_{\beta\in B} \bigg\|\hat g(\beta_0, \bfeta_0) + (\beta- \beta_0) \frac{1}{n} \sum_{i=1}^n (G_i-G) \bigg\| = O_p(1).
\end{align}
Then, by triangle inequality, it follows that for all $\beta\in B$,
\begin{align*}
  C  |\beta-\beta_0| \leq C\|\mu_n^{-1} \sqrt{n} (\beta-\beta_0) G\|\leq \mu_n^{-1}\sqrt{n} \| \hat g(\beta, \bfeta_0) \|+\hat M.
\end{align*}
This concludes the proof. 
\end{proof}

\begin{lemma} \label{lemma: g bdd by beta}
Under Assumptions \ref{assump: many weak moments} and \ref{assump:omega}, there is $C$ and $\hat M=O_p(1)$ such that for all $\beta', \beta\in B$,\\
(i)   $\sqrt{n} \|\bar g(\beta', \bfeta_0) - \bar g(\beta, \bfeta_0)\|/\mu_n\leq C|\beta'-\beta|$;\\
(ii)  $\sqrt{n} \|\hat g(\beta', \bfeta_0) - \hat g(\beta, \bfeta_0)\|/\mu_n\leq \hat M |\beta'-\beta|$.
\end{lemma}
\begin{proof}
(i) From Assumptions \ref{assump: many weak moments} and  \ref{assump:omega}, we have 
\begin{align}
    \mu_n^{-1} \sqrt{n}\|\bar g(\beta', \bfeta_0) - \bar g(\beta, \bfeta_0) \| = \mu_n^{-1} \sqrt{n} \| (\beta'-\beta) G\|= \mu_n^{-1}\sqrt{n} |\beta'-\beta| (G^T G)^{1/2} \leq C|\beta'-\beta|. \nonumber
\end{align}
(ii) From \eqref{eq: g hat}, we have that 
\begin{align*}
    \mu_n^{-1} \sqrt{n} \|\hat g(\beta', \bfeta_0) - \hat g(\beta, \bfeta_0)\| = \mu_n^{-1} \sqrt{n} |\beta'-\beta| \bigg\|\frac{1}{n}\sum_{i=1}^n G_i\bigg\|.
\end{align*}
We  define $\hat M= \mu_n^{-1}\sqrt{n} \|\frac{1}{n}\sum_{i=1}^n G_i\|$, which is
\begin{align*}
 &\mu_n^{-1}\sqrt{n} \bigg\|\frac{1}{n}\sum_{i=1}^n G_i\bigg\| =  \mu_n^{-1}\sqrt{n} \bigg\|\frac{1}{n}\sum_{i=1}^n (G_i- G)+ G\bigg\| \\
 &\leq \mu_n^{-1} \sqrt{n}\bigg\|\frac{1}{n}\sum_{i=1}^n (G_i- G)\bigg\| + 
 \mu_n^{-1} \sqrt{n}\| G\| = O_p(1) + O(1)= O_p(1),
\end{align*}
where $\mu_n^{-1}\sqrt{n}\| \frac{1}{n}\sum_{i=1}^n (G_i- G)\|=O_p(1)$ is established in the proof of Lemma \ref{lemma: cond}. This concludes the proof. 
\end{proof}

\begin{lemma} \label{lemma: Omega multiply}Under Assumption \ref{assump:omega}, 
	$|a^T \{\Omega(\beta', \bfeta_0)-\Omega(\beta, \bfeta_0)\}b| \leq C\|a\|\|b\||\beta'-\beta|  $ for all $a, b \in \mathbb{R}^m$, $\beta', \beta\in B$. 
\end{lemma}
\begin{proof}
	Using $ g_i(\beta', \bfeta_0)= g_i(\beta,\bfeta_0) + (\beta'- \beta)  G_i $, we have that 
	\begin{align*}
		&	|a^T \{\Omega(\beta', \bfeta_0)-\Omega(\beta, \bfeta_0)\}b|  \\
		&= \left|  (\beta'-\beta)^2 a^T E(G_iG_i^T) b + (\beta'-\beta) a^TE(G_ig_i(\beta, \bfeta_0)^T) b +(\beta'-\beta) a^TE(g_i(\beta, \bfeta_0) G_i^T) b  \right| \\
		&\leq C|\beta'-\beta| a^T E(G_iG_i^T) b  + |\beta'-\beta| | a^T E(G_ig_i(\beta, \bfeta_0)^T) b | +  |\beta'-\beta| | a^TE(g_i(\beta, \bfeta_0) G_i^T)  b | \\
		&\leq C|\beta'-\beta|  \|a\| \| E(G_iG_i^T) \| \|b\| +2 |\beta'-\beta| \| a\| \| E(G_ig_i(\beta, \bfeta_0)^T) \| \| b \| \\
		&\leq C |\beta'-\beta|  \|a\| \|b\|,
	\end{align*}
	where 	the last line is because  $ \| E(G_iG_i^T) \| \leq C $ and $\| E(G_ig_i(\beta, \bfeta_0)^T) \|\leq C $ from	  Assumption \ref{assump:omega}. 
\end{proof}

\begin{lemma} \label{lemma: sup g square}
	Under Assumption \ref{assump: moment} and the boundedness of $ \bmZ $ and $ \bmX $, there is $ C>0 $ such that 
	\[
	 	\sup_{\beta\in B} E[\{g_i(\beta, \bfeta_0)^T g_i(\beta, \bfeta_0)\}^2]\leq Cm^2 .
	\]
\end{lemma}
\begin{proof}
	From straightforward calculation, and recall that $ \Delta_i = (A_i -E(A_i\mid \bmV_i)) (Y_i -E(Y_i\mid \bmV_i) ) - \beta (A_i -E(A_i\mid \bmV_i))^2  $, 	we have  
	\begin{align*}
		&E[\{g_i(\beta, \bfeta_0)^T g_i(\beta, \bfeta_0)\}^2] \\
		&= E\left[ \left\{\sum_{j=1}^m (Z_{ij}- \bmX_i^T \bm\pi_{0j} )^2 ( \Delta_i - E(\Delta_i \mid \bmX_i))^2\right\}^2 \right] \\
		&\leq C m^2 E\left[  \{ \Delta_i - E(\Delta_i \mid \bmX_i)\}^4 \right] \\
		&\leq C m^2 E\left[   \Delta_i^4  + \{E(\Delta_i \mid \bmX_i)\}^4 \right]. 
	\end{align*}
Note that $ E( \Delta_i^4)<\infty  $ because 
\begin{align*}
E( \Delta_i^4) &\leq CE\left\{  (A_i -E(A_i\mid \bmV_i))^4  (Y_i -E(Y_i\mid \bmV_i) )^4  \right\} + C \beta^4 E\{(A_i -E(A_i\mid \bmV_i))^8\}\\
&\leq C \sqrt{ E\{(A_i -E(A_i\mid \bmV_i))^8 \}  E\{ (Y_i -E(Y_i\mid \bmV_i) )^8\} }  + C \beta^4 E\{(A_i -E(A_i\mid \bmV_i))^8\} \\
&\leq C \sqrt{ E\{ A_i^8 + (E(A_i\mid \bmV_i))^8 \} E\{ Y_i^8 + (E(Y_i\mid \bmV_i))^8 \} }  + C \beta^4 E\{ A_i^8 + (E(A_i\mid \bmV_i))^8 \}\\
&\leq \infty, 
\end{align*}
where the last line is because  $ E (A_i^8)<\infty $ and $ E (Y_i^8)<\infty $ from Assumption \ref{assump: moment}, and $ E \{ (E(A_i\mid \bmV_i))^8 \}<\infty $ and $ E \{  (E(Y_i\mid \bmV_i))^8 \}<\infty$ from Jensen's inequality for conditional expectation. Another use of Jensen's inequality gives us $ E\left[  \{E(\Delta_i \mid \bmX_i)\}^4 \right]<\infty  $.  Finally, using the  compactness of $ B $ concludes the proof.  
\end{proof}

\begin{lemma} \label{lemma: g beta, nuisance}
Under Assumptions \ref{assump: linear}-\ref{assump:omega}, when $ m^2/n\rightarrow 0 $, \\
(i)  $\|\hat g(\beta_0, \hat\bfeta) -\hat g(\beta_0, \bfeta_0) \|= o_p(n^{-1/2}) + O_p(m\log m/n)$;\\
(ii) $\| \hat G(\hat \bfeta) -  \hat G( \bfeta_0)\|= o_p(n^{-1/2}) + O_p(m\log m/n)$;\\
(iii) $\sup_{\beta\in B} \|\hat g(\beta, \hat\bfeta) -\hat g(\beta, \bfeta_0) \|= o_p(n^{-1/2}) + O_p(m\log m/n)$.
\end{lemma}
\begin{proof}
(i) Let 
\begin{align*}
	\bm D_i=  \left( 
	\begin{array}{ccc}
		\bmX_i^T &&  \\
		&	\ddots&\\
		&&	 \bmX_i^T  \\
	\end{array}
	\right)\in \mathbb{R}^{m\times (md_x)}, ~\bm \pi_0 = 
	\left( 
	\begin{array}{c}
		\bm\pi_{10}  \\
		\dots\\
		\bm\pi_{m0} \\
	\end{array}
	\right)\in \mathbb{R}^{ md_x},  ~ \text{and} ~   
	\hat{\bm \pi} =  \left( 
	\begin{array}{c}
		\hat{\bm\pi}_1  \\
		\dots\\
		\hat{\bm\pi}_m \\
	\end{array}
	\right)\in \mathbb{R}^{ md_x}.
\end{align*}

Under Assumption \ref{assump: linear},  we can write
\begin{align*}
	g^{IF} (\O; \beta, \bfeta_0)& =\big( \bmZ-\bm D \bm\pi_0 \big)  \bigg[ \underbrace{( A- \bmV^T \bm\mu_0 )\{ Y- \bmV^T \bm{\lambda}_{0} - \beta( A- \bmV^T \bm \mu_0 )  \}}_{\Delta} - \omega_0(\bmX ) +\beta \theta_0(\bmX)\bigg],\\
		g^{IF} (\O; \beta, \hat \bfeta)& =\big( \bmZ-\bm D \hat{\bm\pi}\big) 
	\left[ ( A- \bmV^T \hat{\bm\mu} )\{ Y- \bmV^T \hat{\bm{\lambda}} - \beta( A- \bmV^T \hat{\bm \mu} )  \} - \hat \omega(\bmX; \hat{\bm \mu}, \hat{\bm \lambda} ) +\beta\hat \theta(\bmX; \hat{\bm \mu})\right],
\end{align*}
and an intermediate term
\begin{align*}
	g^{IF} (\O; \beta, \hat \bfeta_{p},  \bfeta_{np	0})& =\big( \bmZ-\bm D \hat{\bm\pi}\big)  \left[ ( A- \bmV^T \hat{\bm\mu} )\{ Y- \bmV^T \hat{\bm{\lambda}} - \beta( A- \bmV^T \hat{\bm \mu} )  \} -  \omega_0(\bmX ) +\beta \theta_0(\bmX)\right].	
\end{align*}
Hence, 
\begin{align}
&	g^{IF} (\O; \beta, \hat\bfeta)  - 	g^{IF} (\O; \beta, \bfeta_0) \nonumber\\
	&\qquad = 	g^{IF} (\O; \beta, \hat\bfeta)  - 		g^{IF} (\O; \beta, \hat \bfeta_{p},  \bfeta_{np0}) + 		g^{IF} (\O; \beta, \hat \bfeta_{p},  \bfeta_{np0})- g^{IF} (\O; \beta, \bfeta_0). \nonumber
\end{align}
Let $  \hat 	g( \beta, \hat \bfeta_{p},  \bfeta_{np0}) = n^{-1}\sum_{i=1}^{n} 	g^{IF} (\O_i; \beta, \hat \bfeta_{p},  \bfeta_{np0}) $.  We will show that when $ m^2/n\rightarrow 0 $, 
\begin{itemize}
	\item[(a)] $ \| \hat 	g( \beta_0, \hat \bfeta_{p},  \bfeta_{np0})  - \hat g(\beta_0, \bfeta_0) \| =  O_p(m\log m/n ),  $
	\item[(b)]  $ \|   \hat g(\beta_0,  \hat \bfeta) - \hat 	g( \beta_0, \hat \bfeta_{p},  \bfeta_{np0})  \|= o_p(n^{-1/2})  + o_p( m\sqrt{\log m}/n)   $. 
\end{itemize}

For part (a), write $ 	 \hat 	g( \beta_0, \hat \bfeta_{p},  \bfeta_{np0})  -\hat g(\beta_0, \bfeta_0) = \nabla_{\bfeta_{p}} \hat g(\beta_0, \bfeta_0) (\hat \bfeta_{p}- \bfeta_0) + \text{Rem}_a   $, where 
\begin{align*}
	\text{Rem}_a&= \frac{1}{n} \sum_{i=1}^{n}  \big( \bmZ_i-\bm D_i \hat{ \bm\pi} \big)  \bmV_i^T(\hat{\bm\mu} - \bm \mu_0 )  \{  \bmV_i^T(\hat{\bm\lambda} - \bm\lambda_0 )  -  \beta_0 \bmV_i^T(\hat{\bm\mu} - \bm \mu_0 ) \}
	\\
	&+ 	\frac{1}{n} \sum_{i=1}^{n} \bm D_i (\hat{ \bm\pi} - \bm\pi_0) (A_i - \bmV_i^T \bm\mu_0 ) \bmV_i^T (\hat{\bm \lambda}- \bm\lambda_0 ) \\
	&-	\frac{1}{n} \sum_{i=1}^{n}\bm D_i (\hat{ \bm\pi} - \bm\pi_0)   (- Y_i + 2\beta A_i + \bmV_i^T \bm\lambda_0 - 2\beta\bmV_i^T \bm\mu_0  ) \bmV_i^T (\hat{\bm\mu}- \bm\mu_0 ),
\end{align*}
and
\begin{align*}
	&  \nabla_{\bfeta_{p}} \hat g(\beta_0, \bfeta_0)(\hat \bfeta_{p}- \bfeta_{p0})  \\
	& = - \frac{1}{n}  \sum_{i=1}^{n}	
	\bm D_i (\hat{\bm \pi} - \bm \pi_0) \Delta_i   \\
	  &\qquad +  \frac{1}{n} \sum_{i=1}^{n}\big( \bmZ_i-\bm D_i \bm\pi_0 \big)   ( - Y_i + 2\beta A_i +\bmV_i^T \bm{\lambda}_{0} - 2\beta \bmV_i^T \bm{\mu}_{0}  )    \bmV_i^T ( \hat{\bm{\mu}} - \bm{\mu}_{0})   \\
	  &\qquad- \frac{1}{n} \sum_{i=1}^{n}\big( \bmZ_i-\bm D_i \bm\pi_0 \big) (A_i - \bmV_i^T \bm{\mu}_{0} ) \bmV_i^T ( \hat{\bm{\lambda}}- \bm{\lambda}_{0}  ) 
	 \\
	  &= - A_1+ A_2-  A_3.
\end{align*}

For $ A_1$, from the Strong Schwartz Matrix Inequality\footnote{Strong Schwartz Matrix Inequality: For any conformable matrices $ A $ and $ B $, $ \|AB\|_F\leq \|A\| \|B\|_F $, where $ \|C\|_F = \sqrt{tr(C^TC)}$ is the Frobenius norm of matrix $ C $. Here, with $ A= \frac{1}{n} \sum_{i=1}^{n}  
	\Delta_i   \bm D_i , B= \hat{\bm \pi} - \bm \pi_0	  $, both $ AB $ and $ B $ are column vectors so their Frobenius norms equal the spectral norms. },
\begin{align*}
	\|A_1 \| &= \left\| \frac{1}{n} \sum_{i=1}^{n}  
	 \Delta_i   \bm D_i (\hat{\bm \pi} - \bm \pi_0)
	 \right\| \\
	 & \leq 
	  \left\| \frac{1}{n} \sum_{i=1}^{n}  
	 \Delta_i   \bm D_i 	 \right\| 
 	\left\|
(\hat{\bm \pi} - \bm \pi_0)
	 \right\|\\
	 &:=  \|A_{11} \| 	\left\|
	 (\hat{\bm \pi} - \bm \pi_0)
	 \right\|.
\end{align*}
Since $ A_{11} A_{11}^T  $ is a diagonal matrix, with diagonal elements all equal to 
$n^{-2} \left( \sum_{i=1}^{n} \Delta_i \bmX_i \right)^T  \left( \sum_{i=1}^{n} \Delta_i \bmX_i \right) $, the spectral norm of $ A_{11} $ is 
 $ \|A_{11}\| =  \sqrt{ n^{-2} \left( \sum_{i=1}^{n} \Delta_i \bmX_i \right)^T  \left( \sum_{i=1}^{n} \Delta_i \bmX_i \right)  }= \|n^{-1}  \sum_{i=1}^{n} \Delta_i \bmX_i   \| = O_p(n^{-1/2})$ because $ E(\Delta_i \bmX_i)= 0 $ and $ d_x<\infty $. 
 
 Next, we analyze $ 		\hat{\bm\pi}_j-\bm \pi_{j0}  $
 for $ j=1,\dots, m $. As $ 	\hat{\bm\pi}_j = \arg\min_{\bm \pi_j } \| \tilde \bmZ_j - \tilde{\bm X} \bm\pi_j \|^2 $, where $  \tilde \bmZ_j = (Z_{1j}, \dots, Z_{nj})^T \in \mathbb{R}^n, \tilde{\bmX}^T = (\bmX_1,\dots,\bmX_n )  \in \mathbb{R}^{ d_x\times n} $, thus 
\begin{align*}
 \|\tilde \bmZ_j - \tilde{\bm X} \hat{\bm\pi}_j \|^2  \leq  \| \tilde \bmZ_j - \tilde{\bm X} {\bm\pi}_{j0}\|^2.
\end{align*} 
Using H\"{o}lder's inequality and Cauchy-Schwartz inequality, we have 
\begin{align}
&\|\tilde\bmX (\hat{\bm \pi}_j -{\bm \pi}_{j0} )\|^2 \nonumber\\
&\leq  2|(\tilde \bmZ_j - \tilde\bmX \bm\pi_{j0})^T \tilde\bmX(\hat{\bm\pi}_j - \bm\pi_{j0}) | \nonumber \\
&\leq  2\| \tilde\bmX^T (\tilde\bmZ_j -\tilde \bmX \bm\pi_{j0})\|_{\infty} \|\hat{\bm\pi}_j - \bm\pi_{j0} \|_1 \nonumber \\
&\leq 2\sqrt{d_x}\| \tilde \bmX^T (\tilde\bmZ_j -\tilde \bmX \bm\pi_{j0})\|_{\infty} \|\hat{\bm\pi}_j - \bm\pi_{j0} \|.  \label{eq: pi}
\end{align} 
On the other hand, with $ \sigma_{\min}(\tilde \bmX) $ being the minimum singular value of $\tilde\bmX $, we have 
 $ \|\tilde \bmX (\hat{\bm \pi}_j -{\bm \pi}_{j0} )\|^2  \geq \sigma^2_{\min} (\tilde\bmX) \| \hat{\bm \pi}_j -{\bm \pi}_{j0} \|^2\geq Cn  \| \hat{\bm \pi}_j -{\bm \pi}_{j0} \|^2$, where the last inequality is because $ \xi_{\min} (n^{-1}\tilde \bmX^T \tilde \bmX) \geq C$  from Assumption \ref{assump: moment}. In addition, from applying Lemma 8 in \cite{chernozhukov2015comparison} and  the boundedness of $ \bmZ $ and $ \bmX $, we have 
 \begin{align*}
E\left[ \max_{j=1,\dots, m}  \| \tilde \bmX^T (\tilde \bmZ_j - \tilde \bmX \bm \pi_{j0} ) \|_{\infty}  \right] \leq C  \sqrt{n \log m}.
 \end{align*}
 Then from Markov inequality, we know that  $ \max_{j=1,\dots, m} \| \tilde \bmX^T (\tilde \bmZ_j - \tilde \bmX \bm \pi_{j0} ) \|_{\infty}= O_p( \sqrt{n \log m} )$.  Combining the above derivations, 
we have for $ j=1,\dots, m $, 
\begin{align*}
	\|\hat{\bm \pi}_j -{\bm \pi}_{j0} \|  \leq \frac{C }{n } \| (\bm Z_j - \bmX \bm \pi_{j0} )^T \bm X \|_{\infty}   \leq  \frac{C }{n } \max_{j=1,\dots, m}  \| (\bm Z_j - \bmX \bm \pi_{j0} )^T \bm X \|_{\infty},
\end{align*}
and thus, 
\begin{align*}
&\|	 \hat{\bm \pi} - \bm \pi_0\|^2= \sum_{j=1}^{m}	\|\hat{\bm \pi}_j -{\bm \pi}_{j0} \|^2  \leq  \frac{Cm}{n^2} \{\max_{j=1,\dots,m } \| (\bm Z_j - \bmX \bm \pi_{j0} )^T \bm X \|_{\infty} \}^2 = O_p\left(  \frac{m\log m}{n}  \right). 
\end{align*}
This concludes the proof of $ \|A_1\| = O_p (\sqrt{m\log m}/n ) $.

For $ A_2 $, again using the Strong Schwartz Matrix Inequality, 
\begin{align*}
	\|A_2\| & \leq \left\|  \frac{1}{n} \sum_{i=1}^{n}\big( \bmZ_i-\bm D_i \bm\pi_0 \big)  \bmV_i^T M_i   \right\|\| \hat{\bm{\mu}} - \bm{\mu}_{0}\| \\
	&:= \|A_{21}\| \| \hat{\bm{\mu}} - \bm{\mu}_{0}\|, 
\end{align*}
where $  M_i =   - Y_i + 2\beta A_i +\bmV_i^T \bm{\lambda}_{0} - 2\beta \bmV_i^T \bm{\mu}_{0}    $. We use \citet[Theorem 1]{tropp2015expected} and Markov inequality to construct a bound for $ \|A_{21}\| $. From Theorem 1 of \cite{tropp2015expected}, we need to calculate the matrix variance parameter $ v= \| E ( A_{21} A_{21}^T )\| $ and the large deviation parameter $ L= \{ E \max_i \|S_i\|^2 \}^{1/2} $, where $  S_i = n^{-1}\big( \bmZ_i-\bm D_i \bm\pi_0 \big)  \bmV_i^T M_i $.

 Note that 
 \begin{align*}
 	v& = n \| E( S_i S_i^T) \|\\
 	&=  n^{-1} \left\|  E\left\{ \big( \bmZ_i-\bm D_i \bm\pi_0 \big)  
 	\big( \bmZ_i-\bm D_i \bm\pi_0 \big)^T M_i^2  \bmV_i^T\bmV_i   \right\}   \right\| \\
 	&\leq C n^{-1} m  \left\|  E\left\{ 
 \big( \bmZ_i-\bm D_i \bm\pi_0 \big)  \big( \bmZ_i-\bm D_i \bm\pi_0 \big)^T M_i^2   \right\}   \right\|   \\
 	& \leq C n^{-1} m
 \end{align*}
 where the third line is from $ \bmV_i $ being bounded, the last line is from Assumption \ref{assump:omega}.  Also note that
\begin{align*}
& \left\| \big( \bmZ_i-\bm D_i \bm\pi_0 \big)  \bmV_i^T \right\|^2  = \xi_{\max} \left\{ \big( \bmZ_i-\bm D_i \bm\pi_0 \big)  \bmV_i^T  \bmV_i   \big( \bmZ_i-\bm D_i \bm\pi_0 \big) ^T \right\}\\
&\leq tr \left\{ \big( \bmZ_i-\bm D_i \bm\pi_0 \big)  \bmV_i^T  \bmV_i    \big( \bmZ_i-\bm D_i \bm\pi_0 \big) ^T \right\}  \leq Cm \big( \bmZ_i-\bm D_i \bm\pi_0 \big) ^T \big( \bmZ_i-\bm D_i \bm\pi_0 \big)  \\
&= Cm \sum_{j=1}^m (Z_{ij } - \bmX_i^T \bm\pi_{j0} )^2\leq Cm^2,
\end{align*}
and the last two inequalities are from $\bmV_i$ being bounded. This implies that 
\begin{align*}
	L^2&= E\max_i \|S_i\|^2 =n^{-2}   E\left \{ \max_i  \left\| \big( \bmZ_i-\bm D_i \bm\pi_0 \big)   \bmV_i^T \right\|^2 M_i^2 \right \} \\
	& \leq C  n^{-2} m^2 E(\max_i M_i^2) \leq C  n^{-2} m^2 E \{\max_i (Y_i^2+A_i^2) \}\leq Cn^{-2} m^2 (\log n)^2,
\end{align*}
where the last inequality is from $E \{\max_i (Y_i^2+A_i^2) \} \leq E \{\max_i (Y_i^2)\} + E\{ \max_i(A_i^2) \}   \leq C (\log n)^2$, which uses \citet[Lemma 14.12]{buhlmann2011statistics} with $n=1$ and $m=2$.

From the above analysis and Theorem 1 of \cite{tropp2015expected}, we know the matrix variance parameter $ v $ is driving the upper bound and 
\begin{align*}
	E \|A_{21} \|^2 \leq C\log (m)  v \leq C \frac{m \log m}{n}. 
\end{align*}
Then, from Markov inequality, we know that  $  \|A_{21}\|^2= O_p( m\log m /n )$.  Using the same argument as the proof of $ \|\hat{\bm\pi}_j -   {\bm\pi}_{j0} \| $  in \eqref{eq: pi}, we can show that $ \|\hat{\bm \mu}- \bm\mu_0 \| = O_p( \sqrt{m\log m/n} )$. Thus,  $\| A_{2} \| =  O_p( m \log m/n ) $.  The last term $ A_3 $ is bounded using the same argument. 
 
 Finally, as the remainder term $ \text{Rem}_a $ consists of higher order terms, we can use the above arguments to show that  $ \text{Rem}_a $  is negligible. This concludes the proof of part (a). 

For part (b), we want to show that   $ \| \hat g(\beta_0,  \hat \bfeta) -   \hat 	g( \beta_0, \hat \bfeta_{p},  \bfeta_{np0})  \|=o_p(n^{-1/2}) + o_p(m\sqrt{\log m} /n)$. Note that 
\begin{align*}
& \hat g(\beta_0,  \hat \bfeta) - \hat 	g( \beta_0, \hat \bfeta_{p},  \bfeta_{np0})  \\
 &= \frac{1}{n} \sum_{i=1}^{n}    \big( \bmZ_i-\bm D_i \bm\pi_0 \big)  \left[ - \{\hat\omega(\bmX_i; \hat{\bm \mu}, \hat{\bm \lambda}) - \omega_0(\bmX_i) \} + \beta_0 \{ \hat\theta(\bmX_i; \hat{\bm\mu})- \theta_0(\bmX_i)\} \right] +\text{Rem}_b\\
 &= \frac{1}{n} \sum_{i=1}^{n}    \big( \bmZ_i-\bm D_i \bm\pi_0 \big)  \left[ - \{\hat \omega(\bmX_i; {\bm \mu}_0, {\bm \lambda}_0 ) - \omega_0(\bmX_i; \bm\mu_0,{\bm \lambda}_0) \} + \beta_0 \{ \hat\theta(\bmX_i; {\bm \mu}_0) -  \theta_0(\bmX_i; {\bm \mu}_0)\} \right]  \\
&\qquad + \frac{1}{n} \sum_{i=1}^{n}   \big( \bmZ_i-\bm D_i \bm\pi_0 \big)  \left[ - \{\hat\omega(\bmX_i; \hat{\bm \mu}, \hat{\bm \lambda}) - \hat\omega(\bmX_i; {\bm \mu}_0, {\bm \lambda}_0) \} + \beta_0 \{ \hat\theta(\bmX_i; \hat{\bm\mu})- \hat\theta(\bmX_i; {\bm \mu}_0)\} \right] \\
 &\qquad+\text{Rem}_b\\
 &:= B_1+B_2+\text{Rem}_b 
\end{align*}

For $ B_1 $, we  closely follow the steps in \citet[Chapter 8]{Newey1994_handbook} to  show that 
\begin{align}
\left\|	\frac1n \sum_{i=1}^n  \big( \bmZ_i-\bm D_i \bm\pi_0 \big)  \{\hat\theta(\bmX_i; \bm \mu_0 )   - \theta_0(\bmX_i; \bm \mu_0 ) \} \right\|= o_p(n^{-1/2} )+ O_p(\sqrt{m}/n). \label{eq: kernel}
\end{align}
The other term in $ B_1 $ can be shown in the same way.  First, we  rewrite the nuisance parameter as $ \hat\theta(\bm x; \bm \mu_0 ) =\hat \gamma_2 (\bmx; \bm \mu_0 ) /\hat\gamma_1(\bmx)  $  and $ \theta_0(\bm x; \bm \mu_0 ) = \gamma_{20} (\bmx; \bm \mu_0 ) /\gamma_{10}(\bmx)  $, where 
 \begin{align*}
 	&\hat\gamma_2(\bmx; \bm \mu_0) =\frac1n \sum_{k=1}^{n}  (A_k - \bmV_k^T \bm\mu_0)^2K_{\sigma} (\bmx - \bmX_k), 	\hat\gamma_1(\bmx ) =\frac1n \sum_{k=1}^{n} K_{\sigma} (\bmx - \bmX_k), \\
 	&\gamma_{20} (\bmx; \bm \mu_0 ) = f_0(\bm x) E\{ (A - \bmV^T \bm\mu_0)^2\mid \bm X= \bmx \}, \gamma_{10}(\bm x) = f_0(\bm x),
  \end{align*}
where $ K_\sigma (\bmx ) = \sigma^{-d_x} K(\bmx/\sigma)$, $ d_x $ is the dimension of $ \bmX $, $ K(u) $ is a function satisfying Assumption \ref{assump: kernel},  $ \sigma $ is a bandwidth term, and $ f_0(\bmx) $ is the density of $ \bmx $. Then, we obtain the linearization (an functional analogue of  Taylor expansion) of $  \big( \bmZ-\bm D \bm\pi_0 \big)  \{\hat\theta(\bmX; \bm \mu_0 )   - \theta_0(\bmX,\bm \mu_0 ) \}$  in \eqref{eq: kernel} as  $ L(\bm O;\hat  \gamma- \gamma_0, \bm \mu_0) = L(\bm O;\hat  \gamma, \bm \mu_0) - L(\bm O; \gamma_0, \bm \mu_0) $,  where
\begin{align*}
 L(\bm O; \gamma, \bm \mu_0)  = (\bm Z - \bm D\bm\pi_0 ) f_0(\bmX)^{-1} \{-  \theta_0(\bmX; \bm \mu_0) \gamma_1(\bmX) + \gamma_2(\bmX; \bm \mu_0) \},
\end{align*}
and  $ \gamma =( \gamma_1(\bmx), \gamma_2(\bmx; \bm\mu_0) )$, $ \gamma_0 =( \gamma_{10}(\bmx), \gamma_{20}(\bmx; \bm\mu_0) )$.  With $ f_0(\bmx)$ and $ \gamma_1(\bmx) $ being bounded away from zero, $ \gamma_{20} (\bmx; \bm \mu_0) $ being bounded,  the remainder term from the linearization  satisfies
\begin{align}
	&\left\| (\bmZ - \bm D \bm \pi_0) \{\hat \gamma_2(\bmX; \bm \mu_0)/ \hat \gamma_1(\bmX) - \theta_0(\bmX; \bm \mu_0)\}-   L(\bm O; \hat \gamma- \gamma_0 , \bm \mu_0)  \right\| \nonumber\\
	&\leq \| \bmZ - \bm D \bm\pi_0 \| \sup_{\bmx\in \text{supp} (\bmX)} \|\hat \gamma(\bmx; \bm \mu_0)- \gamma_0 (\bm x)\|^2 \leq C \sqrt{m} \sup_{\bmx\in \text{supp} (\bmX)} \|\hat \gamma(\bmx; \bm \mu_0)- \gamma_0 (\bm x)\|^2 , \nonumber
\end{align}
where $ \text{supp} (\bmX) $ is the support of $ \bmX $. The above term is  $ o_p(n^{-1/2}) $ from  Lemma 8.10 of \cite{Newey1994_handbook} and  choosing the bandwidth $ \sigma $ to satisfy $ \sigma^4 \sqrt{nm}\rightarrow0 $ and $ \sigma^{d_x} \sqrt{n}/(\sqrt{m} \log n)\rightarrow\infty $ as $ n\rightarrow\infty $  in Assumption \ref{assump: kernel}. This means that 
\begin{align*}
	\frac1n \left\| \sum_{i=1}^{n} (\bmZ_i - \bm D_i \bm \pi_0) \{\hat \gamma_2(\bmX_i; \bm \mu_0)/ \hat \gamma_1(\bmX_i) - \theta_0(\bmX_i; \bm \mu_0)\}-   L(\bm O_i; \hat \gamma- \gamma_0 , \bm \mu_0)  \right\| = o_p(n^{-1/2}). 
\end{align*}

Write $ \frac1n  \sum_{i=1}^{n} L(\bm O_i; \hat\gamma- \gamma_0; \bm\mu_0 )=  \frac1n  \sum_{i=1}^{n} L(\bm O_i; \hat\gamma- \bar \gamma; \bm\mu _0 ) + \frac1n  \sum_{i=1}^{n} L(\bm O_i; \bar\gamma- \gamma_0; \bm\mu_0 )$, where $ \bar \gamma= E(\hat\gamma) $. Next we show
\begin{align*}
	&\frac1n \| \sum_{i=1}^{n} L(\bm O_i; \hat\gamma- \gamma_0; \bm\mu_0 ) \| \leq \frac1n  \|   \sum_{i=1}^{n} L(\bm O_i; \hat\gamma- \bar \gamma; \bm\mu _0 ) \| + \frac1n  \| \sum_{i=1}^{n} L(\bm O_i; \bar\gamma- \gamma_0; \bm\mu_0 ) \|  \\
	&=O_p(\sqrt{m}/n ) .
\end{align*}
 The definition of $ L(\bm O; \gamma, \bm\mu_0) $, boundedness of $ \bmZ, \bmX $, and Assumption \ref{assump: kernel}(ii) 
  give that \\  $ \|L(\bm O; \gamma, \bm\mu_0)  \|\leq C\sqrt{m} \|\gamma \| $,  which implies that $ \|L(\bm O; \bar \gamma- \gamma_0, \bm\mu_0)  \|^2 \leq C m \|\bar \gamma- \gamma \|^2 $. From $ E( L(\bm O; \gamma, \bm\mu_0))  = 0$,  we  have that 
 \begin{align*}
 	E\left( \|  \frac1n  \sum_{i=1}^{n} L(\bm O_i; \bar\gamma- \gamma_0; \bm\mu_0 ) \|^2  \right) =\frac{1}{n^2}  \sum_{i=1}^{n}  E\left\{ \|   L(\bm O_i; \bar\gamma- \gamma_0; \bm\mu_0 ) \|^2  \right\} \leq \frac{Cm \|\bar\gamma - \gamma_0\|^2}{n}. 
 \end{align*}
Hence, $\|  \frac1n  \sum_{i=1}^{n} L(\bm O_i; \bar\gamma- \gamma_0; \bm\mu_0 ) \|   = o_p(n^{-1/2})$  from Markov inequality and  $m \|\bar\gamma- \gamma_0 \|^2= O(\sigma^4m )\rightarrow 0 $, a result from  \citet[Lemma 8.9]{Newey1994_handbook} and the choice of bandwidth $ \sigma $ in Assumption \ref{assump: kernel}(iii).
Then, we deal with $   \frac1n  \sum_{i=1}^{n} L(\bm O_i; \hat\gamma- \bar \gamma; \bm\mu_0 ) $. Let 
\begin{align*}
  &m(\bm O_i , \bm O_k) = (\bmZ_i -\bm D_i \bm\pi_0 ) f_0(\bmX_i)^{-1} \left\{ - \theta_0 (\bmX_i;\bm\mu_0 ) K_\sigma(\bmX_i - \bmX_k)  + (A_k - \bmV_k^T \bm \mu_0  )^2  K_\sigma(\bmX_i - \bmX_k )\right\}, \\
  &m_2(\bm O_k ) = \int m (\bm O_i, \bm O_k) dP_0(\bm O_i)  = 0, \\
  & m_1(\bm O_i ) = \int m (\bm O_i, \bm O_k) dP_0(\bm O_k) = L(\bm O_i; \bar\gamma, \bm \mu_0 ).
\end{align*}
We can write $   \frac1n  \sum_{i=1}^{n} L(\bm O_i; \hat\gamma- \bar \gamma; \bm\mu ) $ in the form of a  V-statistic 
\begin{align*}
	&\left\| \frac1n  \sum_{i=1}^{n} L(\bm O_i; \hat\gamma- \bar \gamma; \bm\mu_0  ) \right\|\\
	&=  \left\|  \frac1n  \sum_{i=1}^{n} L(\bm O_i; \hat\gamma; \bm\mu_0  ) -  \frac1n  \sum_{i=1}^{n} L(\bm O_i;\bar \gamma; \bm\mu_0  ) \right\| \\
		&=\left\|  n^{-2} \sum_{i=1}^{n} \sum_{k=1}^{n} m(\bm O_i, \bm O_j) - n^{-1} \sum_{i=1}^{n} m_1(\bm O_i)  \right\| \\
		&= \left\| n^{-2} \sum_{i=1}^{n} \sum_{k=1}^{n} m(\bm O_i, \bm O_j) - n^{-1} \sum_{i=1}^{n} m_1(\bm O_i)  - n^{-1} \sum_{i=1}^{n} m_2(\bm O_i) + E\{ m_1(\bm O)\}  \right\|\\
		&= O_p(\sqrt{m}/n ) 
\end{align*}
where the last line is from Lemma 8.4 of \cite{Newey1994_handbook}. This concludes the proof of \eqref{eq: kernel} from the triangle inequality. 

For $ B_2 $, we show that 
\begin{align}
	\left\|	\frac1n \sum_{i=1}^n  \big( \bmZ_i-\bm D_i \bm\pi_0 \big)  \{\hat\theta(\bmX_i; \hat{\bm \mu} )   -\hat \theta(\bmX_i; \bm \mu_0 ) \} \right\|= o_p(m \sqrt{\log m}/n)+ O_p(n^{-1} m \sqrt{m\log m/n}). \nonumber
\end{align}
The other term in $ B_2 $ can be shown in the same way.  Let $ \hat \gamma_3(\bmx) = \frac1n\sum_{k=1}^{n} (A_k - \bmV_k^T \bm\mu_0) \bmV_k^T K_\sigma(\bmx- \bmX_k) $ and  $ \hat \gamma_{3j}(\bmx) = \frac1n\sum_{k=1}^{n} (A_k - \bmV_{k}^T \bm\mu_0) \bmV_{kj} K_\sigma(\bmx- \bmX_k) $, where $  \hat \gamma_{3j}(\bmx) $ is the $ j $th component of $  \hat \gamma_3(\bmx) $. Using  triangle inequality and Strong Schwartz Matrix Inequality, we write
\begin{align}
&\left\| \frac1n \sum_{i=1}^n  \big( \bmZ_i-\bm D_i \bm\pi_0 \big)  \{\hat\theta(\bmX_i; \hat{\bm \mu} )   -\hat \theta(\bmX_i; \bm \mu_0 ) \}    \right\|\nonumber\\
&= \left\| - \frac2n \sum_{i=1}^n  \big( \bmZ_i-\bm D_i \bm\pi_0 \big)  \frac{\hat{\gamma}_3(\bmX_i)}{\hat\gamma_1(\bmX_i)} (\hat{\bm\mu}- \bm \mu_0)  + \text{Rem}_c\right\|\nonumber \\ 
&\leq  \left\|  \frac2n \sum_{i=1}^n  \big( \bmZ_i-\bm D_i \bm\pi_0 \big)  \frac{\hat{\gamma}_3(\bmX_i)}{\hat\gamma_1(\bmX_i)} (\hat{\bm\mu}- \bm \mu_0) \right\| +\left\| \text{Rem}_c\right\|\nonumber \\ 
&\leq  \left\|  \frac2n \sum_{i=1}^n  \big( \bmZ_i-\bm D_i \bm\pi_0 \big)  \frac{\hat{\gamma}_3(\bmX_i)}{\hat\gamma_1(\bmX_i)} \right\|\left\|(\hat{\bm\mu}- \bm \mu_0) \right\| +\left\| \text{Rem}_c\right\| \label{eq: B2}
\end{align}
where \begin{align*}
	\text{Rem}_c= (\hat{\bm\mu}- \bm\mu_0 )^T  \left\{\frac1n \sum_{i=1}^{n} (\bmZ_i - \bm D_i \bm\pi_0 )   \frac{\frac1n \sum_{k=1}^{n} \bmV_k \bmV_k^T K_\sigma(\bmX_i - \bmX_k)}{\frac{1}{n} \sum_{k=1}^{n} K_\sigma (\bmX_i - \bmX_k)} \right\} (\hat{\bm\mu}- \bm\mu_0 ). 
\end{align*}
Note that we have shown above that  $ \left\|(\hat{\bm\mu}- \bm \mu_0) \right\| = O_p(\sqrt{m\log m/n} )$. Because $ E\{(A - \bmV^T \bm\mu_0 )\bmV^T \mid \bmX\} = 0$, the true conditional expectation that $\hat\gamma_3(\bmX_i)/\hat\gamma_1(\bmX_i) $ is estimating is zero. We
 can use the same argument as in the proof of \eqref{eq: kernel} to show that 
\begin{align*}
	  &\left\|  \frac2n \sum_{i=1}^n  \big( \bmZ_i-\bm D_i \bm\pi_0 \big)  \frac{\hat{\gamma}_3(\bmX_i)}{\hat\gamma_1(\bmX_i)} \right\|^2\leq   \left\|  \frac2n \sum_{i=1}^n  \big( \bmZ_i-\bm D_i \bm\pi_0 \big)  \frac{\hat{\gamma}_3(\bmX_i)}{\hat\gamma_1(\bmX_i)} \right\|_F^2  \\
	  &= \sum_{k=1}^{m+d_x}    \left\|  \frac2n \sum_{i=1}^n  \big( \bmZ_i-\bm D_i \bm\pi_0 \big)  \frac{\hat{\gamma}_{3k}(\bmX_i)}{\hat\gamma_1(\bmX_i)} \right\|^2  = o_p(m/n)+ O_p(m^2/n^2). 
\end{align*}
Since the remainder term $ \text{Rem}_c $ contains the higher order terms, we conclude that \eqref{eq: B2} is $ o_p(m\sqrt{\log m}/n) + O_p( n^{-1} m\cdot \sqrt{m\log m/n})$.

(ii) Because $ g_i(\beta, \bfeta) $ is linear in $ \beta $, $ G_i(\bfeta) $ inherits its nice property in terms of the nuisance parameters. In other words, we still have 
\[
E\left\{ \nabla_{\bfeta} G(\bfeta_0)\right\} = 0.
\]
This key property ensures that the claim in (ii) is true and can be proved in the same way as (i). The details are omitted.

(iii) Since $\hat g(\beta, \bfeta)$ is linear in $\beta$, we have $ \hat g(\beta, \bfeta)- \hat g(\beta_0, \bfeta)= (\beta- \beta_0) \hat G(\bfeta).$ Hence, 
\begin{align*}
    &\sup_{\beta\in B} \|\hat g(\beta, \hat\bfeta) -\hat g(\beta, \bfeta_0) \| \\
    &\leq   \sup_{\beta\in B} \|\hat g(\beta, \hat\bfeta) - \hat g(\beta_0, \hat\bfeta) -\hat g(\beta, \bfeta_0) +\hat g(\beta_0, \bfeta_0) \| + \|\hat g(\beta_0, \hat\bfeta) -\hat g(\beta_0, \bfeta_0) \| \\
    &= \sup_{\beta\in B} \|(\beta-\beta_0) \{\hat G(\hat \bfeta) -  \hat G( \bfeta_0)\}\| + \|\hat g(\beta_0, \hat\bfeta) -\hat g(\beta_0, \bfeta_0) \| \\
    &= \sup_{\beta\in B} |\beta-\beta_0| \| \hat G(\hat \bfeta) -  \hat G( \bfeta_0)\| + \|\hat g(\beta_0, \hat\bfeta) -\hat g(\beta_0, \bfeta_0) \|. 
\end{align*}
The result follows from (i)-(ii) and the compactness of $B$.

\end{proof}

A direct implication of the following result (i) is that when $ m^2/n\rightarrow 0 $, $  \|\hat\Omega(\beta_0,\hat \bfeta)- \Omega \| = o_p(1)$; when $ m^3/n\rightarrow 0 $, $  \sqrt{m}\|\hat\Omega(\beta_0,\hat \bfeta)- \Omega \| = o_p(1)$. The implication is similar for (ii)-(v).

{Next, under an extra assumption that $E(R_A R_Y\mid \bmX = \bmx ) = \omega_0^T \bm Q $ and $E(R_A^2 \mid \bmX = \bmx ) = \theta_0^T \bm Q $, where $\bm Q$ is a vector that includes all quadratic terms of $\bmX$,  i.e., they  follow linear models that include a full set of quadratic terms, we can then estimate these nuisance parameters using least squares (denoted as $\hat \omega$ and $\hat\theta$) instead of nonparametric kernel. We prove the following result:

{\bf Lemma S5':}
Suppose that Assumptions \ref{assump: linear}-\ref{assump:omega} hold. Also suppose that  $E(R_A R_Y\mid \bmX = \bmx ) = \omega_0^T \bm Q $ and $E(R_A^2 \mid \bmX = \bmx ) = \theta_0^T \bm Q $, where $\bm Q$ is a vector that includes all quadratic terms of $\bmX$. When $ m^2/n\rightarrow 0 $, \\
(i)  $\|\hat g(\beta_0, \hat\bfeta) -\hat g(\beta_0, \bfeta_0) \|=O_p(m\log m/n)$;\\
(ii) $\| \hat G(\hat \bfeta) -  \hat G( \bfeta_0)\|=  O_p(m\log m/n)$;\\
(iii) $\sup_{\beta\in B} \|\hat g(\beta, \hat\bfeta) -\hat g(\beta, \bfeta_0) \|=  O_p(m\log m/n)$.

\begin{proof}
(i)	It suffices to prove part (b) in the proof of Lemma \ref{lemma: g beta, nuisance} (i). The rest are the same as the proof of Lemma \ref{lemma: g beta, nuisance}. Write
	\begin{align*}
		& \hat g(\beta_0,  \hat \bfeta) - \hat 	g( \beta_0, \hat \bfeta_{p},  \bfeta_{np0})  \\
		&= \frac{1}{n} \sum_{i=1}^{n}    \big( \bmZ_i-\bm D_i \bm\pi_0 \big)  \left[ - (\hat\omega - \omega_0)^T \bm Q_i + \beta_0 (\hat\theta -\theta_0 )^T  \bm Q_i  \right] +\text{Rem}_b
	\end{align*}  
Note that $\bmQ$ has a finite dimension, and thus 
\begin{align*}
	\hat \omega - \omega_0 & = \var (\bmQ)^{-1} \frac1n \left\{\sum_{i=1}^n (\bm Q_i - E \bmQ) (A_i - \bmV_i^T \hat{\bm\mu}) (Y_i - \bmV_i^T \hat{\bm\lambda}) -  \cov (\bmQ, R_AR_Y ) \right\} + \text{Rem}_c \\
	&= \var (\bmQ)^{-1} \frac1n \left\{\sum_{i=1}^n (\bm Q_i - E \bmQ) R_{Ai} R_{Yi}-  \cov (\bmQ, R_AR_Y ) \right\}  \\
	&\quad - \var (\bmQ)^{-1} \frac1n \left\{\sum_{i=1}^n R_{Y_i} (\bm Q_i - E \bmQ) \bmV_i^T (\hat{\bm\mu}- \bm\mu_0) \right\}  \\
	&\quad - \var (\bmQ)^{-1} \frac1n \left\{\sum_{i=1}^n R_{A_i} (\bm Q_i - E \bmQ) \bmV_i^T (\hat{\bm\lambda}- \bm\lambda_0) \right\}+    \text{Rem}_c \\
	&:= B_1 - B_2 - B_3 + \text{Rem}_c  .
\end{align*}
Here, $B_1 = O_p(n^{-1/2})$ by finite-dimensional linear model theory. For the second term, using the Strong Schwartz Matrix Inequality, 
\begin{align*}
	B_2 & = \var (\bmQ)^{-1} \frac1n \left\{\sum_{i=1}^n R_{Y_i} (\bm Q_i - E \bmQ) \bmV_i^T  \right\} (\hat{\bm\mu}- \bm\mu_0) \\
	& \leq C \| \frac1n \sum_{i=1}^n R_{Y_i} (\bm Q_i - E \bmQ) \bmV_i^T \| \|\hat{\bm\mu}- \bm\mu_0 \| .
\end{align*} 
Similar to the proof of Lemma \ref{lemma: g beta, nuisance}, we use  \citet[Theorem 1]{tropp2015expected}. We calculate the matrix variance parameter 
\begin{align*}
	v= n^{-1} \| E (  R_{Y_i}^2 (\bm Q_i - E \bmQ) \bmV_i^T  \bmV_i  (\bm Q_i - E \bmQ) ^T ) \| \leq Cn^{-1} m .
\end{align*}
Also note that 
\begin{align*}
	& \left\|  (\bm Q_i - E \bmQ) \bmV_i^T  \right\|^2  = \xi_{\max} \left\{  (\bm Q_i - E \bmQ) \bmV_i^T   \bmV_i (\bm Q_i - E \bmQ)^T \right\}\\
	&\leq tr \left\{ (\bm Q_i - E \bmQ) \bmV_i^T   \bmV_i (\bm Q_i - E \bmQ)^T \right\}  \leq Cm (\bm Q_i - E \bmQ)  (\bm Q_i - E \bmQ)^T \leq  Cm .
\end{align*}
This implies that 
\begin{align*}
	L^2 & = n^{-2} E \{ \max_i \|   (\bm Q_i - E \bmQ) \bmV_i^T \|^2 R_{Y_i}^2\} \leq C n^{-2}  m E (\max_i R_{Yi}^2 ) \\
	&\leq C  n^{-2}  m  E (\max_i Y_i^2 ) \leq C  n^{-2}  m (\log n )^2
\end{align*}

From the above analysis and Theorem 1 of \cite{tropp2015expected}, we know the matrix variance parameter $\nu$ is driving the upper bound and 
\begin{align*}
	E \| \frac1n \sum_{i=1}^n R_{Y_i} (\bm Q_i - E \bmQ) \bmV_i^T \|^2 \leq  C \log (m) \nu \leq C  \frac{m \log m}{n}. 
\end{align*}
From Markov inequality, we know that $ \| \frac1n \sum_{i=1}^n R_{Y_i} (\bm Q_i - E \bmQ) \bmV_i^T \|^2 = O_p( m\log m/n)$. We showed in the proof of Lemma \ref{lemma: g beta, nuisance}, $ \|\hat{\bm\mu}- \bm\mu_0 \|= O_p(\sqrt{m\log m /n})$. Thus, $\|B_2\| = O_p( m\log m/n)$.  The last term $B_3$ is bounded using the same argument. As the remainder term consists of higher order terms and can be shown to be negligible.  Therefore, we have  
 \[
 	\hat \omega - \omega_0 = O_p(n^{-1/2}) + O_p( m\log m/n)
 \]

Then, using a similar argument as that for $A_2$ in  Lemma \ref{lemma: g beta, nuisance}, we have 
\begin{align*}
 & \|	\frac{1}{n} \sum_{i=1}^{n}    \big( \bmZ_i-\bm D_i \bm\pi_0 \big)    \bm Q_i^T (\hat\omega - \omega_0)  \|\\
 &\leq  \| 	\frac{1}{n} \sum_{i=1}^{n}    \big( \bmZ_i-\bm D_i \bm\pi_0 \big)    \bm Q_i^T \| \|\hat\omega - \omega_0 \| \\
 &= O_p(\sqrt{m \log m /n})  \{ O_p(n^{-1/2}) + O_p( m\log m/n) \} \\
 &= O_p (\sqrt{m\log m }/n ) + O_p ((m\log m/n)^{3/2}).  
\end{align*}                                                                                                    

Thus, we see that part (b) is negligible compared to part (a), so  $\|\hat g(\beta_0, \hat\bfeta) -\hat g(\beta_0, \bfeta_0) \|$  is asymptotically equivalent with part (a).                         

The proof of (ii) and (iii) are the same as that of Lemma \ref{lemma: g beta, nuisance}.

\end{proof}
}

\begin{lemma} \label{lemma: Omega eta hat}
	  Under Assumptions    \ref{assump: linear}-\ref{assump:omega},   \\
	  (i) $ \|\hat\Omega(\beta_0,\hat \bfeta)- \Omega \| =O_p(\sqrt{m\log m/n}) + O_p(\sqrt{m^3(\log m)^2/n^2})$; \\
	  (ii) $   \|n^{-1} \sum_{i=1}^{n}g_i(\beta_0,\hat \bfeta) G_i(\hat\bfeta)^T - E\{g_i G_i^T\} \|=O_p(\sqrt{m\log m/n}) + O_p(\sqrt{m^3(\log m)^2/n^2})$; \\
	  	  (iii) $    \|n^{-1} \sum_{i=1}^{n}G_i(\hat \bfeta) G_i(\hat\bfeta)^T - E\{G_i G_i^T\} \| =O_p(\sqrt{m\log m/n}) + O_p(\sqrt{m^3(\log m)^2/n^2})$; \\ 
	  	  (iv) $ 	\sup_{\beta\in B} \|\hat\Omega(\beta,\hat \bfeta)- \Omega(\beta, \bfeta_0)\| =O_p(\sqrt{m\log m/n}) + O_p(\sqrt{m^3(\log m)^2/n^2})$; \\
	  	 (v) $ 	\sup_{\beta\in B}  \|n^{-1} \sum_{i=1}^{n}g_i(\beta,\hat \bfeta) G_i(\hat\bfeta)^T - E\{g_i(\beta,\bfeta_0) G_i^T\} \| =O_p(\sqrt{m\log m/n}) + O_p(\sqrt{m^3(\log m)^2/n^2})$; \\
\end{lemma}
\begin{proof}
	(i)
Define $ \Delta_{i0}= ( A_i- \bmV_i^T \bm\mu_0 )\{ Y_i- \bmV_i^T \bm{\lambda}_{0} - \beta_0( A_i- \bmV_i^T \bm \mu_0 ) \}$. We will also write $ \hat \Delta_{i0}= ( A_i- \bmV_i^T \hat{\bm\mu} )\{ Y_i- \bmV_i^T \hat{\bm{\lambda}} - \beta_0( A_i- \bmV_i^T \hat{\bm \mu} ) \}$. Then 
\begin{align*}
&\hat	\Omega (\beta_0,\hat\bfeta) = \frac{1}{n} \sum_{i=1}^{n} g_i(\beta_0,\hat\bfeta) g_i(\beta_0,\hat\bfeta)^T  =   \frac{1}{n} \sum_{i=1}^{n} (\bmZ_i - \bm D_i \hat{\bm\pi}) (\bmZ_i - \bm D_i \hat{\bm\pi})^T \{\hat\Delta_{i0} - \hat w(\bmX_i) + \beta_0 \hat \theta(\bmX_i)\}^2.
\end{align*} 
By triangle inequality, 
	\begin{align*}
		 &\|\hat\Omega(\beta_0,\hat \bfeta)- \Omega \|\\
		 &= \|\hat\Omega(\beta_0,\hat \bfeta)  - \hat \Omega(\beta_0, \bfeta_0) + \hat \Omega(\beta_0, \bfeta_0) - \Omega \|\\
		 &\leq  \|\hat\Omega(\beta_0,\hat \bfeta)  - \hat \Omega(\beta_0, \bfeta_0) \| +\|	 \hat \Omega(\beta_0, \bfeta_0) - \Omega \| \\
		 & \leq \left\|  \frac{1}{n} \sum_{i=1}^{n} (\bmZ_i - \bm D_i \hat{\bm\pi}) (\bmZ_i - \bm D_i \hat{\bm\pi})^T \hat\Delta_{i0}^2-(\bmZ_i - \bm D_i {\bm\pi}_0) (\bmZ_i - \bm D_i {\bm\pi}_0)^T \Delta_{i0}^2  \right\|\\
		 & \qquad+  \left\|  \frac{1}{n} \sum_{i=1}^{n} (\bmZ_i - \bm D_i \hat{\bm\pi}) (\bmZ_i - \bm D_i \hat{\bm\pi})^T \hat\Delta_{i0}(\hat\omega(\bmX_i)- \beta_0\hat\theta(\bmX_i)) \right.\\
		 & \qquad\qquad\qquad\qquad \left.-(\bmZ_i - \bm D_i {\bm\pi}_0) (\bmZ_i - \bm D_i {\bm\pi}_0)^T \Delta_{i0}(\omega_0(\bmX_i)- \beta_0\theta_0(\bmX_i)) \right\|\\
		 &\qquad + \left\|  \frac{1}{n} \sum_{i=1}^{n} (\bmZ_i - \bm D_i \hat{\bm\pi}) (\bmZ_i - \bm D_i \hat{\bm\pi})^T(\hat\omega(\bmX_i)- \beta_0\hat\theta(\bmX_i))^2  \right.\\
		 &\qquad\qquad\qquad\qquad \left.-(\bmZ_i - \bm D_i {\bm\pi}_0) (\bmZ_i - \bm D_i {\bm\pi}_0)^T (\omega_0(\bmX_i)- \beta_0\theta_0(\bmX_i))^2  \right\|\\
		 &\qquad +\|	 \hat \Omega(\beta_0, \bfeta_0) - \Omega \| \\
		 &:= \|I_1\|+\|I_2\|+\|I_3\|+\|I_4\|.
	\end{align*}

A key technique we use here is to vectorize the matrix, then a proof very similar to the proof of Lemma \ref{lemma: g beta, nuisance} will establish the result. Specifically, since $ \|I_1\|\leq \|I_1\|_F= \|\text{vec} (I_1)\| $, where $ \text{vec}(A)= (a_1^T, \dots, a_m^T)^T\in \mathbb{R}^{mn}  $ for a matrix $ A= (a_1, \dots, a_m) \in\mathbb{R}^{n\times m}$.  

Let $ \otimes $ be the Kronecker product. For $  \text{vec} (I_1) $, 
	\begin{align}
	 & \text{vec} (I_1)  \nonumber \\
		& = \nabla_{\bfeta_p}  \left[ \frac{1}{n} \sum_{i=1}^{n}  (\bmZ_i - \bm D_i {\bm\pi}_0) \otimes  (\bmZ_i - \bm D_i {\bm\pi}_0) \Delta_i^2  \right]  (\hat{\bfeta}_p- \bfeta_{p0}) + \text{Rem}_d \nonumber\\
		&= - \frac1n \sum_{i=1}^{n} \Delta_i^2 (\bmZ_i - \bm D_i \bm\pi_0 ) \otimes  \bm D_i (\hat{\bm\pi}- \bm\pi_0 )  - \frac1n \sum_{i=1}^{n} \Delta_i^2  \{ \bm D_i (\hat{\bm\pi}- \bm\pi_0 )\}\otimes (\bmZ_i - \bm D_i \bm\pi_0 ) \nonumber\\
		&\qquad + \frac{1}{n} \sum_{i=1}^{n}  (\bmZ_i - \bm D_i {\bm\pi}_0) \otimes  (\bmZ_i - \bm D_i {\bm\pi}_0) 2\Delta_i (A_i - \bmV_i^T \bm\mu_0) \bmV_i^T (\hat{\bm \lambda}- \bm\lambda_0) \nonumber \\
		&\qquad + \frac{1}{n} \sum_{i=1}^{n}  (\bmZ_i - \bm D_i {\bm\pi}_0) \otimes  (\bmZ_i - \bm D_i {\bm\pi}_0) 2\Delta_i (- Y_i + 2\beta_0 A_i + \bmV_i^T \bm\lambda_0 - 2\beta_0 \bmV_i^T\bm\mu_0) \bmV_i^T (\hat{\bm \mu}- \bm\mu_0)\nonumber\\
		&\qquad + \text{Rem}_d.  \label{eq: vec I1}
	\end{align}
For the above first term, from the Strong Schwartz Matrix Inequality, 
\begin{align*}
	& \left\| \frac1n \sum_{i=1}^{n} \Delta_i^2 (\bmZ_i - \bm D_i \bm\pi_0 ) \otimes  \bm D_i (\hat{\bm\pi}- \bm\pi_0 )  \right\|^2 \nonumber\\
	&= \sum_{j=1}^{m}   	\left\| \frac1n \sum_{i=1}^{n} \Delta_i^2 (\bmZ_i - \bm D_i \bm\pi_0 )   \bm X_i^T (\hat{\bm\pi}_j- \bm\pi_{j0} )  \right\|^2  \\
	&\leq 	\left\| \frac1n \sum_{i=1}^{n} \Delta_i^2 (\bmZ_i - \bm D_i \bm\pi_0 )   \bm X_i^T   \right\|^2  \sum_{j=1}^{m}     \left\| \hat{\bm\pi}_j- \bm\pi_{j0}  \right\|^2 \\
	&= 	\left\| \frac1n \sum_{i=1}^{n} \Delta_i^2 (\bmZ_i - \bm D_i \bm\pi_0 )   \bm X_i^T   \right\|^2    \left\| \hat{\bm\pi}- \bm\pi_{0}  \right\|^2 \\
	&= O_p(m^2(\log m)^2/n^2),
\end{align*} 
where the last expression is obtained using the same argument as in the derivation  of $ A_{21} $ in the proof of Lemma \ref{lemma: g beta, nuisance}, and $ \| \hat{\bm\pi} - \bm\pi_0 \|^2  = O_p (m\log m/n ) $ derived after \eqref{eq: pi}. The second term in \eqref{eq: vec I1} has the same norm as the first term in \eqref{eq: vec I1}.  For the third term in \eqref{eq: vec I1}, we have 
\begin{align*}
	&\left\| \frac{1}{n} \sum_{i=1}^{n}  (\bmZ_i - \bm D_i {\bm\pi}_0) \otimes  (\bmZ_i - \bm D_i {\bm\pi}_0) 2\Delta_i (A_i - \bmV_i^T \bm\mu_0) \bmV_i^T (\hat{\bm \lambda}- \bm\lambda_0) \right\|^2 \\
	&= \sum_{j=1}^{m} \left\| \frac{1}{n} \sum_{i=1}^{n}  (\bmZ_i - \bm D_i {\bm\pi}_0)   (Z_{ij} - \bm X_i^T  {\bm\pi}_{j0}) 2\Delta_i (A_i - \bmV_i^T \bm\mu_0) \bmV_i^T (\hat{\bm \lambda}- \bm\lambda_0) \right\|^2 \\
	&\leq  \sum_{j=1}^{m} \left\| \frac{1}{n} \sum_{i=1}^{n}  (\bmZ_i - \bm D_i {\bm\pi}_0)   (Z_{ij} - \bm X_i^T  {\bm\pi}_{j0}) 2\Delta_i (A_i - \bmV_i^T \bm\mu_0) \bmV_i^T \right\|^2 \| \hat{\bm \lambda}- \bm\lambda_0\|^2 \\
	& =  \sum_{j=1}^{m} \left\| \sum_{i=1}^{n}  S_{ij} \right\|^2 \| \hat{\bm \lambda}- \bm\lambda_0\|^2 ,
\end{align*}
where $ S_{ij} = \frac1n (\bmZ_i - \bm D_i {\bm\pi}_0)   (Z_{ij} - \bm X_i^T  {\bm\pi}_{j0}) 2\Delta_i (A_i - \bmV_i^T \bm\mu_0) \bmV_i^T  $. Then, we calculate the matrix variance parameter 
\begin{align*}
v_j &= n \|E(S_{ij} S_{ij}^T )\| \\
&= n^{-1}  \left\|E[(\bmZ_i - \bm D_i {\bm\pi}_0)  (\bmZ_i - \bm D_i {\bm\pi}_0)^T  \{ (Z_{ij} - \bm X_i^T  {\bm\pi}_{j0}) 2\Delta_i (A_i - \bmV_i^T \bm\mu_0)\}^2  \bmV_i^T \bmV_i  ] \right\|\\
&\leq C n^{-1} m   \left\|E[(\bmZ_i - \bm D_i {\bm\pi}_0)  (\bmZ_i - \bm D_i {\bm\pi}_0)^T  \{  \Delta_i (A_i - \bmV_i^T \bm\mu_0)\}^2   ] \right\|\\
& \leq Cn^{-1} m  
\end{align*}
and  the large deviation parameter 
\begin{align*}
	L_j^2  &= E\left( \max_i \|S_{ij }\|^2\right) \leq Cn^{-2} m^2 E\left\{ \max_i (\Delta_i^2 (A_i - \bmV_i^T \bm\mu_0)^2)\right\} \leq Cn^{-2} m^2 (\log n )^8,
\end{align*}
where the last inequality is from Assumption \ref{assump: moment} and \cite{buhlmann2011statistics}.  Hence, we have from \citet[Theorem 1]{tropp2015expected} that 
\begin{align*}
	E\left( \sum_{j=1}^{m} \left\| \sum_{i=1}^{n}  S_{ij} \right\|^2\right) \leq \frac{m^2\log m}{n}.
\end{align*}
This combined with Markov inequality, we have that $ \sum_{j=1}^{m} \left\| \sum_{i=1}^{n}  S_{ij} \right\|^2= O_p (m^2\log m/n) $. This implies the third term in \eqref{eq: vec I1} is
\begin{align*}
	&\left\| \frac{1}{n} \sum_{i=1}^{n}  (\bmZ_i - \bm D_i {\bm\pi}_0) \otimes  (\bmZ_i - \bm D_i {\bm\pi}_0) 2\Delta_i (A_i - \bmV_i^T \bm\mu_0) \bmV_i^T (\hat{\bm \lambda}- \bm\lambda_0) \right\|^2  = O_p\left(\frac{m^3(\log m)^2}{n^2}\right)
\end{align*}
The other proofs are very similar to the derivations in the proof of Lemma \ref{lemma: g beta, nuisance} and are omitted. Hence, we have shown that $ \|I_1\| = O_p(m^{3/2}\log m/n) $.

We similarly vectorize  $ I_3 $ and analyze $ {\rm vec}(I_3) $, following the steps in the proof of Lemma \ref{lemma: g beta, nuisance}. First, we show that 
	\begin{align}
		\left\|	\frac1n \sum_{i=1}^n  \big( \bmZ_i-\bm D_i \bm\pi_0 \big) \otimes\big( \bmZ_i-\bm D_i \bm\pi_0 \big)  \{\hat\theta^2(\bmX_i; \bm \mu )   - \theta_0^2(\bmX_i; \bm \mu ) \} \right\|  =  o_p(\sqrt{m/n}) \label{eq: Omega, kernel}.
	\end{align}
Following the notations in \eqref{eq: kernel}. We obtain the linearization of $		\frac1n \sum_{i=1}^{n}  \big( \bmZ_i-\bm D_i \bm\pi_0 \big) \otimes\big( \bmZ_i-\bm D_i \bm\pi_0 \big)  \{\hat\theta^2(\bmX_i; \bm \mu )   - \theta_0^2(\bmX_i; \bm \mu ) \} $ as $ \frac1n \sum_{i=1}^{n} L(\bm O_i; \hat\gamma- \gamma_0, \bm \mu) = \frac1n \sum_{i=1}^{n} \{ L(\bm O_i; \hat\gamma, \bm \mu) -  L(\bm O_i; \gamma_0, \bm \mu) \} $, where 
\begin{align*}
	 L(\bm O; \gamma, \bm \mu) =   \big( \bmZ-\bm D \bm\pi_0 \big) \otimes\big( \bmZ-\bm D \bm\pi_0 \big) 2\theta_0 (\bmX; \bm\mu) f_0(\bmX)^{-1} \{-\theta_0(\bmX; \bm\mu) \gamma_1(\bmX) + \gamma_2(\bmX; \bm\mu)\}, 
\end{align*}
and the remainder term from the linearlization is $ o_p(\sqrt{m/n}) $. Then, we have $ \|\frac{1}{n}\sum_{i=1}^{n}L(\bm O_i;\hat\gamma- \gamma_0; \bm\mu ) \|=  O_p(m/n)$. Thus, we have that \eqref{eq: Omega, kernel} is of order $ o_p(\sqrt{m/n}) $. We can also show that $ 	\left\|	\frac1n \sum_{i=1}^n  \big( \bmZ_i-\bm D_i \bm\pi_0 \big) \otimes\big( \bmZ_i-\bm D_i \bm\pi_0 \big)  \{\hat\theta^2(\bmX_i; \hat {\bm \mu} )   -\hat \theta^2(\bmX_i; \bm \mu_0 ) \} \right\|  =  o_p(\sqrt{m^3\log m}/n) $. The other terms in $ I_3 $ can be shown similarly. Therefore, $\|I_3\|\leq  \|I_3\|_F =  \|{\rm vec} (I_3)\| = o_p(\sqrt{m/n})  +  o_p(\sqrt{m^3\log m}/n)$.

Next, notice that $\| I_2\| $ can be bounded by 
\begin{align*}
 &  \bigg\|  \frac{1}{n} \sum_{i=1}^{n}\left\{ (\bmZ_i - \bm D_i \hat{\bm\pi}) (\bmZ_i - \bm D_i \hat{\bm\pi})^T \hat\Delta_i -(\bmZ_i - \bm D_i {\bm\pi}_0 ) (\bmZ_i - \bm D_i {\bm\pi_0})^T \Delta_i \right\} (\omega_0(\bmX_i)- \beta_0\theta_0(\bmX_i))\bigg\|\\
	  + &\bigg\|  \frac{1}{n} \sum_{i=1}^{n}(\bmZ_i - \bm D_i {\bm\pi}_0 ) (\bmZ_i - \bm D_i {\bm\pi_0})^T \Delta_i (\hat\omega(\bmX_i)- \beta_0\hat\theta(\bmX_i) - \omega_0(\bmX_i) + \beta_0 \theta_0(\bmX_i) )\bigg\|\\
	  +& \|\text{Rem}_I\|
\end{align*}  
where $ \text{Rem}_I $ is a higher order term, the above first term can be bounded in the same way as $ I_1 $, the above second term can be bounded in the same way as $ I_3 $.

	Finally, we again use the matrix concentration inequality \citet[Theorem 1]{tropp2015expected}  to show that $ \|I_4\|= 	\|\hat \Omega(\beta_0, \bfeta_0)-  \Omega(\beta_0, \bfeta_0)\| =O_p(\sqrt{m\log m/n}) $. We can show that the matrix variance parameter $ v \leq Cm/n$ and the large deviation parameter $ L^2 \leq Cm\log n/n^2$.  Hence, the matrix variance parameter term $ v $ drives the order, and  $ E	\|\hat \Omega(\beta_0, \bfeta_0)-  \Omega(\beta_0, \bfeta_0)\|^2\leq   m\log m/n $. The result follows  from Markov inequality.

Combining the above arguments, we have that $ \|\hat{\Omega}(\beta_0,\bfeta_0) - \Omega\| = O_p(\sqrt{m\log m/n}) + O_p(\sqrt{m^3(\log m)^2/n^2})$. 
	
	(ii)-(iii) The proof follows the same steps as in part (i) and is omitted. 
	
	(iv) Note that 
	\begin{align*}
		&	\sup_{\beta\in B} \|\hat\Omega(\beta,\hat \bfeta)- \Omega(\beta, \bfeta_0)\| \\
		&\leq \frac{1}{n}\| \sum_{i=1}^{n} g_i(\beta_0, \hat \bfeta)g_i(\beta_0, \hat \bfeta)^T - g_i g_i^T   \| +   	\sup_{\beta\in B}  |\beta- \beta_0| \frac{2}{n}\| \sum_{i=1}^{n} G_i(\hat \bfeta)g_i(\beta_0, \hat \bfeta)^T - G_i g_i^T   \| \\
		&\qquad + 		  	\sup_{\beta\in B}  |\beta- \beta_0|^2 \frac{1}{n}\| \sum_{i=1}^{n} G_i(\hat \bfeta)G_i( \hat \bfeta)^T - G_i G_i^T   \|. 
	\end{align*}
The result follows from (i)-(iii) and the compactness of $ B $. 

(v) The proof follows the same steps as in part (iv). 
\end{proof}

{Similar to Lemma S5', Lemma S6' is a parallel result to Lemma S6 assuming that $E(R_A R_Y\mid \bmX = \bmx ) = \omega_0^T \bm Q $ and $E(R_A^2 \mid \bmX = \bmx ) = \theta_0^T \bm Q $ are estimated using parametric methods.  

{\bf Lemma S6':}
Suppose that Assumptions \ref{assump: linear}-\ref{assump:omega} hold. Also suppose that  $E(R_A R_Y\mid \bmX = \bmx ) = \omega_0^T \bm Q $ and $E(R_A^2 \mid \bmX = \bmx ) = \theta_0^T \bm Q $, where $\bm Q$ is a vector that includes all quadratic terms of $\bmX$.  When $m^2/n\to 0$,   Lemma \ref{lemma: Omega eta hat} (i)-(v)  hold.

\begin{proof}
	(i) The proof is similar to the proof of (i) in Lemma \ref{lemma: Omega eta hat}. It suffices to prove the term $I_3$.
Note that 
\begin{align*}
		\left\|	\frac1n \sum_{i=1}^n  \big( \bmZ_i-\bm D_i \bm\pi_0 \big) \otimes\big( \bmZ_i-\bm D_i \bm\pi_0 \big)  \{\hat\omega^2(\bmX_i; \bm \mu )   - \omega_0^2(\bmX_i; \bm \mu ) \} \right\|^2 \leq  \sum_{j=1}^{m} \left\| \sum_{i=1}^{n}  S_{ij} \right\|^2 \| \hat{\bm \omega}- \bm\omega_0\|^2 ,
\end{align*}
where $S_{ij}= \frac1n (\bmZ_i - \bm D_i {\bm\pi}_0)   (Z_{ij} - \bm X_i^T  {\bm\pi}_{j0}) 2\omega_0^T \bmQ_i \bmQ_i^T  $. Then similar to the proof of Lemma \ref{lemma: Omega eta hat}, we have $ \sum_{j=1}^{m} \left\| \sum_{i=1}^{n}  S_{ij} \right\|^2= O_p (m\log m/n) $. Then with $\hat{\bm \omega}- \bm\omega_0= O_p(n^{-1/2}) + O_p( m\log m/n) $ proved in Lemma S5', we have 
\begin{align*}
 &	\left\|	\frac1n \sum_{i=1}^n  \big( \bmZ_i-\bm D_i \bm\pi_0 \big) \otimes\big( \bmZ_i-\bm D_i \bm\pi_0 \big)  \{\hat\omega^2(\bmX_i; \bm \mu )   - \omega_0^2(\bmX_i; \bm \mu ) \} \right\|  \\
	& =  O_p ( \frac{ \sqrt{m\log m} }{n }  ( 1+  \frac{m \log m  }{\sqrt{n}})) =  o_p(\sqrt{m/n}).  
\end{align*}
Hence, $I_3$ is negligible compared to the other terms, and Lemma \ref{lemma: Omega eta hat} (i) still holds. 

(ii)-(v) The proof is the same as those in Lemma \ref{lemma: Omega eta hat}.

\end{proof}

}

\begin{lemma} \label{lemma:unif Q}
Under Assumptions \ref{assump: many weak moments}-\ref{assump:omega}, and  $ m^2/n\rightarrow 0 $,    $$ \sup_{\beta\in B} \mu_n^{-2} n  |\hat Q(\beta, \hat\bfeta) - Q(\beta, \bfeta_0) |= o_p(1).$$      
\end{lemma}
\begin{proof}
Note that by Assumption \ref{assump:omega}, $ n E\{\|\hat g (\beta_0, \bfeta_0)\|^2 \}/m =  tr(\Omega(\beta_0, \bfeta_0)) /m \leq C$, so by Markov inequality, $\|\hat g (\beta_0, \bfeta_0)\|= O_p(\sqrt{m/n})$. Also by Lemma \ref{lemma: g bdd by beta}, Lemma \ref{lemma: g beta, nuisance},  triangle inequality, the compactness of $B$, and $ m^2/n\rightarrow 0 $, 
\begin{align}
   & \sup_{\beta\in B} \|\hat g(\beta, \hat\bfeta) \|=  \sup_{\beta\in B} \|\hat g(\beta, \hat\bfeta) -\hat g(\beta, \bfeta_0) + \hat g(\beta, \bfeta_0) - \hat g(\beta_0, \bfeta_0) +\hat g(\beta_0, \bfeta_0)\|\nonumber\\
    &\leq \sup_{\beta\in B} \|\hat g(\beta, \hat\bfeta) -\hat g(\beta, \bfeta_0) \|+  \sup_{\beta\in B} \| \hat g(\beta, \bfeta_0) - \hat g(\beta_0, \bfeta_0) \|+ \|\hat g(\beta_0, \bfeta_0)\|\nonumber\\ 
    &= o_p(\sqrt{m}/\sqrt{n}) +  O_p(\mu_n/\sqrt{n})  +  O_p(\sqrt{m}/\sqrt{n})  \nonumber\\ 
  &=  O_p(\mu_n/\sqrt{n}).    \label{eq: g hat bdd in prob}
\end{align}
A useful implication of the above derivation is that $ \|\hat g(\beta_0, \hat\bfeta) \| = O_p(\sqrt{m}/\sqrt{n}) $. 

Let $\hat a(\beta, \hat\bfeta)= \mu_n^{-1} \sqrt{n} \Omega(\beta, \bfeta_0)^{-1}\hat g(\beta, \hat\bfeta) $. By Assumption \ref{assump:omega} and \eqref{eq: g hat bdd in prob}, 
\begin{align*}
    \|\hat a(\beta, \hat\bfeta)\|^2= \mu_n^{-2} n \hat g(\beta, \hat\bfeta)^T \Omega(\beta, \bfeta_0)^{-1}    \Omega(\beta, \bfeta_0)^{-1}  \hat g(\beta, \hat\bfeta) \leq C \mu_n^{-2} n \| \hat g(\beta, \hat\bfeta)\|^2,
\end{align*}
so that $\sup_{\beta\in B}  \|\hat a(\beta, \hat\bfeta)\|^2 = O_p(1) $. Also, by Assumption \ref{assump:omega} and Lemma \ref{lemma: Omega eta hat}, we have
\begin{align}
    |\xi_{\min}(\hat \Omega(\beta,\hat \bfeta))-\xi_{\min}( \Omega(\beta, \bfeta_0)) |\leq \sup_{\beta\in B } \|\hat \Omega(\beta, \hat \bfeta)- \Omega(\beta, \bfeta_0) \| + o_p(1)= o_p(1), \label{eq: Omega hat bdd}
\end{align}
so that $\xi_{\min}(\hat \Omega(\beta, \hat\bfeta))\geq C $ and hence $\xi_{\max}(\hat \Omega(\beta, \hat\bfeta)^{-1})\leq C $ for all $\beta\in B$, w.p.a.1.

Therefore, 
\begin{align*}
  &\mu_n^{-2} n 2 \hat Q(\beta, \hat\bfeta)= \hat a(\beta, \hat\bfeta)^T \Omega(\beta, \bfeta_0) \hat\Omega(\beta, \hat\bfeta)^{-1}\Omega(\beta, \bfeta_0) \hat a(\beta, \hat\bfeta),\\
  &\mu_n^{-2} n 2 \tilde Q(\beta, \hat\bfeta)= \hat a(\beta, \hat\bfeta)^T \Omega(\beta, \bfeta_0) \hat a(\beta, \hat\bfeta),
\end{align*}
and 
\begin{align*}
    &2\mu_n^{-2} n |\hat Q (\beta, \hat\bfeta)- \tilde{Q} (\beta, \hat\bfeta) | \\
    &\leq |\hat a(\beta, \hat\bfeta)^T \{ \hat\Omega(\beta, \hat\bfeta) - \Omega(\beta, \bfeta_0)\}\hat a(\beta, \hat\bfeta)| \\
    &\qquad + |\hat a(\beta, \hat\bfeta)^T \{ \hat\Omega(\beta, \hat\bfeta) - \Omega(\beta, \bfeta_0)\} \hat\Omega(\beta, \hat\bfeta)^{-1} \{ \hat\Omega(\beta, \hat\bfeta) - \Omega(\beta, \bfeta_0)\} \hat a(\beta, \hat\bfeta)  | \\
    &\leq \|\hat a(\beta, \hat\bfeta)\|^2 \left\{\| \hat\Omega(\beta, \hat\bfeta) - \Omega(\beta, \bfeta_0)\| + C\| \hat\Omega(\beta, \hat\bfeta) - \Omega(\beta, \bfeta_0)\|^2 \right\} = o_p(1).
\end{align*}
Consequently, we have shown that 
\begin{align}
    \sup_{\beta\in B} \mu_n^{-2} n |\hat Q (\beta, \hat\bfeta)- \tilde{Q} (\beta, \hat\bfeta) | =o_p(1).
\end{align}
Next, we show that 
\[
\sup_{\beta\in B} \mu_n^{-2} n |\tilde Q (\beta, \hat\bfeta)- \tilde{Q} (\beta, \bfeta_0) | =o_p(1).
\]
This can be easily seen as  
\begin{align*}
  & 2 \mu_n^{-2} n |\tilde Q (\beta, \hat\bfeta)- \tilde{Q} (\beta, \bfeta_0) | \\
  &= \mu_n^{-2}n|\hat g(\beta, \hat\bfeta)^T \Omega(\beta, \bfeta_0)^{-1} \hat g(\beta, \hat\bfeta) - \hat g(\beta, \bfeta_0)^T \Omega(\beta, \bfeta_0) ^{-1}\hat g(\beta, \bfeta_0)| \\
   &=   \mu_n^{-2}n|\{\hat g(\beta, \hat\bfeta)-\hat g(\beta, \bfeta_0) \}^T \Omega(\beta, \bfeta_0) ^{-1} \hat g(\beta, \hat\bfeta) +  \hat g(\beta, \bfeta_0)^T \Omega(\beta, \bfeta_0) ^{-1} \{\hat g(\beta, \hat\bfeta) - \hat g(\beta, \bfeta_0)\} | \\
   &\leq  \mu_n^{-2}n|\{\hat g(\beta, \hat\bfeta)-\hat g(\beta, \bfeta_0) \}^T \Omega(\beta, \bfeta_0)^{-1} \hat g(\beta, \hat\bfeta) |+\mu_n^{-2}n | \hat g(\beta, \bfeta_0)^T \Omega(\beta, \bfeta_0) ^{-1} \{\hat g(\beta, \hat\bfeta) - \hat g(\beta, \bfeta_0)\} | \\
   &\leq  \mu_n^{-2}n C \|\hat g(\beta, \hat\bfeta)-\hat g(\beta, \bfeta_0) \| \|\hat g(\beta, \hat\bfeta) \|+\mu_n^{-2}n C \| \hat g(\beta, \bfeta_0)\| \| \hat g(\beta, \hat\bfeta) - \hat g(\beta, \bfeta_0) \|\\
   &\leq \mu_n^{-2} n  C\sup_{\beta\in B} \|\hat g(\beta, \hat\bfeta)-\hat g(\beta, \bfeta_0) \| \sup_{\beta\in B} \|\hat g(\beta, \hat\bfeta) \| \\
   &\qquad \qquad + \mu_n^{-2} n C\sup_{\beta\in B}\| \hat g(\beta, \bfeta_0)\| \sup_{\beta\in B} \| \hat g(\beta, \hat\bfeta) -\hat  g(\beta, \bfeta_0) \| \\
   & = o_p(1)
\end{align*}
where the last line is from Lemma \ref{lemma: g beta, nuisance},  $\sup_{\beta\in B}\| \hat g(\beta, \bfeta_0)\|= O_p(\mu_n/\sqrt{n})$ and $\sup_{\beta\in B}\| \hat g(\beta, \hat\bfeta)\|= O_p(\mu_n/\sqrt{n})$ from \eqref{eq: g hat bdd in prob}, and $m^2/n\rightarrow 0$. 

Finally, it remains to show that 
\[
\sup_{\beta\in B} \mu_n^{-2} n |\tilde Q (\beta, \bfeta_0)- {Q} (\beta, \bfeta_0) | =o_p(1).
\]
For $\beta, \beta'\in B$, let  $Q(\beta', \beta, \bfeta_0)= E\{g_i(\beta', \bfeta_0)^T\} \Omega(\beta, \bfeta_0)^{-1} E\{g_i(\beta', \bfeta_0)\}/2+ m/(2n) $ and  $a(\beta', \beta, \bfeta_0)= \mu_n^{-1} \sqrt{n} \Omega(\beta, \bfeta_0)^{-1} E\{g_i(\beta', \bfeta_0)\} $. By Assumption \ref{assump:omega} and Lemma \ref{lemma: g bdd by beta}, \\ $\sup_{\beta\in B, \beta'\in B }\|a(\beta', \beta, \bfeta_0) \|\leq C $. Then, by Lemma \ref{lemma: Omega multiply},  it follows that 
\begin{align*}
   & \mu_n^{-2}n |Q(\beta', \beta', \bfeta_0) -  Q(\beta', \beta, \bfeta_0)  |= |a(\beta', \beta', \bfeta_0)^T\{\Omega(\beta', \bfeta_0) - \Omega(\beta, \bfeta_0)\}  a(\beta', \beta, \bfeta_0) | \\
    & \leq C |\beta'-\beta|.
\end{align*}
Also, by triangle inequality, Assumption \ref{assump:omega} and Lemma \ref{lemma: g bdd by beta}, 
\begin{align*}
    &\mu_n^{-2} n |Q(\beta', \beta, \bfeta_0) -  Q(\beta, \beta, \bfeta_0)  | \\
    &\leq C\mu_n^{-2} n \left[\|E\{g_i(\beta',\bfeta_0) \} -E\{g_i(\beta,\bfeta_0) \} \| + \|E\{g_i(\beta',\bfeta_0)\| \|E\{g_i(\beta',\bfeta_0) \} -E\{g_i(\beta,\bfeta_0) \} \| \right] \\
    &\leq C|\beta'-\beta|.
\end{align*}
Then by triangle inequality, it follows that $\mu_n^{-2}n |Q(\beta', \bfeta_0)- Q(\beta, \bfeta_0)|= \mu_n^{-2}n |Q(\beta', \beta',\bfeta_0)- Q(\beta, \beta, \bfeta_0)|\leq C|\beta'-\beta|$. Therefore, $\mu_n^{-2}n Q(\beta,\bfeta_0) $ is equicontinuous for $\beta,\beta'\in B$. An analogous argument with $\tilde Q(\beta', \beta, \bfeta_0)= \hat g(\beta', \bfeta_0)^T \Omega(\beta, \bfeta_0)^{-1} \hat g(\beta', \bfeta_0)/2 $ and $\hat a(\beta',\beta, \bfeta_0) =\mu_n^{-1}\sqrt{n}\Omega(\beta, \bfeta_0) \hat g(\beta', \bfeta_0) $ replacing $Q(\beta', \beta, \bfeta_0)$ and $a(\beta',\beta)$, respectively, implies that \\ $\mu_n^{-2} n |\tilde Q(\beta', \bfeta_0)- \tilde Q(\beta, \bfeta_0) | \leq \hat M \|\tilde \beta - \beta\|$ for $\beta,\beta'\in B$, with $\hat M= O_p(1)$, giving stochastic equicontinuity of $\mu_n^{-2} n \tilde Q(\beta, \bfeta_0)$.

Since $\mu_n^{-2}n Q(\beta,\bfeta_0) $ and $\mu_n^{-2} n \tilde Q(\beta, \bfeta_0)$ are stochastically equicontinuous, it suffices by Theorem 2.1 of \cite{Newey1991} to show that 
\[
\mu_n^{-2}n \tilde Q(\beta, \bfeta_0)= \mu_n^{-2}n  Q(\beta, \bfeta_0) +o_p(1)
\]
for each $\beta$. Applying Lemma A1 of \cite{Newey:2009aa} with $Y_i=Z_i=g_i(\beta, \bfeta_0)$, $A= \Omega(\beta, \bfeta_0)^{-1}$, and $a_n=\mu_n^2$. By Assumption \ref{assump:omega}, $\xi_{\max}(A^T A)=\xi_{\max}(A A^T) = \xi_{\max} \{\Omega(\beta, \bfeta_0)^{-2} \}\leq C$, $\xi_{\max}(\Sigma_{YY}) = \xi_{\max} \{\Omega(\beta, \bfeta_0)\}\leq C$, $E\{(Y_i^TY_i)^2\}/(na_n^2)=$\\ $ E[\{g_i(\beta, \bfeta_0)^Tg_i(\beta, \bfeta_0) \}^2 ]/(n\mu_n^4)\rightarrow 0$ from Lemma \ref{lemma: sup g square},   and $n\mu_Y^T\mu_Y/a_n^2\leq C \{ n Q(\beta, \bfeta_0)/\mu_n^2- m/\mu_n^2\}/\mu_n^2\rightarrow 0 $ from the equicontinuity of $\mu_n^{-2}n Q(\beta, \bfeta_0) $. Thus, the conditions of Lemma A1 of \cite{Newey:2009aa} are satisfied. Note that $A\Sigma_{YZ}^T = A\Sigma_{ZZ}=  A\Sigma_{YY} = mI_m\mu_n^2$, so by the Lemma A1, 
\begin{align*}
    &\mu_n^{-2} n \tilde Q(\beta, \bfeta_0)= tr(I_m)/\mu_n^2 +\mu_n^{-2} n E\{g_i(\beta, \bfeta_0)^T\} \Omega(\beta, \bfeta_0)^{-1} E\{g_i(\beta, \bfeta_0)\} + o_p(1) \\
    &= \mu_n^{-2} n Q(\beta, \bfeta_0)+ o_p(1).
\end{align*}
This completes the proof.
\end{proof}

\begin{lemma} \label{lemma: Q tilde} Under Assumptions \ref{assump: many weak moments}-\ref{assump: 4th moment},  and $ m^3/n\rightarrow 0  $, 
	   $$n\mu_n^{-1}\partial \hat Q(\beta, \hat\bfeta)/\partial \beta |_{\beta=\beta_0}=n\mu_n^{-1}\partial \tilde  Q(\beta, \hat\bfeta)/\partial \beta |_{\beta=\beta_0} +o_p(1).$$  
\end{lemma}
\begin{proof}
Notice that 
\begin{align*}
	\frac{\partial \Omega(\beta,  \bfeta_0)^{-1}}{\partial \beta} \big|_{\beta= \beta_0}= - \Omega^{-1}\left[ \frac{\partial \Omega(\beta, \bfeta_0)}{\partial \beta}\right] \big|_{\beta=\beta_0}  \Omega^{-1}= - \Omega^{-1} E\big\{ g_i G_i^T + G_i g_i^T\big\} \Omega^{-1}. 
\end{align*}
Recall that $\tilde{Q}(\beta, \hat \bfeta) = \hat g (\beta, \hat\bfeta)^T \Omega (\beta, \bfeta_0)^{-1} \hat g (\beta, \hat\bfeta)/2   $, which is the same with $\hat{Q}(\beta, \hat\bfeta)   $ but with $ \hat \Omega({\beta}, \hat \bfeta) $ replaced by  $ \Omega(\beta, \bfeta_0) $. Differentiating $\tilde{Q}(\beta, \hat\bfeta) $ with respect to $\beta$, we have
\begin{align*}
	\frac{\partial \tilde{Q} (\beta, \hat \bfeta)}{\partial \beta}  \big|_{\beta= \beta_0} &=\hat{g} (\beta_0, \hat\bfeta)^T \Omega^{-1} \frac{1}{n} \sum_{i=1}^n G_i( \hat\bfeta)  - \frac{1}{2} \hat{g} (\beta_0, \hat\bfeta)^T \Omega^{-1} E\big\{ g_i G_i^T + G_i g_i^T\big\} \Omega^{-1}  \hat{g} (\beta_0, \hat\bfeta)\\
	&= \hat{g} (\beta_0, \hat\bfeta)^T \Omega^{-1} \frac{1}{n} \sum_{i=1}^n G_i(\hat\bfeta) -  \hat{g} (\beta_0, \hat\bfeta)^T \Omega^{-1} E\big\{  G_i g_i^T\big\} \Omega^{-1}  \hat{g} (\beta_0, \hat\bfeta)\\
	&= \hat{g} (\beta_0, \hat\bfeta)^T \Omega^{-1} \frac{1}{n} \sum_{i=1}^n \bigg[ G_i(\hat \bfeta)- E\big\{  G_i g_i^T\big\} \Omega^{-1}  g_i (\beta_0, \hat\bfeta)\bigg]\\
	&= \frac{1}{n^2} \sum_{i, j=1}^n \left\{ G+ \underbrace{G_i(\hat\bfeta) - G - E(G_i g_i^T) \Omega^{-1} g_i(\beta_0, \hat\bfeta)}_{\tilde U_i} \right\}^T \Omega^{-1} g_j (\beta_0, \hat{\bfeta})\\
	&= G^T \Omega^{-1} \hat g(\beta_0, \hat\bfeta) + \left\{n^{-1} \sum_{i=1}^n \tilde U_i \right\}^T \Omega^{-1} \hat g(\beta_0, \hat \bfeta).
\end{align*}
Similarly, we can derive that 
\begin{align*}
&	\frac{\partial \hat{Q} (\beta, \hat \bfeta)}{\partial \beta}  \big|_{\beta= \beta_0} =  G^T \hat \Omega(\beta_0, \hat\bfeta)^{-1} \hat g(\beta_0, \hat\bfeta) + \left\{n^{-1} \sum_{i=1}^n \check U_i \right\}^T \hat\Omega(\beta_0, \hat\bfeta)^{-1} \hat g(\beta_0, \hat \bfeta),
\end{align*}
where 
\[
\check U_i=  G_i(\hat\bfeta)- G- \left( n^{-1} \sum_{i=1}^n G_i(\hat\bfeta)g_i(\beta_0, \hat\bfeta)^T \right) \hat\Omega(\beta_0, \hat\bfeta)^{-1} g_i(\beta_0, \hat\bfeta).
\]
Hence, it suffices to show that\\
(a) $n\mu_n^{-1}( n^{-1} \sum_{i=1}^n \check U_i- \tilde U_i)^T \hat \Omega(\beta_0,\hat\bfeta)^{-1} \hat g(\beta_0,\hat\bfeta)=o_p(1)$;\\
(b) $n\mu_n^{-1}( n^{-1} \sum_{i=1}^n\tilde U_i)^T \{\hat \Omega(\beta_0,\hat\bfeta)^{-1} - \Omega^{-1}\}\hat g(\beta_0,\hat\bfeta)=o_p(1)$; \\
(c) $n\mu_n^{-1} G^T \big\{\hat \Omega(\beta_0,\hat\bfeta)^{-1} - \Omega^{-1} \big\} \hat g(\beta_0,\hat\bfeta)=o_p(1) $.

For part (a), we have 
\begin{align*}
     &n\mu_n^{-1}( n^{-1} \sum_{i=1}^n \check U_i- \tilde U_i)^T \hat \Omega(\beta_0,\hat\bfeta)^{-1} \hat g(\beta_0,\hat\bfeta)\\
     &= n\mu_n^{-1} \hat g(\beta_0,\hat\bfeta) \left\{E(G_ig_i^T) \Omega^{-1}- \left(n^{-1} \sum_{i=1}^n G_i(\hat\bfeta ) g_i(\beta_0, \hat\bfeta)^T \right) \hat\Omega(\beta_0,\hat\bfeta)^{-1}  \right\}^T \hat \Omega(\beta_0,\hat\bfeta)^{-1} \hat g(\beta_0,\hat\bfeta) \\
     &\leq  C \mu_n^{-1} n \|\hat g(\beta_0,\hat\bfeta)\|^2 \left\|E(G_ig_i^T) \Omega^{-1}- \left(n^{-1} \sum_{i=1}^n G_i(\hat\bfeta ) g_i(\beta_0, \hat\bfeta)^T \right) \hat\Omega(\beta_0,\hat\bfeta)^{-1}  \right \|\\
     &=    C \mu_n^{-1} n \|\hat g(\beta_0,\hat\bfeta)\|^2 \\
     &\qquad \left\|E(G_ig_i^T) \{\Omega^{-1} - \hat\Omega(\beta_0,\hat \bfeta)^{-1} \}- \left(n^{-1} \sum_{i=1}^n G_i(\hat\bfeta ) g_i(\beta_0, \hat\bfeta)^T -E(G_ig_i^T)   \right) \hat\Omega(\beta_0,\hat\bfeta)^{-1}  \right \| \\
     &\leq  C \mu_n^{-1} n \|\hat g(\beta_0,\hat\bfeta)\|^2  \left\|E(G_ig_i^T) \{\Omega^{-1} - \hat\Omega(\beta_0,\hat \bfeta)^{-1} \}\right\| \\
     &\qquad+C \mu_n^{-1} n \|\hat g(\beta_0,\hat\bfeta)\|^2  \left\| \left(n^{-1} \sum_{i=1}^n G_i(\hat\bfeta ) g_i(\beta_0, \hat\bfeta)^T -E(G_ig_i^T)   \right) \hat\Omega(\beta_0,\hat\bfeta)^{-1}  \right \| \\
       &\leq  C \mu_n^{-1} n \|\hat g(\beta_0,\hat\bfeta)\|^2  \left\|  \Omega^{-1}  \{\hat\Omega(\beta_0,\hat \bfeta) - \Omega \}\hat\Omega(\beta_0,\hat \bfeta)^{-1}  \right\| \\
     &\qquad+C \mu_n^{-1} n \|\hat g(\beta_0,\hat\bfeta)\|^2 \left\| \left(n^{-1} \sum_{i=1}^n G_i(\hat\bfeta ) g_i(\beta_0, \hat\bfeta)^T -E(G_ig_i^T)   \right) \hat\Omega(\beta_0,\hat\bfeta)^{-1}  \right \| \\
        &\leq  C \mu_n^{-1} \underbrace{n \|\hat g(\beta_0,\hat\bfeta)\|^2}_{O_p(m)} \left\{ \left\|   \hat\Omega(\beta_0,\hat \bfeta) - \Omega  \right\| + \left\| n^{-1} \sum_{i=1}^n G_i(\hat\bfeta ) g_i(\beta_0, \hat\bfeta)^T -E(G_ig_i^T)    \right \| \right\}\\
     &=o_p(1)
\end{align*}
 using $ n \|\hat g(\beta_0,\hat\bfeta)\|^2= O_p(m) $ shown in   \eqref{eq: g hat bdd in prob},  $ \xi_{\max} \{ \hat\Omega(\beta_0, \hat \bfeta)^{-1} \}\leq C  $ shown below \eqref{eq: Omega hat bdd}, $ \xi_{\max} ( \Omega^{-1} )\leq C  $, 
 $\xi_{\max}\{E(G_ig_i^T)\} \leq C  $,  $ \xi_{\max}  \{  E(G_iG_i^T) \}< C,$ and $ \xi_{\max} \{ E(g_ig_i^T)\}\leq C$ from Assumption \ref{assump:omega},   Lemma \ref{lemma: Omega eta hat}(i)-(ii), and  $m/\mu_n^2\leq C$ from Assumption  \ref{assump: many weak moments}. 

For part (b), notice first that from $U_i$ being the residual and Assumption \ref{assump: 4th moment}, we have $E(U_i)=0$ and $E(\|U_i\|^2) \leq E(\|G_i\|^2) \leq Cm  $. From Markov inequality, $\frac{1}{\sqrt{nm}} \| \sum_{i=1}^n U_i \|=  O_p(1)$. Moreover, 
\begin{align}
    &\left\|\frac{1}{\sqrt{n}} \sum_{i=1}^n ( \tilde U_i- U_i) \right\| \label{eq: U_i}\\
    &=  \left\| \frac{1}{\sqrt{n}} \sum_{i=1}^n \big\{G_i(\hat\bfeta) - G_i -  E(G_ig_i^T)\Omega^{-1} (g_i(\beta_0, \hat\bfeta) - g_i) \big\} \right\| \nonumber \\
    &=  \left\| \frac{1}{\sqrt{n}} \sum_{i=1}^n \big\{G_i( \hat\bfeta) - G_i \big\} \right\| +\left\| \frac{1}{\sqrt{n}} \sum_{i=1}^n  E(G_ig_i^T)\Omega^{-1} (g_i(\beta_0, \hat\bfeta) - g_i)  \right\| \nonumber\\
     &\leq  \left\|  \frac{1}{\sqrt{n}} \sum_{i=1}^n \big\{G_i(\hat\bfeta) - G_i \big\} \right\| + C\sqrt{n} \left\| \hat g_i(\beta_0, \hat\bfeta) - \hat g(\beta_0,\bfeta_0)  \right\| \nonumber\\
     &= o_p(1) \nonumber
\end{align}
where the fourth line is because $\xi_{\max}(\Omega^{-1})\leq C$  and $\xi_{\max}\{E(G_ig_i^T)\} \leq C  $, the last line  is from Lemma \ref{lemma: g beta, nuisance}(i)-(ii).   Hence, $\frac{1}{\sqrt{nm} } \|\sum_{i=1}^n \tilde U_i \| =  O_p(1)$. In consequence, by Lemma \ref{lemma: Omega eta hat}(i), and \eqref{eq: g hat bdd in prob},  
\begin{align*}
	&n\mu_n^{-1}( n^{-1} \sum_{i=1}^n\tilde U_i)^T \{\hat \Omega(\beta_0,\hat\bfeta)^{-1} - \Omega^{-1}\}\hat g(\beta_0,\hat\bfeta)\\
     &=  \big( \frac{1}{\sqrt{nm}}\sum_{i=1}^n\tilde U_i\big)^T \sqrt{m}\{\hat \Omega(\beta_0,\hat\bfeta)^{-1} - \Omega^{-1}\} \mu_n^{-1}\sqrt{n}\hat g(\beta_0,\hat\bfeta) \\
     &\leq \underbrace{\left\| \frac{1}{\sqrt{nm}}\sum_{i=1}^n\tilde U_i\right\|}_{O_p(1)} \underbrace{\left\| \sqrt{m}\{\hat \Omega(\beta_0,\hat\bfeta)^{-1} - \Omega^{-1}\} \right\|}_{o_p(1)} \underbrace{\left\| \mu_n^{-1}\sqrt{n}\hat g(\beta_0,\hat\bfeta) \right\|}_{O_p(1)}\\
     &= o_p(1).
\end{align*}

Part (c) can be shown in a similar fashion as the Part (b). Specifically,
\begin{align*}
    &n\mu_n^{-1} G^T \big\{\hat \Omega(\beta_0,\hat\bfeta)^{-1} - \Omega^{-1} \big\} \hat g(\beta_0,\hat\bfeta) \\
    &\leq \underbrace{\|\sqrt{n} \mu_n^{-1} G \|}_{O(1)} \underbrace{\| \sqrt{m}\big\{\hat \Omega(\beta_0,\hat\bfeta)^{-1} - \Omega^{-1} \big\} \|}_{o_p(1) } \underbrace{\| \sqrt{n}m^{-1/2}\hat g(\beta_0,\hat\bfeta) \|}_{O_p(1)}= o_p(1).
\end{align*}

\end{proof}

\begin{lemma} \label{lemma: Q second deriv} Under Assumptions \ref{assump: many weak moments}-\ref{assump: 4th moment},  and $ m^3/n\rightarrow 0  $, 
\[
 n \mu_n^{-2}  \sup_{\beta\in B}  \left| \frac{\partial^2 \hat Q (\beta, \hat\bfeta)}{\partial \beta^2}-   \frac{\partial^2 \hat Q (\beta, \bfeta_0)}{\partial \beta^2}  \right|  = o_p(1).
\]
\end{lemma}
\begin{proof}
	Recall that we have shown $ \sup_{\beta\in B} \|\hat g(\beta, \hat\bfeta)\| = O_p(\mu_n/\sqrt{n})$  and $  \sup_{\beta\in B} \|\hat g(\beta, \bfeta_0)\| = O_p(\mu_n/\sqrt{n})  $ in \eqref{eq: g hat bdd in prob}. Similar to the proof of \eqref{eq: g hat bdd in prob}, we can show that $ n E\{ \|\hat G(\bfeta_0) \|^2 
	\}/m \leq n E\{ \|\hat G(\bfeta_0) - G \|^2 
	\}/m + n \| G \|^2 
/m  = E\left\{G_i^T G_i\right\}/m + \Theta (\mu_n^2/m)\leq tr\{E(G_iG_i^T )\}/m+ \Theta (\mu_n^2/m) \leq C+ \Theta (\mu_n^2/m)$ from Assumption \ref{assump:omega}(i) and  
	\begin{align*}
		\|\hat G(\hat \bfeta)\| = \|\hat G(\hat \bfeta) - \hat G(\bfeta_0) \| + \| \hat G(\bfeta_0) \| = o_p(n^{-1/2}) + O_p(\sqrt{\mu_n^2/n}) = O_p(\sqrt{\mu_n^2/n}). 
	\end{align*}
	Then we calculate 
	\begin{align*}
		\frac{\partial \hat\Omega(\beta, \bfeta)}{\partial \beta} &= \frac{\partial }{\partial \beta} \left\{ \frac1n \sum_{i=1}^n g_i(\beta, \bfeta) g_i(\beta, \bfeta)^T \right\} = \frac1n \sum_{i=1}^{n} \left\{  G_i(\bfeta) g_i(\beta, \bfeta)^T  +g_i(\beta, \bfeta) G_i(\bfeta)^T \right\}, \\
		\frac{\partial^2 \hat\Omega(\beta, \bfeta)}{\partial \beta^2} &=  \frac2n \sum_{i=1}^n G_i( \bfeta) G_i(\bfeta)^T,
	\end{align*}
and 
	\begin{align*}
		\frac{\partial^2 \hat Q(\beta, \bfeta)}{\partial \beta^2} &= -2 \hat G (\bfeta)^T  \hat \Omega(\beta, \bfeta)^{-1} \frac{\partial \hat\Omega(\beta, \bfeta)}{\partial \beta}  \hat \Omega(\beta, \bfeta)^{-1}\hat g(\beta, \bfeta) \\
		&\qquad + \hat G(\bfeta)^T  \hat \Omega(\beta, \bfeta)^{-1}  \hat G(\bfeta) \\
		&\qquad+ \hat g(\beta, \bfeta)^T \hat \Omega(\beta, \bfeta)^{-1}  \frac{\partial \hat\Omega(\beta, \bfeta)}{\partial \beta} \hat \Omega(\beta, \bfeta)^{-1}  \frac{\partial \hat\Omega(\beta, \bfeta)}{\partial \beta} \hat \Omega(\beta, \bfeta)^{-1}  \hat g(\beta, \bfeta) \\
		&\qquad- \frac{1}{2} \hat g(\beta, \bfeta)^T  \hat \Omega(\beta, \bfeta)^{-1} \frac{\partial^2 \hat\Omega(\beta, \bfeta)}{\partial \beta^2} \hat \Omega(\beta, \bfeta)^{-1}  \hat g(\beta, \bfeta)\\
		&:= -2 J_1(\beta, \bfeta) + J_2(\beta, \bfeta) + J_3(\beta, \bfeta) -\frac12 J_4(\beta, \bfeta).
	\end{align*}

Next, we show that $n\mu_n^{-2} \sup_{\beta\in B} \|J_1(\beta, \hat \bfeta)- J_1(\beta, \bfeta_0)\|= o_p(1)$. Note that 
\begin{align*}
	&J_1(\beta, \hat\bfeta)- 	J_1(\beta, \bfeta_0) \\
	&=  \hat G (\hat \bfeta)^T  \hat \Omega(\beta, \hat\bfeta)^{-1} \frac{\partial \hat\Omega(\beta,\hat \bfeta)}{\partial \beta}  \hat \Omega(\beta, \hat\bfeta)^{-1}\hat g(\beta, \hat\bfeta) - \hat G (\bfeta_0)^T  \hat \Omega(\beta, \bfeta_0)^{-1} \frac{\partial \hat\Omega(\beta, \bfeta_0)}{\partial \beta}  \hat \Omega(\beta, \bfeta_0)^{-1}\hat g(\beta, \bfeta_0) \\
	&= \hat G(\hat \bfeta)^T    \hat\Omega(\beta, \hat\bfeta)^{-1} \left\{ \frac{\partial \hat \Omega(\beta, \hat \bfeta)}{\partial \beta} -  \frac{\partial \hat\Omega(\beta, \bfeta_0)}{\partial \beta} \right\}    \hat\Omega(\beta, \hat\bfeta)^{-1} \hat g(\beta, \hat\bfeta)\\
	&\qquad + \left\{ \hat G(\hat \bfeta) - \hat G(\bfeta_0) \right\}^T   \hat\Omega(\beta, \hat\bfeta)^{-1}   \frac{\partial \hat\Omega(\beta, \bfeta_0)}{\partial \beta}  \hat\Omega(\beta, \hat\bfeta)^{-1} \hat g(\beta, \hat\bfeta)\\
	&\qquad+  \hat G(\bfeta_0) ^T   \hat\Omega(\beta, \hat\bfeta)^{-1}   \frac{\partial \hat\Omega(\beta, \bfeta_0)}{\partial \beta}  \hat\Omega(\beta, \hat\bfeta)^{-1}  \left\{\hat g(\beta, \hat\bfeta) - \hat g(\beta, \bfeta_0) \right\}\\
	& \qquad+  \hat G(\bfeta_0) ^T   \hat\Omega(\beta, \hat\bfeta)^{-1}   \frac{\partial \hat\Omega(\beta, \bfeta_0)}{\partial \beta}  \{\hat\Omega(\beta, \hat\bfeta)^{-1} - \hat\Omega(\beta, \bfeta_0)^{-1}\}  \hat g(\beta, \bfeta_0) \\
	&\qquad + \hat G(\bfeta_0) ^T   \{\hat\Omega(\beta, \hat\bfeta)^{-1} - \hat\Omega(\beta, \bfeta_0)^{-1}\}   \frac{\partial \hat\Omega(\beta, \bfeta_0)}{\partial \beta}  \hat\Omega(\beta, \bfeta_0)^{-1}   \hat g(\beta, \bfeta_0). \\
	&: = J_{11} + J_{12} + J_{13} + J_{14} + J_{15}. 
\end{align*}
For term $ J_{11} $, from Lemma \ref{lemma: Omega eta hat}(v), 
\begin{align*}
	&\sup_{\beta\in B}\left\|  \frac{\partial \hat \Omega(\beta, \hat \bfeta)}{\partial \beta} -  \frac{\partial \hat\Omega(\beta, \bfeta_0)}{\partial \beta}   \right\|\leq 2\sup_{\beta\in B} \left\| \frac1n \sum_{i=1}^{n} G_i(\hat\bfeta) g_i(\beta, \hat\bfeta)^T - G_i(\bfeta_0) g_i(\beta, \bfeta_0)^T   \right\| = o_p(m^{-1/2}).
\end{align*}
Hence, from $  \xi_{\max} (\hat\Omega(\beta, \hat\bfeta)^{-1}) \leq C$  for all $ \beta\in B $, w.p.a.1 in $ \eqref{eq: Omega hat bdd} $, we have that 
	\begin{align*}
&\sup_{\beta\in B}	\|J_{11}\| \leq C \|\hat G(\hat\bfeta)\| \sup_{\beta\in B}  \left\| \frac{\partial \hat \Omega(\beta, \hat \bfeta)}{\partial \beta} -  \frac{\partial \hat\Omega(\beta, \bfeta_0)}{\partial \beta}  \right\| \sup_{\beta\in B} \|\hat g(\beta, \hat\bfeta)\| \\
&= O_p(\sqrt{\mu_n^2/n})  o_p(m^{-1/2}) O_p(\mu_n/\sqrt{n})  = o_p(\mu_n^2/(n\sqrt{m})). 
\end{align*}

For term $ J_{12} $, from $ \xi_{\max} \{ E(G_ig_i(\beta, \bfeta_0)^T)\}\leq C $, and Lemma \ref{lemma: Omega eta hat},  
we have that w.p.a.1, $ \xi_{\max}\left\{  \partial \hat\Omega(\beta, \bfeta_0)/\partial \beta \right\} \leq C $ for all $ \beta\in B $. Then, from Lemma \ref{lemma: g beta, nuisance}, we have that 
	\begin{align*}
	\sup_{\beta\in B}	\|J_{12}\| \leq  C\| \hat G(\hat \bfeta) - \hat G(\bfeta_0) \| \sup_{\beta\in B}  \|\hat g(\beta, \hat\bfeta)\| =  o_p (n^{-1/2}) O_p(\mu_n/\sqrt{n})  = o_p(\mu_n/n). 
 \end{align*}

For term $ J_{13} $, 
\begin{align*}
\sup_{\beta\in B}	\|J_{13}\|  \leq  C \|	\hat G(\bfeta_0) \| \sup_{\beta\in B} \|\hat g(\beta, \hat\bfeta) - \hat g(\beta, \bfeta_0) \|= O_p(\sqrt{\mu_n^2/n}) o_p(n^{-1/2}) = o_p(\mu_n/n). 
\end{align*}

For term $ J_{14} $ 
\begin{align*}
	&\sup_{\beta\in B}	\|J_{14}\|  \leq C \| \hat G(\bfeta_0)\|   \sup_{\beta\in B}   \|\hat\Omega(\beta, \hat\bfeta) - \hat\Omega(\beta, \bfeta_0)\|  \sup_{\beta\in B}  \|\hat g(\beta, \bfeta_0) \| \\
	&=  O_p(\sqrt{\mu_n^2/n}) o_p(m^{-1/2})  O_p(\mu_n/\sqrt{n})  = o_p(\mu_n^2/(n\sqrt{m})).
\end{align*}
The term  $ J_{15} $ is  bounded by the same factor. Therefore, 
\begin{align*}
n \mu_n^{-2}\sup_{\beta\in B}\| J_1(\beta, \hat\bfeta)- 	J_1(\beta, \bfeta_0)  \| = o_p(1). 
\end{align*}
Then, it follows by arguments exactly analogous to those just given that 
\begin{align*}
	&n \mu_n^{-2}\sup_{\beta\in B}\| J_2(\beta, \hat\bfeta)- 	J_2(\beta, \bfeta_0)  \| = o_p(1), \qquad 	n \mu_n^{-2}\sup_{\beta\in B}\| J_3(\beta, \hat\bfeta)- 	J_3(\beta, \bfeta_0)  \| = o_p(1) \\
	&n \mu_n^{-2}\sup_{\beta\in B}\| J_4(\beta, \hat\bfeta)- 	J_4(\beta, \bfeta_0)  \| = o_p(1), 
\end{align*}
which completes the proof.

\end{proof}

\subsection{Proof of Theorem \ref{theo: GMM}} \label{sec: proof of theo gmm}
\subsubsection{Consistency} \label{subsec: consistency}
From Lemma \ref{lemma: cond}, it suffices to show that $
\mu_n^{-1}\sqrt{n} \| \bar g (\hat\beta, \bfeta_0)\| =o_p(1)$, where  $\bar g(\beta, \bfeta)= E\{ g_i (\beta, \bfeta)\}$.

First notice that from definition,  we have 
\begin{align*}
  \mu_n^{-2} n \hat Q(\hat\beta, \hat\bfeta) \leq \mu_n^{-2} n \hat Q(\beta_0, \hat\bfeta). 
\end{align*}
Consider any $\epsilon, \delta>0$.  By Lemma \ref{lemma:unif Q}, we have
\begin{align}
    \mu_n^{-2}n Q(\hat\beta, \bfeta_0)&\leq   \mu_n^{-2}n \hat Q(\hat \beta, \hat\bfeta) +o_p(1) \leq \mu_n^{-2} n \hat Q( \beta_0, \hat\bfeta) +o_p(1) \leq \mu_n^{-2} n  Q( \beta_0, \bfeta_0) +o_p(1). 
\end{align}
Hence, $\mu_n^{-2} n \{Q(\hat\beta, \bfeta_0) - Q(\beta_0, \bfeta_0)\}=o_p(1)$. By Assumption \ref{assump:omega}, we further have that
\begin{align*}
   &\mu_n^{-2} n \{Q(\hat\beta, \bfeta_0) - Q(\beta_0, \bfeta_0)\}  = \mu_n^{-2} n \{Q(\hat\beta, \bfeta_0) - m/(2n)\} \\
   &=\mu_n^{-2} n  \bar g (\hat\beta, \bfeta_0) \Omega (\hat\beta, \bfeta_0)^{-1} \bar g (\hat\beta, \bfeta_0) \geq C \mu_n^{-2}n \| \bar g (\hat\beta, \bfeta_0)\|^2.
\end{align*} 
Hence, $
\mu_n^{-1}\sqrt{n} \| \bar g (\hat\beta, \bfeta_0)\| =o_p(1)$. 

\subsubsection{Asymptotic Normality}\label{subsec: normal}
From Taylor expansion of the first order condition $\partial \hat Q(\beta, \hat\bfeta)/\partial \beta |_{\beta=\hat\beta}=0$, we have that 
\[
0= \frac{\partial \hat Q(\beta, \hat\bfeta)}{\partial \beta} \bigg|_{\beta=\hat\beta} = \frac{\partial \hat Q(\beta, \hat\bfeta)}{\partial \beta}\bigg|_{\beta=\beta_0} + \frac{\partial^2 \hat Q (\beta, \hat\bfeta)}{\partial \beta^2} \bigg|_{\beta=\bar\beta} (\hat\beta- \beta_0)
\]
where $\bar\beta$ is some value between $\beta_0$ and $\hat\beta$. We first analyze the term $\partial \hat Q(\beta, \hat\bfeta)/\partial \beta|_{\beta=\beta_0} $.

According to Lemma \ref{lemma: Q tilde}, we have 
\begin{align*}
		&n \mu_n^{-1}\frac{\partial \hat{Q} (\beta, \hat \bfeta)}{\partial \beta}  \big|_{\beta= \beta_0}= n \mu_n^{-1}\frac{\partial \tilde{Q} (\beta, \hat \bfeta)}{\partial \beta}  \big|_{\beta= \beta_0}+o_p(1)\\
	&= \underbrace{n\mu_n^{-1} G^T \Omega^{-1} \hat{g} (\beta_0, \hat\bfeta)}_{A1} + \underbrace{\frac{1}{n\mu_n} \sum_{i,j=1}^n \tilde U_i^T \Omega^{-1} (g_j(\beta_0, \hat \bfeta) - g_j)}_{A2} + \underbrace{\frac{1}{n\mu_n} \sum_{i,j=1}^n \tilde U_i^T \Omega^{-1} g_j}_{A3} +o_p(1)
\end{align*}
where $\tilde U_i= G_i( \hat\bfeta) - G - E(G_i g_i^T) \Omega^{-1} g_i(\beta_0, \hat\bfeta)$.

We analyze the three terms individually. For the first term,
\begin{align*}
	A1&= n\mu_n^{-1} G^T \Omega^{-1} \hat{g} (\beta_0, \hat\bfeta)\\
	&=n\mu_n^{-1} G^T \Omega^{-1} \hat{g} (\beta_0, \bfeta_0) + n \mu_n^{-1} G^T \Omega^{-1}\big\{\hat g(\beta_0, \hat \bfeta)- \hat g(\beta_0, \bfeta_0)\big\}. \\
	&= \sqrt{n}\mu_n^{-1} G^T \Omega^{-1} \sqrt{n}\ \hat{g} (\beta_0, \bfeta_0)+ o_p(1),
\end{align*}
where the last expression is from $ \| \sqrt{n} \big\{  \hat g(\beta_0, \hat \bfeta)- \hat g(\beta_0, \bfeta_0\big\}  \| =o_p(1)$ from Lemma \ref{lemma: g beta, nuisance}(i)  and $  \| \sqrt{n}\mu_n^{-1} G^T\Omega^{-1} \| \leq C$ from Assumption \ref{assump: many weak moments}. 

Next, from straightforward decomposition, we have
\begin{align*} 
A2& \leq \frac{C}{n\mu_n} \sum_{i,j=1} \tilde U_i^T  \{ g_j (\beta_0, \hat\bfeta)- g_j\}= \frac{C}{\sqrt{n}\mu_n} \left[\sum_{i=1}^n\tilde U_i \right]^T  \sqrt{n}\big\{\hat g(\beta_0, \hat\bfeta) - \hat g(\beta_0, \bfeta_0) \big\} \\
&\leq C  \left\| \frac{1}{\sqrt{n} \mu_n}  \sum_{i=1}^n\tilde U_i   \right\| \left\|  \sqrt{n}\big\{\hat g(\beta_0, \hat\bfeta) - \hat g(\beta_0, \bfeta_0) \big\}   \right\|
\end{align*}
The result follows from  $\frac{1}{\sqrt{n}\mu_n} \| \sum_{i=1}^n \tilde U_i \|=  O_p(1)$  shown in the proof of Lemma \ref{lemma: Q tilde} and Lemma \ref{lemma: g beta, nuisance}(i).

Moreover, we decompose $A3$ as 
\begin{align*}
	A3&= \frac{1}{n\mu_n} \sum_{i, j=1}^n \tilde{U}_i^T \Omega^{-1} g_j= \frac{1}{n\mu_n} \sum_{i, j=1}^n {U}_i^T \Omega^{-1} g_j + \frac{1}{n\mu_n} \sum_{i, j=1}^n (\tilde U_i - U_i)^T \Omega^{-1} g_j\\
	&= \frac{1}{n\mu_n} \sum_{i, j=1}^n {U}_i^T \Omega^{-1} g_j + o_p(1)
\end{align*}
where the last equality is from  
\begin{align*}
	&(n\mu_n)^{-1} \sum_{i,j=1}^n (\tilde U_i- U_i)^T \Omega^{-1} g_j \leq  C \left\|  n^{-1/2}\sum_{i=1}^n (\tilde U_i- U_i)\right\| \left\| \frac{1}{\mu_n\sqrt{n}} \sum_{j=1}^n g_j \right\| = o_p(1),
\end{align*}
 $\| n^{-1/2}\sum_{i=1}^n (\tilde U_i- U_i)\|=o_p(1)$ from  \eqref{eq: U_i}, and $\|\mu_n^{-1}n^{-1/2} \sum_{j=1}^n g_j \| =O_p(1)$ from \\ $E \|\mu_n^{-1}n^{-1/2} \sum_{j=1}^n g_j\|^2 =\mu_n^{-2} tr( \Omega )\leq C \mu_n^{-2} m \leq C$.

In conclusion, we have that $ n \mu_n^{-1} \partial \hat{Q} (\beta, \hat \bfeta)/\partial \beta \big|_{\beta= \beta_0}$ is asymptotically equivalent with $ n \mu_n^{-1} \partial \hat{Q} (\beta,  \bfeta_0)/\partial \beta \big|_{\beta= \beta_0}$. 
Finally, by Lemma \ref{lemma: Q second deriv}, we have
\[
n \mu_n^{-2}\frac{\partial^2 \hat Q (\beta, \hat\bfeta)}{\partial \beta^2} \bigg|_{\beta=\bar\beta} =n \mu_n^{-2} \frac{\partial^2 \hat Q (\beta, \bfeta_0)}{\partial \beta^2} \bigg|_{\beta=\bar\beta}  + o_p(1) .
\]
The asymptotic normality result follows from Theorem 3 in \cite{Newey:2009aa}.

\section{Proof of Theorem \ref{theo: overidentification}}

From the proof of Theorem 4 in \cite{Newey:2009aa}, we have
\begin{align*}
	\hat Q (\beta_0, \bfeta_0 ) &  =\hat{\bmg}(\bbeta_0, {\bfeta}_0 )^T \hat{\Omega} (\bbeta_0, {\bfeta}_0 )^{-1} \hat{\bmg}(\bbeta_0,{\bfeta}_0 )/2 \\
	& = \underbrace{\hat{\bmg}(\bbeta_0, {\bfeta}_0 )^T {\Omega}^{-1} \hat{\bmg}(\bbeta_0,{\bfeta}_0 )/2}_{	\tilde Q (\beta_0, \bfeta_0 )  } + o_p(\sqrt{m}/n)  ,
\end{align*}
and 
\begin{align*}
&	\frac{ n  \hat{\bmg}(\bbeta_0, {\bfeta}_0 )^T {\Omega}^{-1} \hat{\bmg}(\bbeta_0,{\bfeta}_0 ) - m  }{\sqrt{2m}}\\
& = \underbrace{	\frac{ \sum_{i=1}^n  {\bmg}_i (\bbeta_0, {\bfeta}_0 )^T {\Omega}^{-1} {\bmg}_i (\bbeta_0, {\bfeta}_0 )/n - m  }{\sqrt{2m}} }_{o_p(1)}+ 	\frac{ n^{-1} \sum_{i\neq j }  {\bmg}_i (\bbeta_0, {\bfeta}_0 )^T {\Omega}^{-1} {\bmg}_j (\bbeta_0, {\bfeta}_0 ) }{\sqrt{2m}}  \xrightarrow{d} N(0,1) .
\end{align*}
These results imply that 
\begin{align*}
	\frac{2n \hat Q (\beta_0, \bfeta_0 )  - m }{\sqrt{2m}} \xrightarrow{d} N(0,1) .
\end{align*}
By standard results that as $m\rightarrow\infty$, the $(1-\alpha)$th quantile $\chi_{1-\alpha}^2(m)$ of a $\chi^2(m)$ distribution has the property that $ \{ \chi_{1-\alpha}^2(m) - m \} / \sqrt{2m}$ converges to the $(1-\alpha)$th quantile of the standard normal distribution. Hence, 
\begin{align*}
	P\left( 2n \hat Q (\beta_0, \bfeta_0) \geq  \chi_{1-\alpha}^2(m) \right) = 	P\left( \frac{2n \hat Q (\beta_0, \bfeta_0) - m }{\sqrt{2m}} \geq \frac{ \chi_{1-\alpha}^2(m) - m}{\sqrt{2m}} \right) \to \alpha  .
\end{align*}

By a Taylor expansion and from $\partial \hat Q(\beta, \hat\bfeta)/\partial \beta |_{\beta=\hat\beta}=0$, for  $\bar\beta$ on the line joining $\hat\beta$ and $\beta_0$, we have 
\begin{align*}
	& 2n \{  \hat Q (\beta_0, \hat \bfeta )   -  \hat Q (\hat\beta, \hat \bfeta )   \}  \\
	& =  \mu_n^2 (\hat\beta - \beta_0 )^2 n \mu_n^{-2}\left\{  \frac{\partial^2  \hat Q ( \beta, \hat\bfeta)}{\partial \beta^2 } \mid_{\beta = \bar\beta}\right\}   \\ 
	&=  \mu_n^2 (\hat\beta - \beta_0 )^2 n \mu_n^{-2} \bigg\{  \frac{\partial^2  \hat Q ( \beta, \bfeta_0)}{\partial \beta^2 } \mid_{\beta = \bar\beta} +{\frac{\partial^2  \hat Q ( \beta, \hat \bfeta)}{\partial \beta^2 } \mid_{\beta = \bar\beta}  - \frac{\partial^2  \hat Q ( \beta,  \bfeta_0 )}{\partial \beta^2 } \mid_{\beta = \bar\beta} } \bigg\} \\
	&=   \mu_n^2 (\hat\beta - \beta_0 )^2  n \mu_n^{-2} \bigg\{  \frac{\partial^2  \hat Q ( \beta, \bfeta_0)}{\partial \beta^2 } \mid_{\beta = \bar\beta}  +o_p(1) \bigg\} \\
	&= \mu_n^2 (\hat\beta - \beta_0 )^2 \{  n \mu_n^{-2} G^T \Omega^{-1} G +o_p(1) \}  \\
	&= O_p(1) 
\end{align*}
where the fourth line is from Lemma \ref{lemma: Q second deriv}, and the fifth line is from Lemma A13 in \cite{Newey:2009aa}. Moreover, recall the definition that $\hat a(\beta, \hat\bfeta)= \mu_n^{-1} \sqrt{n} \Omega(\beta, \bfeta_0)^{-1}\hat g(\beta, \hat\bfeta) $. By Assumption \ref{assump:omega} and \eqref{eq: g hat bdd in prob}, 
\begin{align*}
	\|\hat a(\beta_0, \hat\bfeta)\|^2= \mu_n^{-2} n \hat g(\beta_0, \hat\bfeta)^T \Omega(\beta_0, \bfeta_0)^{-1}    \Omega(\beta_0, \bfeta_0)^{-1}  \hat g(\beta_0, \hat\bfeta) \leq C \mu_n^{-2} n \| \hat g(\beta_0, \hat\bfeta)\|^2. 
\end{align*}

Then recall the definition that 
\begin{align*}
	&  2n \hat Q(\beta_0, \hat\bfeta)=\mu_n^{2} \hat a(\beta_0, \hat\bfeta)^T \Omega \hat\Omega(\beta_0, \hat\bfeta)^{-1}\Omega \hat a(\beta_0, \hat\bfeta),\\
	& 2n \tilde Q(\beta_0, \hat\bfeta)=\mu_n^{2} \hat a(\beta_0, \hat\bfeta)^T \Omega \hat a(\beta_0, \hat\bfeta),
\end{align*}
we have 
\begin{align*}
	&2 n |\hat Q (\beta_0, \hat\bfeta)- \tilde{Q} (\beta_0, \hat\bfeta) | \\
	&\leq \mu_n^{2} |\hat a(\beta_0, \hat\bfeta)^T \{ \hat\Omega(\beta_0, \hat\bfeta) - \Omega\}\hat a(\beta_0, \hat\bfeta)| \\
	&\qquad +\mu_n^{2}  |\hat a(\beta_0, \hat\bfeta)^T \{ \hat\Omega(\beta_0, \hat\bfeta) - \Omega\} \hat\Omega(\beta_0, \hat\bfeta)^{-1} \{ \hat\Omega(\beta_0, \hat\bfeta) - \Omega\} \hat a(\beta_0, \hat\bfeta)  | \\
	&\leq \mu_n^{2}  \|\hat a(\beta_0, \hat\bfeta)\|^2 \left\{\| \hat\Omega(\beta_0, \hat\bfeta) - \Omega\| + C\| \hat\Omega(\beta_0, \hat\bfeta) - \Omega\|^2 \right\} \\
	&\leq  C n \| \hat g(\beta_0, \hat\bfeta)\|^2  \left\{\| \hat\Omega(\beta_0, \hat\bfeta) - \Omega\| + C\| \hat\Omega(\beta_0, \hat\bfeta) - \Omega\|^2 \right\}  \\
	&= o_p(m ) 
\end{align*}
where the last line is from $ \|\hat g(\beta_0, \hat\bfeta) \| = O_p(\sqrt{m}/\sqrt{n}) $ and Lemma \ref{lemma: Omega eta hat}(iv).

Next, we show that 
\[
2 n |\tilde Q (\beta_0, \hat\bfeta)- \tilde{Q} (\beta_0, \bfeta_0) | =O_p(1).
\]
This can be easily seen as  
\begin{align*}
	& 2 n |\tilde Q (\beta_0, \hat\bfeta)- \tilde{Q} (\beta_0, \bfeta_0) | \\
	&= n|\hat g(\beta_0, \hat\bfeta)^T \Omega^{-1} \hat g(\beta_0, \hat\bfeta) - \hat g(\beta_0, \bfeta_0)^T \Omega ^{-1}\hat g(\beta_0, \bfeta_0)| \\
	&=  n|\{\hat g(\beta_0, \hat\bfeta)-\hat g(\beta_0, \bfeta_0) \}^T \Omega ^{-1} \hat g(\beta_0, \hat\bfeta) +  \hat g(\beta_0, \bfeta_0)^T \Omega ^{-1} \{\hat g(\beta_0, \hat\bfeta) - \hat g(\beta_0, \bfeta_0)\} | \\
	&\leq  n|\{\hat g(\beta_0, \hat\bfeta)-\hat g(\beta_0, \bfeta_0) \}^T \Omega^{-1} \hat g(\beta_0, \hat\bfeta) |+ n | \hat g(\beta_0, \bfeta_0)^T \Omega ^{-1} \{\hat g(\beta_0, \hat\bfeta) - \hat g(\beta_0, \bfeta_0)\} | \\
	&\leq n C \|\hat g(\beta_0, \hat\bfeta)-\hat g(\beta_0, \bfeta_0) \| \|\hat g(\beta_0, \hat\bfeta) \|+ n C \| \hat g(\beta_0, \bfeta_0)\| \| \hat g(\beta_0, \hat\bfeta) - \hat g(\beta_0, \bfeta_0) \|\\
	& = o_p(\sqrt{m} ) 
\end{align*}
where the last line is from Lemma \ref{lemma: g beta, nuisance}(i),  $ \|\hat g(\beta_0, \hat\bfeta) \| = O_p(\sqrt{m}/\sqrt{n}) $ and $ \|\hat g (\beta_0, \bfeta_0)\|=O_p(\sqrt{m}/\sqrt{n})$, and $m^3/n\rightarrow 0$. 

Lastly, from $	\hat Q (\beta_0, \bfeta_0 ) =	\tilde Q (\beta_0, \bfeta_0 ) + o_p(\sqrt{m}/n) $ at the beginning of this section, we conclude that 
\begin{align}
&	\frac{2n  \{ \hat Q (\hat\beta,   \hat \bfeta )  -  \hat Q (\beta_0, \bfeta_0 ) \} }{\sqrt{2(m-1)}}  \nonumber\\
&=	\frac{2n  \{ \hat Q ( \hat \beta,   \hat \bfeta )  - \hat Q (\beta_0,   \hat \bfeta ) + \hat Q (\beta_0,   \hat \bfeta )- \tilde{Q} (\beta_0, \hat\bfeta)+ \tilde{Q} (\beta_0, \hat\bfeta) -  \tilde{Q} (\beta_0, \bfeta_0)  + \tilde{Q} (\beta_0, \bfeta_0)  - \hat Q (\beta_0, \bfeta_0 ) \} }{\sqrt{2(m-1)}} \nonumber\\
	&  = o_p(1) \label{eq: overidentification}
\end{align}

Therefore, 
\begin{align*}
	\frac{2n  \hat Q (\hat\beta, \hat  \bfeta)  - (m-1) }{\sqrt{2(m-1)}} &= \frac{2n  \hat Q (\beta_0, \bfeta_0 )  - (m-1) }{\sqrt{2(m-1)}} + o_p(1) \\
	&=  \frac{\sqrt{2m}}{\sqrt{2(m-1)}}\frac{2n  \hat Q (\beta_0, \bfeta_0 )  - m }{\sqrt{2m}}  + \frac{1}{\sqrt{2(m-1)}}+ o_p(1) \\
	&\xrightarrow{d} N(0,1) .
\end{align*}
Hence, it follows by the argument above that 
\begin{align*}
	P\left( 2n \hat Q ( \hat \beta, \hat  \bfeta) \geq  \chi_{1-\alpha}^2(m-1) \right) = 	P\left( \frac{2n \hat Q (\hat \beta, \hat  \bfeta) - (m-1) }{\sqrt{2(m-1)}} \geq \frac{ \chi_{1-\alpha}^2(m-1) - (m-1)}{\sqrt{2(m-1)}} \right) \to \alpha  .
\end{align*}

\section{Other Exposure and Outcome Types} \label{sec: binary}

Many MR applications consider binary outcomes \citep{Holmes:2014, Holmes:2017}. Binary exposure is not very common in MR studies \citep{Burgess:2018aa}, but are still of interest \citep{Nead:2015aa,Gage:2017aa, Larsson:2017, Vaucher:2018aa}. In this section, we extend the methods to consider binary exposure and/or binary outcome. We  focus on identification and leave formal treatment of inference under many weak invalid IVs to future work. 

In this section, when the exposure (or outcome) variable is continuous, we consider the linear model (i.e., the identity link function); when the exposure (or outcome) variable is binary, we consider the log-linear model (i.e., the exponential link function). Therefore, for continuous outcome with $ f_y(x)=x $, $ \beta_0 a = E[Y\mid 
A=a, U, \bmZ, \bmX]- E[Y\mid A=0, U, \bmZ, \bmX]$ encodes the treatment effect on the outcome mean upon increasing the exposure by $ a $ unit; for binary outcome with  $ f_y(x)=\exp(x) $, $ \beta_0 a = \log\{ P(Y=1\mid A=a, U, \bmZ, \bmX)\}- \log \{ P(Y=1\mid A=0, U, \bmZ, \bmX)\}$ encodes the log risk ratio. Other types of link function (e.g., logistic or probit) are not considered because of the noncollapsibility \citep{Baiocchi:2014aa} and the effect of $ U $ and $ A, \bmZ $ are not easily separable \citep{Clarke:2010aa}. 

Consider the following structural equations:
\begin{align}
	& E[Y\mid A,  U, \bmZ, \bmX]=f_y\left\{\beta_0A+ \bfalpha_0^T  \bmZ +\xi_y(U,\bmX)\right\}\label{eq: out model X general}\\
	& E[A\mid U, \bmZ, \bmX]=f_a \left\{ \bfgamma^T_0 \bmZ+ \xi_a(U, \bmX)\right\} \label{eq: exp model X general}
\end{align}
where $ \beta_0, \bfalpha_0, \bfgamma_0 $ are unknown true parameters, $  \xi_y, \xi_a$ are unspecified functions, $ f_y, f_a$ are pre-specified link functions, and $ \bmZ\perp U|\bmX $.

Next, we state our identification results. Let $ R_A= A- E(A\mid \bmZ, \bmX) $ when $ f_a(x)=x $; $ R_A= A\exp(- \bfgamma_0^T \bmZ) - E(  A\exp(- \bfgamma_0^T \bmZ)\mid \bmX )$ when $ f_a(x)=\exp(x) $.  Let $ r_Y= Y-\beta A $ when $ f_y(x)=x $; $ r_Y = Y\exp(-\beta A- \bfalpha^T \bmZ) $ when $ f_y(x)= \exp(x) $; $ r_{Y0} $  denote $ r_Y $ when $ \beta=\beta_0 $ and $ \bfalpha=\bfalpha_0 $. We use  $ \bfeta $ to denote the nuisance parameters which may be different for each scenario  and let $ \bfeta_0 $ be the true values.

\begin{prop} \label{prop1}
	Under \eqref{eq: out model X general}-\eqref{eq: exp model X general} and $ \bmZ\perp U\mid \bmX $, \\
	(a) When $ f_y(x)=x $, $ \beta_0 $ is the unique solution to $ E[ g(\O; \beta, \bfeta_0) ]=0 $, where 
	\begin{align}
		g(\O; \beta, \bfeta_0) = (\bmZ- E(\bmZ|\bmX)) R_A r_Y,  \label{eq: iden binary1}
	\end{align}
	provided that $ E[ (\bmZ- E(\bmZ|\bmX)) R_AA ] \neq \zero$.\\
	(b) When $ f_y(x)=\exp(x) $, $ \beta_0 $  and $\bfalpha_0 $ are identified from $ E[ g( \O; \beta, \bfalpha, \bfeta_0) ] =0 $, where 
	\begin{align}
		g(\O; \beta, \bfalpha, \bfeta_0)=\left[  \begin{array}{c}
			(\bmZ- E(\bmZ|\bmX)) R_A r_Y\\
			(\bmZ- E(\bmZ|\bmX)) r_Y
		\end{array} \right], \label{eq: iden binary2}
	\end{align}
	provided that $ E[\partial g(\O; \beta_0, \bfalpha_0, \bfeta_0) /\partial (\beta, \bfalpha)] $ is of rank $ m+1 $.
\end{prop}
Proposition \ref{prop1} provides identification formulas for $ \beta_0 $ for binary exposure and/or binary outcome.  The proof will show that $ E(R_Ar_{Y0}|\bmZ, \bmX) =E(R_Ar_{Y0}|\bmX)  $ holds almost surely in all cases and $ E(r_{Y0}|\bmZ, \bmX)=E(r_{Y0}| \bmX)  $ holds almost surely when $ f_y(x)=\exp(x) $. In all the cases, $ R_A $ is the residual in $ A $ after netting out the effect of $ \bmZ $; $ r_{Y0} $ is the residual in $ Y $ after netting out the effect of $ A $ when $ f_y(x)=x $, and is the residual in $ Y $ after netting out the effect of $ A $ and $ \bmZ $ when $ f_y(x)=\exp(x) $.  Note that  when $ f_a(x)=x $, the exposure model \eqref{eq: exp model X general} can be relaxed,  and   identification in  Proposition \ref{prop1} remains true as long as $ E(A\mid U, \bmZ, \bmX)$ can be expressed as $ \gamma_0(\bmZ, \bmX)+\xi_a(U, \bmX) $, where $\gamma_0, \xi_a $ are unspecified functions.  Similarly, when  $ f_y(x)=x $, the outcome model \eqref{eq: out model X general}  can be relaxed, and identification in Proposition \ref{prop1}  remains true as long as $ E(Y|U, \bmZ, \bmX) $ can be expressed as  $\beta_0A+  \alpha_0(\bmZ, \bmX)+\xi_y(U, \bmX)$, where $ \alpha_0, \xi_y$ are unspecified functions.

For binary exposure and continuous outcome (i.e., $ f_a(x)= \exp(x), f_y(x)=x $),   $\bfgamma_0 $ can be identified a prior from a separate set of estimation equations $ E[ (\bmZ-E(\bmZ|\bmX))R_A ] =\zero$.

In Proposition \ref{prop1}(b) where the outcome is binary (i.e., $ f_y(x)=\exp(x) $), $\bfalpha_0$ cannot be identified a prior and needs to be identified simultaneously with $\beta_0$, so that  the estimation equations for $ \beta_0 $ and $ \bfalpha_0 $ are stacked. 

\begin{proof}
(a) The case where $ f_a(x)= f_y(x)= x $ is already established in the main article.  Now, consider the case with binary exposure and continuous outcome, i.e., $ f_a(x)=\exp(x) $ and $ f_y(x)= x $. Then
\begin{align*}
	&E(R_Ar_{Y0} \mid \bmZ, \bmX) \\
	&= E\left[  \left\{ A\exp(-\bm\gamma_0^T\bmZ) - E( A\exp(-\bm \gamma_0^T \bmZ )\mid\bmX)  \right\} (Y-\beta_0 A)  \mid \bmZ, \bmX \right] \\
	&=E\left[  \left\{ A\exp(-\bm\gamma_0^T\bmZ) - E( A\exp(-\bm \gamma_0^T \bmZ )\mid\bmX)  \right\} (E(Y\mid A, U, \bmZ,\bmX)-\beta_0 A)  \mid \bmZ, \bmX \right] \\
	&= E\left[  \left\{ A\exp(-\bm\gamma_0^T\bmZ) - E( A\exp(-\bm \gamma_0^T \bmZ )\mid\bmX)  \right\} \left\{ \bm\alpha_0^T \bmZ+ \xi_y(U, \bmX)\right\}  \mid \bmZ, \bmX \right] \\
	&= E\left[  \left\{ \xi_a(U, \bmX)- E( \xi_a(U, \bmX) \mid\bmX)  \right\} \left\{ \bm\alpha_0^T \bmZ+ \xi_y(U, \bmX)\right\}  \mid \bmZ, \bmX \right] \\
	&= \text{cov} \left\{ \xi_a(U, \bmX), \bm\alpha_0^T \bmZ+ \xi_y(U, \bmX) \mid \bmZ, \bmX \right\} \\
	&=  \text{cov} \left\{ \xi_a(U, \bmX),  \xi_y(U, \bmX) \mid \bmZ, \bmX \right\} \\
	&=  \text{cov} \left\{ \xi_a(U, \bmX),  \xi_y(U, \bmX) \mid \bmX \right\} .
\end{align*}
Thus, 
\begin{align*}
&E\{ g(\bm O; \beta_0 ,\bfeta_0)\} =  E\left\{  (\bmZ- E(\bmZ\mid \bmX)) E(R_Ar_Y\mid \bmZ, \bmX) \right\} \\
&= E\left\{  (\bmZ- E(\bmZ\mid \bmX))  \text{cov} \left\{ \xi_a(U, \bmX),  \xi_y(U, \bmX) \mid \bmX \right\} \right\} = 0. 
\end{align*}

Next, note that
\begin{align*}
&E\{ g(\bm O; \beta ,\bfeta_0)\} - E\{ g(\bm O; \beta_0 ,\bfeta_0)\}  =(\beta_0-\beta) E\{ (\bmZ- E(\bmZ\mid \bmX) ) R_A A\}
\end{align*}
Therefore, $ \beta_0 $ is the unique solution to $E\{ g(\bm O; \beta ,\bfeta_0)\} =\zero $, provided that $ E\{ (\bmZ- E(\bmZ\mid \bmX) ) R_A A\}\neq \zero  $.

(b) Consider first the case with continuous exposure and binary outcome, i.e., $ f_a(x)=x $ and $ f_y(x)= \exp(x) $. Then 
\begin{align*}
	E(r_{Y0}\mid \bmZ, \bmX) = E \left\{\exp(\xi_y (U, \bmX)) \mid \bmZ, \bmX\right\} = E \left\{\exp(\xi_y (U, \bmX)) \mid \bmX\right\} ,
\end{align*}
and 
\begin{align*}
	&E\left\{  R_Ar_{Y0}\mid \bmZ, \bmX \right\}  = 	E\left\{  R_A E (r_{Y0} \mid A, U, \bmZ, \bmX) \mid \bmZ, \bmX \right\} \\
	& = E\left\{  R_A \exp(\xi_y(U, \bmX)) \mid \bmZ, \bmX \right\}   = \text{cov}\left\{ \bm\gamma_0^T \bmZ + \xi_a(U, \bmX),  \exp(\xi_y(U, \bmX)) \mid \bmZ, \bmX \right\}\\
	&=  \text{cov}\left\{ \xi_a(U, \bmX),  \exp(\xi_y(U, \bmX)) \mid \bmZ, \bmX \right\}  =  \text{cov}\left\{ \xi_a(U, \bmX),  \exp(\xi_y(U, \bmX)) \mid\bmX \right\} .
\end{align*}
These imply that $ E\{ g(\bm O; \beta_0, \bm\alpha_0, \bfeta_0)\}=0$. 

Identifiability also requires that $ E[\partial g(\O; \beta, \bfalpha, \bfeta_0) /\partial (\beta, \bfalpha)] $ is of rank $ m+1 $.

Finally, for binary exposure and binary outcome, we have that 
\begin{align*}
	&E\left\{  R_Ar_{Y0}\mid \bmZ, \bmX \right\}  = 	E\left\{  R_A E (r_{Y0} \mid A, U, \bmZ, \bmX) \mid \bmZ, \bmX \right\}  = E\left\{  R_A \exp(\xi_y(U, \bmX)) \mid \bmZ, \bmX \right\}   \\
	&= \text{cov}\left\{\xi_a(U, \bmX),  \exp(\xi_y(U, \bmX)) \mid \bmZ, \bmX \right\} =  \text{cov}\left\{ \xi_a(U, \bmX),  \exp(\xi_y(U, \bmX)) \mid\bmX \right\} .
\end{align*}
The rest of the proof follows the same step as the proof of continuous exposure and binary outcome.  
\end{proof}

\end{document}